\documentclass[3p,times]{elsarticle}
\makeatletter
\def\ps@headings{%
\def\@oddhead{\mbox{}\scriptsize\rightmark \hfil \thepage}%
\def\@evenhead{\scriptsize\thepage \hfil \leftmark\mbox{}}%
\def\@oddfoot{}%
\def\@evenfoot{}}
\makeatother
\usepackage{amsthm}
\usepackage{graphicx}
\usepackage{subcaption}
\usepackage{psfrag}
\usepackage{url}
\usepackage[ruled,vlined,linesnumbered,commentsnumbered]{algorithm2e}
\usepackage{array}
\usepackage{amsmath,amssymb}
\usepackage{color}
\usepackage{epstopdf}
\usepackage{bibspacing}
\usepackage{float}
\usepackage[font={small}]{caption}

\newtheorem{theorem}{Theorem}

\newtheorem{defn}[theorem]{Definition}

\begin{document}
\begin{frontmatter}
\title{On Hardening Problems in Critical Infrastructure Systems}
\author{Joydeep Banerjee, Kaustav Basu, Arunabha Sen}
\address{School of Computing, Informatics, \& Decision Systems Engineering}
\address{Arizona State University}
\address{Email: jbanerje@asu.edu, kbasu2@asu.edu, asen@asu.edu}

\begin{abstract}
The power and communication networks are highly interdependent and form a part of the critical infrastructure of a country. Similarly, dependencies exist within the networks itself. It is essential to have a model which captures these dependencies precisely. Previous research has proposed certain models but these models have certain limitations. The limitations of the aforementioned models have been overcome by the Implicative Interdependency Model, which uses Boolean Logic to denote the dependencies. This paper formulates the Entity Hardening problem and the Targeted Entity Hardening problem using the Implicative Interdependency Model. The Entity Hardening problem describes a situation where an operator, with a limited budget, must decide which entities to harden, which in turn would minimize the damage, provided a set of entities fail initially. The Targeted Entity Hardening problem is a restricted version of the Entity Hardening problem. This problem presents a scenario where, the protection of certain entities is of higher priority. If these entities were to be nonfunctional, the economic and societal damage would be higher when compared to other entities being nonfunctional. It has been shown that both problems are NP-Complete. An Integer Linear Program (ILP) has been devised to find the optimal solution. A heuristic has been described whose accuracy is found by comparing its performance with the optimal solution using real-world and simulated data.

\end{abstract}


\end{frontmatter}

\section{Introduction}
\label{Intro}
Critical Infrastructures of a nation like Power, Communication, Transportation Networks etc. exhibit strong intra-network and inter-network dependencies to drive their functionality. The symbiotic relationship that exists between Power and Communication Network provides an example of the inter-network dependency. To elaborate this further, consider entities in either network. The Supervisory Control and Data Acquisition System (SCADA) is an integral entity in the power network controlling the electricity generation and power flow in the power grid. These controls are essentially carried out by signals generated from the communication network. Similarly, every entity in the communication network requires power to be operational. These dependencies cause failure in any of these two networks to have its impact on the other which may eventually lead to a cascade of failures. 

Additionally, intra-network dependencies exist as well in a critical infrastructure. In an abstract level, a power network is composed of the following entities --- Generation Bus, Load Bus, Neutral Bus and Transmission Lines. When a transmission line trips, the power flowing through the transmission lines needs to be redirected to satisfy load demand of the load buses. This may cause the power flow in some other transmission line to go beyond its line capacity causing it to trip. Eventually, these failures might result in a cascade of trippings/failures resulting in a blackout. Cascading failures in power and/or communication network due to intra/inter dependencies have disastrous effects as seen in power blackouts which occurred in New York (2003) ~\cite{andersson2005causes}, San Diego (2011) ~\cite{sandiego} and India (2012) ~\cite{tang2012analysis}. Thus modeling these dependencies is critical in understanding and preventing such failures which might be triggered by natural as well as man-made attacks.

As noted, modeling these complex dependencies and analysis of failure in these infrastructures are considered to be highly important. A number of models have been proposed that capture these kind of dependencies \cite{Bul10}, \cite{Gao11}, \cite{Sha11}, \cite{Ros08}, \cite{Zha05}, \cite{Par13}, \cite{Ngu13}, \cite{Zus11}. However, each of these models has their own shortcomings in bringing out the complex nature of the dependencies that might exist \cite{Banerjee2014Survey}. Authors in \cite{sen2014identification} bring out the need to address the complex dependency which can be explained through the following example. Let $a_x$ (which can be a generator, substation, transmission line etc.) be a power network entity and $b_w,b_y,b_z$ (which can be a router, end system etc.) a set of communication network entities. Consider the dependency where the entity $a_x$ is operational if (i) entities $b_w$ and (\emph{logical AND}) $b_y$ are operational, or (\emph{logical OR}) (ii) entity $b_z$ is operational. Models in \cite{Bul10}, \cite{Gao11}, \cite{Sha11}, \cite{Ros08}, \cite{Zha05}, \cite{Par13}, \cite{Ngu13}, \cite{Zus11} fails to capture this kind of dependency. Motivated by these findings and limitations of the existing models, the authors in \cite{sen2014identification} proposed a \emph{Boolean logic based dependency} model termed as \emph{Implicative Interdependency Model} (IIM). For the example stated above, the dependency of $a_x$ on $b_w,b_y,b_z$ can be represented as $a_x \leftarrow b_w b_y + b_z$. This equation representing the dependency of an entity is termed as \textit{Interdependency Relation} (IDR). Using this model a number of problems were studied on interdependent power and communication infrastructure system\cite{sen2014identification}, \cite{das2014root}, and \cite{das2015smallest}. 

In this paper, we restrict to intra-network dependencies in Power Network and inter-network dependency in Power-Communication network and utilize the IIM to analyze and solve two important problems pertaining to critical infrastructures. For an existing critical infrastructure system, an operator would have the capability to measure the extent of failure when a certain set of entities fail initially. Consider a scenario where an operator identifies a set of critical entities which when failed initially would cause the maximum damage. In an ideal case, there would be enough resources available to support those critical entities from initial failure. However, if the availability of resources is a constraint, then an operator might have to choose entities which when supported would minimize the damage. We define the entities to support as the entities to harden and the problem as the \emph{Entity Hardening Problem}. An entity $x_i$ when hardened is resistant to both \emph{initial} and \emph{induced} failure (failing of entities in the cascading process after the initial failure). In the physical world, an entity can be hardened with respect to cyber attacks (say) by having a strong firewall. Similarly some entities can be hardened by  --- (a) strengthening their physical structures for protection from natural disaster, (b) placing redundant entity $a'$ for an entity $a$ which can operate when $a$ fails, (c) increasing physical limits of the entity (maximum power flow capacity of the transmission line, maximum generation capacity of a generator bus). There exist multiple such ways to harden an entity from different kind of failures. Even though there may be circumstances under which an entity cannot be hardened, in this paper we relax such possibilities and assume that there always exist a way to harden a given entity. Hardening entities can prevent cascading failures caused by some initial failure. Thus this results in protecting a set of entities including the hardened entities from an initial failure trigger. Using these definitions the \emph{Entity Hardening Problem} finds a set of $k$ entities that should be hardened (with $k$ being the resource constraint) in an intra-network or inter-network critical infrastructure system that protects the maximum number of entities from failure when a set of $K$ entities fail initially. 

The second problem, \emph{Targeted Entity Hardening}, discussed in this article is a restricted version of the \emph{Entity Hardening Problem}. For an intra-dependent power network or interdependent power and communication network, certain entities might have higher priority to be protected. There might exist entities whose non-functionality poses higher economic or societal damage as compared to other entities. For example, power and communication network entities corresponding to office buildings running global stock exchanges, the U.S. White House, transportation sectors like airports etc. presumably are more important to be protected. Let $F$ denote the failed set of entities (including initial and induced failure) when a set of $K$ entities fail initially. We define a set $P$ (with $P \subseteq F$) of entities which have a higher priority to be protected. The \emph{Targeted Entity Hardening} problem finds the minimum set of entities which when hardened would ensure that none of the entities in set $P$ fail. 

This paper is more inclined towards finding and analyzing the solution of the two problems discussed. Even though procedures are described to generate dependency equations, they are primarily intended to perform a comparative analysis of the provided solutions to the problems. For an intra-dependent or inter-dependent critical infrastructure(s) which can be modeled through IIM, the broader idea is to use the solutions for different decision-making tasks. The paper is structured as follows. The motivation behind IIM along with a formal description is provided in Section \ref{IIM_Section}. The two hardening problems are more formally defined along with their Decision and Optimization Versions in Section \ref{ProbForm}. The computational complexity of the problems along with the solutions to some restricted cases is provided in Section \ref{CompAna}. As both the problems are NP-complete, we provide an optimal Integer Linear Program (ILP) solutions to them along with sub-optimal Heuristics in Section \ref{Solutions}. In Section \ref{ExpRes} we describe a procedure to generate the dependency equations of the IIM model for intra-dependent power network along with a rule defined approach to generate the same for interdependent power-communication network. For power network, different bus system data are used to generate the dependency equations which are obtained from MatPower \cite{zimmerman2011matpower}. For interdependent Power-Communication network we used real world data of Maricopa County, Arizona, USA obtained from Geotel (\emph{http://www.geo-tel.com}) for communication network and Platts (\emph{http://www.platts.com}) for power network. In the same section, we provide comparative studies of the heuristic to optimal ILP solutions for both the problem using the generated dependency equations. 

\section{Related Work}
In the last few years, there has been considerable activity in the research community to study Critical Infrastructure Interdependency. One of the earliest studies on robustness and resiliency issues related to Critical Infrastructures of the U.S. was conducted by the Presidential Commission on Critical Infrastructures, appointed by President Clinton in 1996 \cite{clintonexecutiveorder1996}. Rinaldi et al. are among the first group of researchers to study interdependency between Critical Infrastructures and to propose the use of complex adaptive systems as models of critical infrastructure interdependencies \cite{rinaldi2004modeling}, \cite{rinaldi2001identifying}. Pederson et al. in \cite{pederson2006critical}, provided a survey of Critical Infrastructure Interdependency modeling, undertaken by U.S. and international researchers. Motivated by the power failure event in Italy 2003, Buldyrev et al. in \cite{Bul10}, proposed a graph-based interdependency model, where the number of nodes in the power network was assumed to be the same as the number of nodes in the communication network, and in addition there existed a one-to-one dependency between a node in the power network to a node in the communication network. The authors opine in a subsequent paper \cite{Sha11} that the assumption regarding one-to-one dependency relationship is unrealistic and a single node in one network may be dependent on multiple nodes in the other network. Lin et al. presented an event driven co-simulation framework for interconnected power and communication networks in \cite{lin2011power}, \cite{lin2012geco}. A game theoretic model for a multi-layer infrastructure networks using flow equilibrium was proposed in \cite{Zha05}. Security of interdependent and identical Networked Control System (NCS) was studied in \cite{amin2013security}, where each plant was modeled by a discrete-time stochastic linear system, with systems controlled over a shared communication network. The importance of simultaneously considering power and communication infrastructures was highlighted in \cite{muller2016interfacing}. The results of a systematic study of human initiated cascading failures in critical interdependent societal infrastructures were reported in \cite{barrett2012human}. Focusing on the blackout of the Polish power grid, the authors in \cite{pfitzner2011controlled} studied the impact of the order of tripping of overhead lines on the severity of the failure. Analyzing failure in smart grid under targeted initial attack was studied in \cite{ruj2014analyzing}. The effect of cyber (communication) and power network dependencies in smart grid was studied in \cite{falahati2012reliability} for reliability assessments. Recovery of information of the failed entities in a power grid after a failure event was studied in \cite{soltan2015joint}. As described in Section \ref{Intro}, the models used by each of the papers have shortcomings to analyze different aspects of vulnerability in critical infrastructures. IIM is used to overcome such limitations and is utilized to address the Entity Hardening and Targeted Entity Hardening problem in this paper. 
\section{Implicative Interdependency Model}
\label{IIM_Section}
The need for a model to capture the complex intra and inter network dependencies is elaborated through a descriptive example of interdependent power and communication network. Consider the system shown in Figure \ref{fig:dependentPowerCom} where the power network entities such as generators, transmission lines and substations are denoted by $a_0$ through $a_{11}$ and communication entities such as GPS transmitters and satellites are denoted by $b_0$ through $b_4$. The Smart Control Center (SCC) is represented by the variable $c_0$ as it is a part of both the power and the communication network.  For the SCC to be operational, it must receive electricity either from the generator via the different power grid entities, or from the battery. Similarly, the functioning of the generator will be affected if it fails to receive appropriate control signals from the SCC. The mutual dependency between the generator and the SCC can be expressed in terms of two implicative dependency relations --- (i) $a_{11} \leftarrow b_4 c_0$, (ii) $c_0 \leftarrow (b_0 b_3 (b_1 + b_2)) (a_0 a_1 + a_2 a_3 a_4 a_5 a_6 a_7 a_8 a_9 a_{10} a_{11})$. It may be noted that the SCC will not be operational if it does not receive electric power produced at the generating station and carried over the power grid entities to the SCC and its battery backup also fails. This dependency can be expressed by the implicative relation $c_0 \leftarrow a_0 a_1 + a_2 a_3 a_4 a_5 a_6 a_7 a_8 a_9 a_{10} a_{11}$  implying that $c_0$ will be operational (i) if entities $a_0$ and $a_1$ are operational, or (ii) if entities $a_2$ through $a_{11}$ are operational. However, the SCC will also not be operational if it does not receive data from the communication system (IEDs, satellites, etc.). This dependency can be expressed by the relation $c_0 \leftarrow (b_0 b_3 (b_1 + b_2))$. This implies that $c_0$ will be operational (i) if entities $b_1$ or $b_2$ is operational, and (ii) if entities $b_0$ and $b_3$ are operational. Combining the dependency of the SCC on the power grid and the communication network, the consolidated dependency relation can be expressed as $c_0 \leftarrow (b_0 b_3 (b_1 + b_2)) (a_0 a_1 + a_2 a_3 a_4 a_5 a_6 a_7 a_8 a_9 a_{10} a_{11})$. Likewise, the dependency relation for the generating station can be expressed as $a_{11} \leftarrow b_4 c_0$, implying that the generating station will not be operational unless it receives appropriate signals from the SCC $c_0$, carried over wired or wireless link $b_4$. These two implicative relations demonstrate that dependency (or interdependency) is a complex combination of conjunctive and disjunctive terms. We term the model capturing this complex dependencies and interdependencies as \textit{Implicative Interdependency Model}.

\begin{figure}[ht]
    \centering
    \includegraphics[width=8cm,height=8cm]{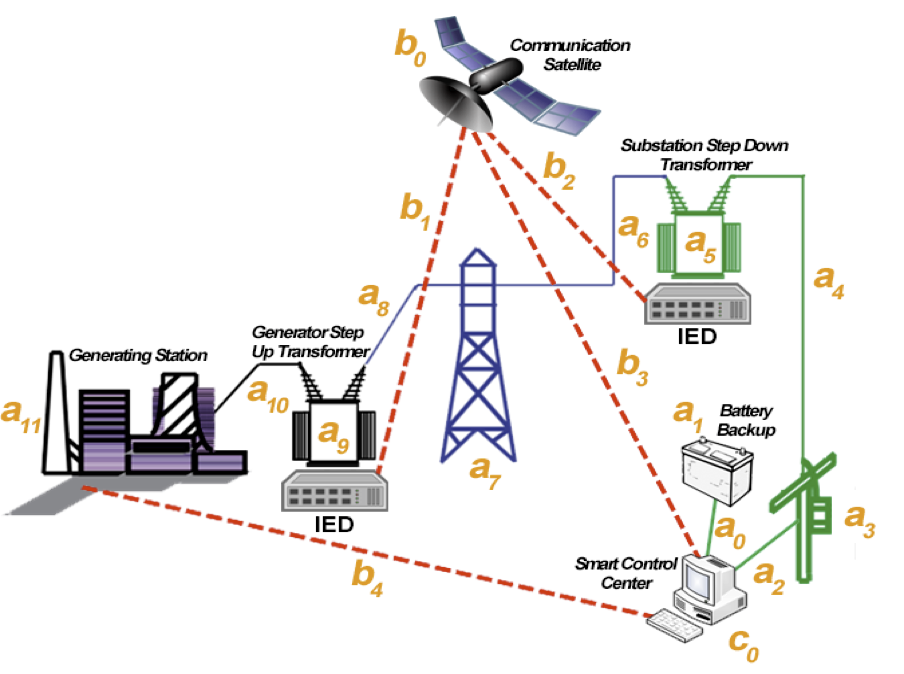}
    \caption{Example of Power - Communication Infrastructure Interdependency}
    \label{fig:dependentPowerCom}
\end{figure}

In the IIM an intra-network or inter-network critical infrastructure system is represented by $\mathcal{I}(E,\mathcal{F}(E))$, where $E$ is the set of entities and $\mathcal{F}(E)$ is the set of dependency relations. Throughout this paper, an intra-dependent critical infrastructure or interdependent critical infrastructure is termed as \emph{system} denoted by $\mathcal{I}(E,\mathcal{F}(E))$. The dynamics of the model is explained through an example. Consider sets $A$ and $B$ (with $E=A \cup B$) representing entities in power and communication network (say) with $A = \{a_1,a_2,a_3\}$ and $B = \{b_1,b_2,b_3,b_4\}$ respectively. The function  $\mathcal{F}(E)$ giving the set of dependency equations are provided in Table \ref{tbl:example1idr}.  In the given example, an IDR $b_3 \leftarrow a_2 + a_1 a_3$ implies that entity $b_3$ is operational if entity $a_2$ \emph{or} entity $a_1$ \emph{and} $a_3$ are operational. In the IDRs each conjunction term e.g. $a_1 a_3$ is referred to as \emph{minterms}. 
\begin{table}[ht!]
\parbox{.4\linewidth}{
		\centering
		\begin{tabular}{|l||l|}  \hline
			{\bf Power Network} & {\bf Comm. Network} \\ \hline
			$a_1\leftarrow b_2 $ & $b_1 \leftarrow a_1 + a_2$ \\ \hline
			$a_2 \leftarrow  b_2$ & $b_2 \leftarrow a_1a_2$ \\ \hline
			$a_3 \leftarrow b_4$ & $b_3 \leftarrow a_2 + a_1 a_3$ \\ \hline
			$--$ & $b_4 \leftarrow a_3$ \\ \hline
		\end{tabular}
		\caption{IDRs for the constructed example}
		\protect\label{tbl:example1idr}
	}
	\hfill
\parbox{.4\linewidth}{
		\centering
			\begin{tabular}{|c|c|c|c|c|c|c|c|}  \hline
			\multicolumn{1}{|c|}{Entities} & \multicolumn{7}{c|}{Time Steps ($t$)}\\
			\cline{2-8} & $0$ & $1$ & $2$ & $3$ & $4$ & $5$ & $6$ \\\hline \hline
			$a_1$ & $0$ & $0$ & $1$ & $1$ & $1$ & $1$ & $1$ \\ \hline
			$a_2$ & $1$ & $1$ & $1$ & $1$ & $1$ & $1$ & $1$ \\ \hline
			$a_3$ & $1$ & $1$ & $1$ & $1$ & $1$ & $1$ & $1$ \\ \hline
			$b_1$ & $0$ & $0$ & $0$ & $1$ & $1$ & $1$ & $1$ \\ \hline
			$b_2$ & $0$ & $1$ & $1$ & $1$ & $1$ & $1$ & $1$ \\ \hline
			$b_3$ & $0$ & $1$ & $1$ & $1$ & $1$ & $1$ & $1$ \\ \hline
			$b_4$ & $0$ & $1$ & $1$ & $1$ & $1$ & $1$ & $1$ \\ \hline
			%{\em Optimal} & $a_0$ & $a_1$ &  $\cdots$ & $\cdots$ &  $a_k$ \\ \hline
			%{\em Greedy } & $b_0$ & $b_1$ & $\cdots$ & $\cdots$ &  $b_k$ \\ \hline
		\end{tabular}
		\caption{Failure cascade propagation when entities $\{a_2, a_3\}$ fail at time step $t=0$. A value of $1$ denotes entity failure, and $0$ otherwise}
		\protect\label{tbl:example1cascade}
	}
\end{table}

Initial failure of entities in $A \cup B$ would cause the failure to cascade until a steady state is reached. As noted earlier, the event of an entity failing after the initial failure is termed as \textit{induced failure}. Failure in IIM proceeds in unit time steps with \textit{initial failure} starting at time step $t=0$. Each time step captures the effect of entities killed in all previous time steps. We demonstrate the cascading failure for the interdependent network outlined in Table \ref{tbl:example1idr} through an example. Consider the entities $a_2$ and $a_3$ fail at time step $t=0$. Table \ref{tbl:example1cascade} represents the cascade of failure in each subsequent time steps. In Table \ref{tbl:example1cascade}, for a given entity and time step, $'0'$ represents the entity is operational and $'1'$ non operational. In this example a steady state is reached at time step $t=3$ when all entities are non operational. IIM also assumes that the dependent entities of all failed entities are killed immediately at the next time step. For example at time step $t=1$ entities $a_2$, $a_3$, $b_2$, $b_3$ and $b_4$ are non operational. Due to the IDR $a_1 \leftarrow b_2$ entity $a_1$ is killed immediately at time step $t=2$. At $t=3$ the entity $b_1$ is killed due to the IDR $b_1 \leftarrow a_1 + a_2$ thus reaching the steady state. 

As noted earlier The model captures the cascading failure that propagates through the entities on an event of \textit{initial failure}.  Consider $E= A \cup B$ with $A$ and $B$ representing entities in two separate critical infrastructures. The cascading failure process is shown diagrammatically in Figure \ref{fig:cascade} with sets $A_{d}^0 \subset A$ and  $B_{d}^0 \subset B$ failing at $t=0$. Accordingly, cascading failure in these systems can be represented as a \textit{closed loop} control system shown in Figure \ref{fig:fixedPoint}. The steady state after an initial failure is analogous to the computation of \textit{fixed point} of a function  ${\mathcal G}(.)$ such that ${\mathcal G} (A_{d}^p \cup B_{d}^p) = A_{d}^p \cup B_{d}^p$, with steady state reached at $t=p$. It can be followed directly that for an interdependent system with $|E| = m$, any initial failure would cause the system to reach a steady state within $m-1$ time steps. 

\begin{figure}[ht]
\centering
\begin{subfigure}{.5\textwidth}
  \centering
  \includegraphics[width=8cm,height=8cm,keepaspectratio]{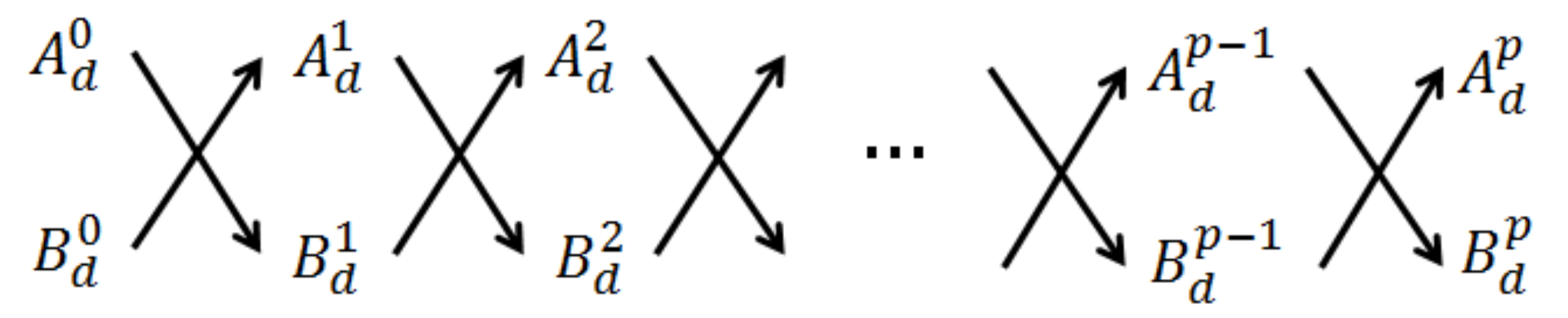}
  \caption{Cascading failures reach steady state after $p$ time steps}
  \label{fig:cascade}
\end{subfigure}%
\begin{subfigure}{.5\textwidth}
  \centering
  \includegraphics[width=8cm,height=8cm,keepaspectratio]{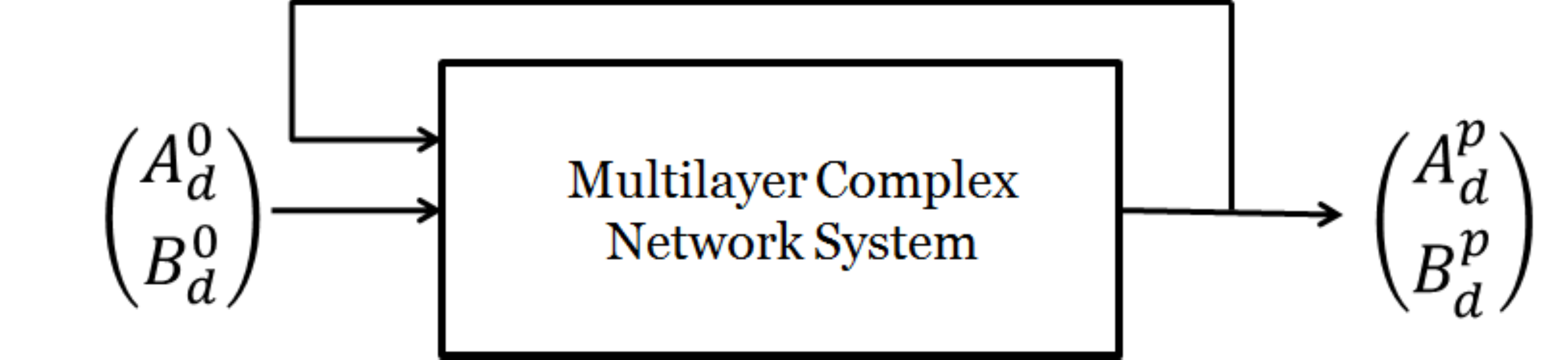}
  \caption{Cascading failures as a fixed point system}
  \label{fig:fixedPoint}
\end{subfigure}
\caption{Cascading Failures in Multi-layered Networks}
\label{fig:cascadeFixedPoint}
\end{figure}

We note some of the challenges in generating the IDRs. The main challenge is an accurate formulation of the IDRs. Two possible ways of doing this would be (i) careful analysis of the underlying infrastructures as in \cite{Zus11}, (ii) consultation with domain experts. In Section \ref{ExpRes} we provide techniques to generate IDR for power network and interdependent power-communication network. However, the underlying assumptions behind the technique pose some limitations on its applicability to the real world problems. The formulation of IDRs from the interdependent network is an ongoing research and the problem is solved under the reasonable assumption that these IDRs can be developed. 

\section{Problem Formulation}
\label{ProbForm}
As discussed in Section \ref{Intro}, an entity when hardened, is protected from both initial and induced failures. With this understanding we formally describe the two hardening problems --- Entity Hardening and Targeted Entity Hardening. It is to be noted both the problems return a set of entities to harden. The approach that should be taken to physically harden the entities rest upon the properties of the entities. 
\subsection{Entity Hardening Problem}
Before stating the problem formally, a brief understanding of entity hardening is provided. Consider the system with set of dependency relations given by Table \ref{tbl:example1idr}. With an initial failure of entities $a_2, a_3$ the subsequent cascading failures is shown in Table \ref{tbl:example1cascade} which fails all the entities in the system. We note three instances where entities  $a_1, a_2$ and $a_3$ are hardened separately with $a_2, a_3$ failing initially. The failure cascade propagation when $a_1, a_2$ and $a_3$ are hardened are shown in Tables \ref{tbl:example1cascadeA}, \ref{tbl:example1cascadeB}, and \ref{tbl:example1cascadeC} respectively. In the tables the cascading failure is shown till $t=3$ because with initial failure of entities $a_2,a_3$ the cascade propagation stops at $t=3$ as seen in Table \ref{tbl:example1cascade}. Hardening entity $a_1$ protect entities $a_1, b_1$ from failure. Similarly, when $a_2$ is hardened it protect $a_1, a_2, b_1, b_2, b_3$ and hardening $a_3$ protect entities $a_3, b_4$. If the hardening budget is $1$ the operator would clearly harden the entity $a_2$ as it protects the maximum number of entities from failure. We now describe the entity hardening problem formally.

\begin{table}[ht!]
\parbox{.3\linewidth}{
		\centering
		\begin{tabular}{|c|c|c|c|c|}  \hline
			\multicolumn{1}{|c|}{Entities} & \multicolumn{4}{c|}{Time Steps ($t$)}\\
			\cline{2-5} & $0$ & $1$ & $2$ & $3$ \\\hline \hline
			$a_1$ & $0$ & $0$ & $0$ & $0$ \\ \hline
			$a_2$ & $1$ & $1$ & $1$ & $1$ \\ \hline
			$a_3$ & $1$ & $1$ & $1$ & $1$ \\ \hline
			$b_1$ & $0$ & $0$ & $0$ & $0$ \\ \hline
			$b_2$ & $0$ & $1$ & $1$ & $1$ \\ \hline
			$b_3$ & $0$ & $1$ & $1$ & $1$ \\ \hline
			$b_4$ & $0$ & $1$ & $1$ & $1$ \\ \hline
			%{\em Optimal} & $a_0$ & $a_1$ &  $\cdots$ & $\cdots$ &  $a_k$ \\ \hline
			%{\em Greedy } & $b_0$ & $b_1$ & $\cdots$ & $\cdots$ &  $b_k$ \\ \hline
		\end{tabular}
		\caption{Failure cascade propagation when entities $\{a_2, a_3\}$ fail at time step $t=0$ and $a_1$ is hardened}
		\protect\label{tbl:example1cascadeA}
	}
	\hfill
\parbox{.3\linewidth}{
		\centering
		\begin{tabular}{|c|c|c|c|c|}  \hline
			\multicolumn{1}{|c|}{Entities} & \multicolumn{4}{c|}{Time Steps ($t$)}\\
			\cline{2-5} & $0$ & $1$ & $2$ & $3$ \\\hline \hline
			$a_1$ & $0$ & $0$ & $0$ & $0$ \\ \hline
			$a_2$ & $0$ & $0$ & $0$ & $0$ \\ \hline
			$a_3$ & $1$ & $1$ & $1$ & $1$ \\ \hline
			$b_1$ & $0$ & $0$ & $0$ & $0$ \\ \hline
			$b_2$ & $0$ & $0$ & $0$ & $0$ \\ \hline
			$b_3$ & $0$ & $0$ & $0$ & $0$ \\ \hline
			$b_4$ & $0$ & $1$ & $1$ & $1$ \\ \hline
			%{\em Optimal} & $a_0$ & $a_1$ &  $\cdots$ & $\cdots$ &  $a_k$ \\ \hline
			%{\em Greedy } & $b_0$ & $b_1$ & $\cdots$ & $\cdots$ &  $b_k$ \\ \hline
		\end{tabular}
		\caption{Failure cascade propagation when entities $\{a_2, a_3\}$ fail at time step $t=0$ and $a_2$ is hardened}
		\protect\label{tbl:example1cascadeB}
	}
	\hfill
	\parbox{.3\linewidth}{
	\centering
		\begin{tabular}{|c|c|c|c|c|}  \hline
			\multicolumn{1}{|c|}{Entities} & \multicolumn{4}{c|}{Time Steps ($t$)}\\
			\cline{2-5} & $0$ & $1$ & $2$ & $3$ \\\hline \hline
			$a_1$ & $0$ & $0$ & $1$ & $1$ \\ \hline
			$a_2$ & $1$ & $1$ & $1$ & $1$ \\ \hline
			$a_3$ & $0$ & $0$ & $0$ & $0$ \\ \hline
			$b_1$ & $0$ & $0$ & $0$ & $1$ \\ \hline
			$b_2$ & $0$ & $1$ & $1$ & $1$ \\ \hline
			$b_3$ & $0$ & $1$ & $1$ & $1$ \\ \hline
			$b_4$ & $0$ & $0$ & $0$ & $0$ \\ \hline
			%{\em Optimal} & $a_0$ & $a_1$ &  $\cdots$ & $\cdots$ &  $a_k$ \\ \hline
			%{\em Greedy } & $b_0$ & $b_1$ & $\cdots$ & $\cdots$ &  $b_k$ \\ \hline
		\end{tabular}
		\caption{Failure cascade propagation when entities $\{a_2, a_3\}$ fail at time step $t=0$ and $a_3$ is hardened}
		\protect\label{tbl:example1cascadeC}
	}
\end{table}

\noindent
\textbf{{\em The Entity Hardening (ENH) problem}}\\
INSTANCE: Given:\\
(i) A system $\mathcal{I}(E,\mathcal{F}(E))$, where the set $E$ represent the set of entities, and $\mathcal{F}(E)$ the set of IDRs.\\
(ii) The set of $K$ initially failing entities $E'$, where $E' \subseteq E$ \\
(iii) Two positive integers $k, k<K$ and $E_F$. \\ \\
DECISION VERSION: Is there a set of entities $\mathcal{H} = E'', E'' \subseteq E, |\mathcal{H}| \leq k$, such that hardening $\mathcal{H}$ entities results in no more than $E_F$ entities to fail after entities in $E'$ fail at time step $t=0$. \\ \\
OPTIMIZATION VERSION: Find a set of $k$ entities to harden which would maximize the number of protected entities with entities in $E'$ failing initially.

\begin{defn}
\label{Kill}
\textit{$KillSet(S)$ : For an initial failure of set $S$, the set of entities that fail due to induced failure in the cascading process including the entities in set $S$ is denoted by $KillSet(S)$}.
\end{defn}

The following points are to be noted regarding the ENH problem --- (a) the condition $k<K$ is assumed as with $k \ge K$ hardening the $K$ initially failing entities would ensure that there are no induced and initial failure. (b) with $E'$ entities failing initially, the entities to be harden are to be selected from $KillSet(E')$. Hardening entities outside $KillSet(E')$ would not result in protection of any non-hardened entity. 

\subsection{Targeted Entity Hardening Problem}
Qualitatively, for a system $\mathcal{I}(E,\mathcal{F}(E))$ the objective of the Targeted Entity Hardening problem is to choose a minimum cardinality set of entities to harden, with a set of initially failing entities, such that all entities in a given set $P$ are protected from failure. We use the example with dependency equations outlined in Table \ref{tbl:example1idr} to describe the Targeted Entity Hardening Problem with $P = \{b_4\}$. With $\{a_2, a_3 \}$ being the two entities failing initially, hardening entity $a_2$ (with $a_3$ failing) would prevent failure of entities $a_1,a_3, b_1, b_1, b_3$. Similarly, hardening the entity $a_3$ (with $a_2$ failing) would prevent the failure of entity $b_4$. Even though hardening $a_2$ prevent failure of more entities than hardening $a_3$, owing to the problem description $a_3$ has to be hardened which is a solution to the Targeted Entity Hardening problem in this scenario. It is to be noted that other entities might also be protected from failure when a set of entities are hardened to protect a given set of entities. The Targeted Entity Hardening problem is formally stated below accompanied with a descriptive diagram provided in Figure \ref{fig:figTargetHard} (in the figure direct failure means initial failure) --- \\ \\
\textbf{{\em The Targeted Entity Hardening (TEH) problem}}\\
INSTANCE: Given:\\
(i) A system $\mathcal{I}(E,\mathcal{F}(E))$, where the set $E$ represent the set of entities, and $\mathcal{F}(E)$ is the set of IDRs.\\
(ii) The set of $K$ entities failing initially $E'$, where $E' \subseteq E$. \\
(iii) The set $F \subseteq E$ contains all the entities failed due to initial failure of $E'$ entities i.e. $KillSet(E')$ \\
(iv) A positive integer $k$ and $k < K$. \\ 
(v) A set $P \subseteq F$. \\ \\
DECISION VERSION: Is there a set of entities $H = E'' \subseteq E, |H| \leq k$, such that hardening $H$ entities would result in protecting all entities in the set $P$ after entities in $E'$ fail at the initial time step. \\ \\
OPTIMIZATION VERSION: Find the minimum set of entities in $E$ to harden that would result in protecting all entities in the set $P$ after entities in $E'$ fail at the initial time step. 

\begin{figure}[h]
    \centering
    \includegraphics[width=7cm,height=7cm,keepaspectratio]{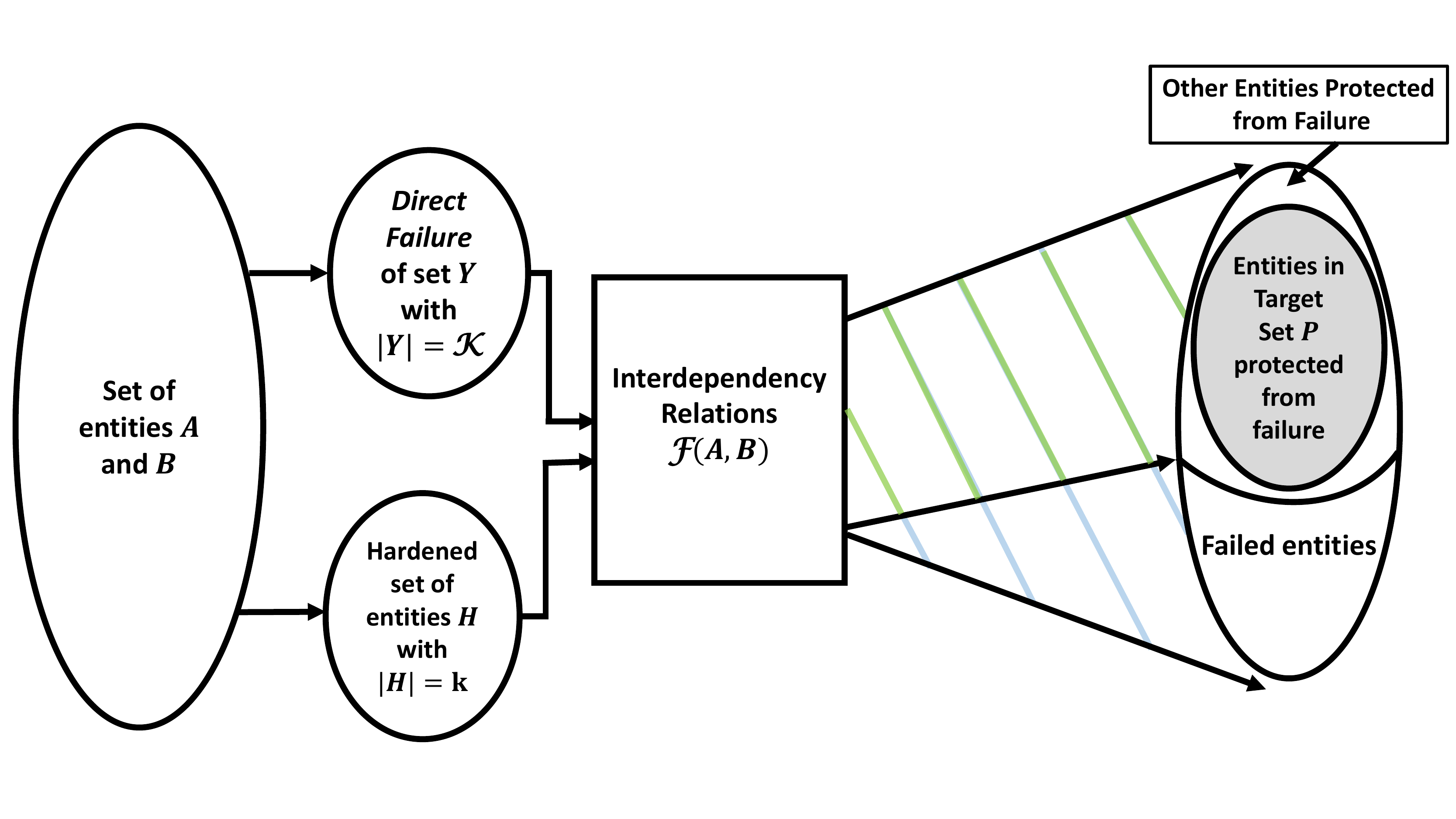}
    \caption{Pictographic description of the Targeted Entity Hardening problem}
    \label{fig:figTargetHard}
\end{figure}

\section{Computational Complexity Analysis}
\label{CompAna}
The computational complexity for both the problems are provided in this section. The problems are proved to be NP-complete. Additionally, approximate and polynomial solutions to few subcases are provided. The subcases impose restrictions on the IDRs and the solutions can be applied to interdependent systems whose dependency equations fall within the definition of the restriction. 

\subsection{Entity Hardening Problem}
We prove that the ENH problem is NP-complete in Theorem \ref{CaseIIE}. Using the results of Theorem \ref{CaseIIE} an in-approximability bound of the problem is provided in Theorem \ref{inapxC2}.
\begin{theorem}{} \label{CaseIIE}The ENH Problem is NP Complete 
\end{theorem}

%Paraphrase this proof
%----------------------------------------------------------------------------------------------------------------------------------------------------
\begin{proof}
The Entity Hardening problem is proved to be NP complete by giving a reduction from the Densest $p$-Subhypergraph problem \cite{hajiaghayi2006minimum}, a known NP-complete problem. An instance of the  Densest $p$-Subhypergraph problem includes a hypergraph $G=(V,E_V)$, a parameter $p$ and a parameter $M$. The problem asks the question whether there exists a set of vertices $|V'| \subseteq V$ and $|V'| \le p$ such that the subgraph induced with this set of vertices has at least $M$ hyperedges that are completely covered. From an instance of the  Densest $p$-Subhypergraph problem we create an instance of the ENH problem in the following way. Consider a system $\mathcal{I}(E,\mathcal{F}(E))$ with $E=A \cup B$, where $A$ and $B$ are entities of two separate critical infrastructures dependent on each other. For each vertex $v_i$ and each hyperedge $e_j$ entities $b_i$ and $a_j$ are added to the set $B$ and $A$ respectively. For each hyperedge $e_j$ with $e_j= \{v_m, v_n, v_q\}$ (say) an IDR of form $a_j \leftarrow b_m b_n b_q$ is created. It is assumed that the value of $K$ is set of $|V|$. The values of $k$ and $E_F$ are set to $p$ and $|V|+|E_V|-p-M$ (where $|A|=|V|$ and $|B|=|E|$) respectively.

In the constructed instance only entities of set $A$ are dependent on entities of set $B$. Additionally the dependency for an entity $a_i$ consists of conjunction of entities in set $B$. Hence for an entity $a_i \in A$ to fail, either it itself has to fail initially or any one of the entity that $a_i$ depends on has to fail. It is to be noted that the entities in set $B$ has no induced failure i.e., there is no cascade. Following from this assertion, with $K=|V'|$, failing entities in $B$ would fail all entities in set $A \cup B$. For this created instance $E'$ is set to $B'$ 

If an entity in set $A$ is hardened then it would have no effect in failure prevention of any other entities. Whereas hardening an entity $b_m \in B$ might result in failure prevention of an entity $a_i \in A$ with IDR $a_j \leftarrow b_m b_n b_q$ provided that entities $b_n, b_q$ are also defended. With $k = p$ (and $K \le |V|=|B|$) it can be ensured that entities to be defended are from set $B'$. 

To prove the theorem, consider that there is a solution to the Densest $p$-Subhypergraph problem. Then there exist $p$ vertices which induces a subgraph which has at least $M$ hyperedges. Hardening the entities $b_i \in B'$ for each vertex $v_i$ in the solution of the Densest $p$-Subhypergraph problem would then ensure that at least $M$ entities in set $A$ are protected from failure. This is because the entities in set $A$ for which the failure is prevented corresponds to the hyperedges in the induced subgraph. Thus the number of entities that fail after hardening $p$ entities is at most $|V|+|E_V|-p-M$, solving the ENH problem. Now consider that there is a solution to the ENH problem. As previously stated, the entities to be hardened will always be from set $B'$. So defending $p$ entities from set $B'$ would result in failure prevention of at least $M$ entities in set $A$ such that $E_F \le |V|+|E_V|-p-M$. Hence, the vertex induced subgraph would have at least $M$ hyperedges completely covered when vertices corresponding to the entities hardened are included in the solution of the Densest $p$-Subhypergraph problem. Hence proved. 
\end{proof}
%----------------------------------------------------------------------------------------------------------------------------------------------------

\begin{theorem} {} \label{inapxC2} For a system $\mathcal{I}(E,\mathcal{F}(E))$ with $n=|E|$ and $\mathcal{F}(E)$ having IDRs of form in the created instance of Theorem \ref{CaseIIE}, the ENH problem is hard to approximate within a factor of $\frac{1}{2^{log(n)^{\lambda}}}$ for some $\lambda >0$. 
\end{theorem}

\begin{proof}
The ENH problem with IDRs of form in the created instance of Theorem \ref{CaseIIE} is a special case of the densest $p$-subhypergraph problem. In \cite{hajiaghayi2006minimum} the densest $p-$subhypergraph problem is proved to be inapproximable within a factor of $\frac{1}{2^{log(n)^{\lambda}}}$ ($\lambda >0$). The same result applied to the ENH problem as well. Hence the theorem follows.
\end{proof}

\subsubsection{Restricted Case I: Problem Instance with One Minterm of Size One}
The IDRs of this restricted case have a single minterm of size $1$. This can be represented as $e_{i} \leftarrow e_{j}$, where $e_{i}$ and $e_{j}$ are entities of a system $\mathcal{I}(E,\mathcal{F}(E))$. Algorithm \ref{algCaseI} solves the ENH problem with this restriction optimally in polynomial time utilizing the notion of \emph{Kill Set} defined in Defintion \ref{Kill} with proof of optimality given in Theorem \ref{CaseIE}. 

\begin{algorithm}
\small
	\KwData{A system $\mathcal{I}(E,\mathcal{F}(E))$, set of $K$ entities failing initially $E' , E' \subseteq E$, hardening budget $k$} 
	\KwResult{Set of hardened entities $\mathcal{H}$}
	\Begin{			
		  For each entity $e_i \in E'$ compute the set of kill sets and store it in a set $\mathcal{C}=\{C_{e_1}, C_{e_2}, ..., C_{e_{K}}\}$, where $C_{e_i} = KillSet(e_i)$ \;
		  Set $\mathcal{H}=\emptyset$ \;
			\For  {(i=1; $i \le \mathcal{K}$; i++)}{
					Choose the set $C_{e_k}$ having the highest cardinality from $\mathcal{C}$ \; 
					Update $\mathcal{C} \leftarrow \mathcal{C} \setminus C_{e_k}$ \;
     				\For  {$C_{e_j} \in \mathcal{C}$}{
						Update $C_{e_j} \leftarrow C_{e_j} \setminus C_{e_k}$ \;
					}
					Update $\mathcal{H} \leftarrow \mathcal{H} \cup \{e_k\}$\;
					If all Kill Sets are empty then break \;
			}
			\Return{$\mathcal{H}$}
	}		
\caption{Entity Hardening Algorithm for systems with Restricted Case I type of dependencies}
\label{algCaseI}
\end{algorithm}

\begin{theorem}{} \label{CaseIE} Algorithm \ref{algCaseI} solves the Entity Hardening problem for the Restricted Case I optimally in polynomial time. 
\end{theorem}

\begin{proof}

It is shown in \cite{sen2014identification} that the kill set for all entities in the interdependent network can be computed in $\mathcal{O}(n^3)$ where $n=|E|$. Thus computing the kill sets of $K$ entities would have a time complexity of $\mathcal{O}(K n^2)$. Each update in line 8 would take $\mathcal{O}(n)$ time and hence the total computation of the inner for loop can be done in $\mathcal{O}(K n)$. The outer for loop iterates for $K$ times thus the time complexity of lines $4-9$ is $\mathcal{O}(K^2 n)$. Hence Algorithm \ref{algCaseI} runs in $\mathcal{O}(K n^2)$. 

For two kill sets $C_{e_i}$ and $C_{e_j}$, it can be shown that either $C_{e_i} \cap C_{e_j} = \emptyset$ or $C_{e_i} \cap C_{e_j} = C_{e_i}$ or $C_{e_i} \cap C_{e_j} = C_{e_j}$ \cite{sen2014identification}. Using this assertion the set $E'$ can be partitioned into disjoint subsets $E_{X_1},E_{X_2},..,E_{X_m}$ where kill sets of two entities $e_a, e_b $have no elements in common with $e_a \in E_{X_i}$ and $e_b \in E_{X_j}$ and $ i \ne j$. Additionally, for any given subset of entities $E_{X_z}$ there exist an entity $e_k \in E_{X_z}$ whose kill set is a super set of kill sets of all other entities in $E_{X_z}$. Thus each of the disjoint subset has an entity whose kill set is the super set among all other entities in that subset. Algorithm \ref{algCaseI} chose such an entity in line 5 for every iteration and updates in line 8 would make the kill set of all the remaining entities in the partition to be empty and hence would not be hardened in future iterations. Clearly choosing these entities would globally maximize the total number of protected entities from failure. Hence the Algorithm \ref{algCaseI} is proved to be optimal.
\end{proof}

\subsubsection{Restricted Case II: Problem Instance with an Arbitrary Number of Minterm of Size One}
The IDRs of this restricted case have arbitrary number of minterm of size $1$. This can be represented as $e_{i} \leftarrow \sum^{p}_{q=1} e_{q}$, where $e_{i}$ and $e_{q}$ are entities of a system $\mathcal{I}(E,\mathcal{F}(E))$ and the number of minterms are $p$. The ENH problem with respect to this restricted case is NP-complete and is proved in Theorem \ref{CaseIII}.  We provide an approximation bound for this restricted case of the problem in Theorem \ref{CaseIIIappx2} using the results of Theorem \ref{CaseIII}. The approximation bound uses the notion of \emph{Protection Set} (Definition \ref{PSet}). The Protection Set of an entity can be computed in $\mathcal{O}((n)^2)$ where $n =|E|$ and $m$ are number of minterms.

%Paraphrase all the text within this bound
%----------------------------------------------------------------------------------------------------------------------------------------------------
\begin{theorem}{} \label{CaseIII}The ENH problem for Restricted Case II is NP Complete 
\end{theorem}

\begin{proof}
The ENH problem for case III is proved to be NP complete by giving a reduction from the Set Cover Problem. An instance of the Set Cover problem is given by a set $S=\{x_1,x_2,...,x_n\}$ of elements, a set of subsets $\mathcal{S}=\{S_1,S_2,...,S_m\}$ where $S_i \subseteq S$ and a positive integer $M$. The decision version of the problems finds whether there exist at most $M$ subsets from set $\mathcal{S}$ whose union would result in the set $S$. From an instance of the set cover problem we create an instance of the ENH problem in the following way. Consider a system $\mathcal{I}(E,\mathcal{F}(E))$ with $E=A \cup B$, where $A$ and $B$ are entities of two separate critical infrastructures dependent on each other. For each element $x_i$ in set $S$ we add an entity $a_i$ in set $A$. For each subset $S_i$ in set $\mathcal{S}$ we add an entity $b_i$ in set $B$. For all subsets in $\mathcal{S}$, say $S_p, S_m,S_n$, which has the element $x_i$ an IDR of form $a_i \leftarrow b_m+b_n+b_l$ is added to $\mathcal{F}(E)$.  The values of positive integers $k$ and $E_F$ are set to $M$ and $m-M$ respectively. It is assumed that the value of $K=m$ and $E'=B$.

The constructed instance ensures that the entities to be hardened are from set $B$. This is because hardening an entity $a_i \in A$ would only result in prevention of its own failure whereas hardening an entity $b_j \in B$ would result in failure prevention of its own and all other entities in set $A$ for which it appears in its IDR. 

Consider there exists a solution to the Set Cover problem. Then there exist $M$ subsets whose union results in the set $S$. Hardening the entities in set $B$ corresponding to the subsets selected would ensure that all entities in set $A$ are prevented from failure. This is because for the dependency of each entity $a_i \in A$ there exist at least one entity (in set $B$) that is hardened. Hence the number of entities that fails after hardening is $m-M$ which is equal to $E_F$, thus solving the ENH problem. Now, consider that there is a solution to the ENH problem. As discussed above the entities to be hardened should be from set $B'$. To achieve $E_F=m-M$ with $k=M$, no entities in the set $A$ must fail. Hence for each entity $a_i \in A$ at least one entity in set $B$ that appears in its IDR has to be hardened. Thus, it directly follows that the union of subsets in set $\mathcal{S}$ is equal to the set $S$, solving the Set Cover Problem. Hence proved.
\end{proof}

\begin{defn}
\label{PSet}
\textit{For an entity $e_i \in E$ the Protection Set is defined as the entities that would be prevented from failure by hardening the entity $e_i$ when all entities in $E'$ fail initially. This is represented as $P(x_i|E')$.}
\end{defn}

\noindent
\begin{theorem}
\label{CaseIIIappx1}
For two entities $e_i,e_j \in A \cup B$, $P(e_i|E') \cup P(e_j|AE') = P(e_i,e_j|E')$ when IDRs are in form of Restricted Case II.
\end{theorem}

\begin{proof}
Assume that defending two entities $e_i$ and $e_j$ would result in preventing failure of $P(e_i,e_j|E')$ entities with $|P(e_i|E') \cup P(e_j|E')| < |P(e_i,e_j|E')|$. Then there exist at least one entity $e_p \notin P(e_i|E') \cup P(e_j|E')$ such that it's failure is prevented only if $e_i$ and $e_j$ are protected together. So two entities $e_m$ and $e_n$ (with $e_m \in P(e_i|E')$ and $e_n \in P(e_j|E')$ or vice versa) have to be present in the IDR of $e_p$. As the IDRs are of restricted Case II so if any one of $e_m$ or $e_n$ is protected then $e_p$ is protected, hence a contradiction. On the other way round $P(e_i,e_j|E')$ contains all entities which would be prevented from failure if $e_i$ or $e_j$ is defended alone. So it directly follows that $|P(e_i|E') \cup P(e_j|E')| > |P(e_i,e_j|E')|$ is not possible. Hence the theorem holds. 
\end{proof}

\begin{theorem}
\label{CaseIIIappx2}
There exists an $1-\frac{1}{e}$ approximation algorithm that approximates the ENH problem for Restricted Case II.
\end{theorem}

\begin{proof}
The approximation algorithm is constructed by reducing the problem for this restricted case to \emph{Maximum Coverage} problem. An instance of the maximum coverage problem consists of a set $S=\{x_1,x_2,...,x_n\}$, a set $\mathcal{S}=\{S_1,S_2,...,S_m\}$ where $S_i \subseteq S$ and a positive integer $M$. The objective of the problem is to find a set $S' \subseteq S$ and $|S'| \le M$ such that $\cup_{S_i \in \mathcal{S}} S_i$ is maximized.  Consider a system $\mathcal{I}(E,\mathcal{F}(E))$ with $E=A \cup B$, where $A$ and $B$ are entities of two separate critical infrastructures dependent on each other. For a given initial failure set $E' = A' \cup B'$ with $|A'| + |B'| \le K$, let $P(e_i|A' \cup B')$ denote the protection set for each entity $e_i \in A \cup B$. We construct a set $S=A \cup B$ and for each entity $e_i$ a set $S_{e_i} \subseteq S$ such that $S_{e_i} =  P(e_i|A' \cup B')$. Each set $S_{e_i}$ is added as an element of a set $\mathcal{S}$. The conversion of the problem to Maximum Coverage problem can be done in polynomial time. By Theorem \ref{CaseIIIappx1} defending a set of entities $X \subseteq S$ would result in failure prevention of $\cup_{e_i \in X} S_{x_i}$ entities. Hence, with the constructed sets $S$ and $\mathcal{S}$ and a positive integer $M$ (with $M=k$) finding the Maximum Coverage would ensure the failure protection of maximum number of entities in $A \cup B$. This is same as the ENH problem of Restricted Case II. As there exists an $1-\frac{1}{e}$ approximation algorithm for the Maximum Coverage problem hence the same algorithm can be used to solve this restricted case of the ENH problem using this transformation.
\end{proof}
%----------------------------------------------------------------------------------------------------------------------------------------------------
\subsection{Targeted Entity Hardening Problem}
In this subsection we prove the computational complexity of the Targeted Entity Hardening Problem to be NP-complete in Theorem \ref{th:thmG}.

\begin{theorem}{} \label{th:thmG} The TEH problem is NP-complete
\end{theorem} 
\begin{proof}
We proof that the Targeted Entity Hardening is NP complete by a reduction from Set Cover problem. An instance of the Set Cover problem consists of (i) a set of elements $U=\{x_1,x_2,\hdots,x_n\}$, (ii) a set of subsets $\mathcal{S} = \{S_1,S_2,\hdots,S_m\}$ with $S_i \subseteq U$ $\forall S_i \in \mathcal{S}$, and (iii) a positive integer $M$. The problem asks the question whether there is a subset $\mathcal{S}'$ of $\mathcal{S}$ with $|\mathcal{S}'| \le M$ such that $\bigcup_{S_k \in \mathcal{S}'} S_k = U$. From an instance of the Set Cover problem we create an instance of the Targeted Entity Hardening Problem as follows. Consider a system $\mathcal{I}(E,\mathcal{F}(E))$ with $E=A \cup B$, where $A$ and $B$ are entities of two separate critical infrastructures dependent on each other. For each element $x_j$ in $U$ we add an entity $a_j$ in set $A$. Similarly for each subset $S_i$ in set $\mathcal{S}$ we add an entity $b_i$ in set $B$. For each element $x_i \in U$ which appears in subsets $S_m,S_n,S_p \in \mathcal{S}$ (say) we add an IDR $a_i \leftarrow b_m+b_n+b_p$ to $\mathcal{F}(E)$. There are no IDRs for entities in set $B$ which prevents any cascading failure. The value of $K$ is set to $|\mathcal{S}|$ and $E'=B$ which fails all entities in $A \cup B$. The set of $P$ entities to be protected is set to $A$ and $k$ is set to $M$.

Consider there exists a solution to the Set Cover problem. Then there exist a set $\mathcal{S}'$ of cardinality $M$ such that $\bigcup_{S_k \in \mathcal{S}'} S_k = U$. For each subsets $S_k \in \mathcal{S}'$ we harden the entity $b_k \in B$. So in each IDR of the $A$ type entities there exist a $B$ type entity that is hardened. Hence all $A$ type entities will be protected from failure thus solving the Targeted Entity Hardening problem. 

On the other way round consider there is a solution to the Targeted Entity Hardening problem. This ensures either that for each entity $a_j \in A$ (i) $a_j$ itself is hardened, or (ii) at least one entity from set $B$ in $a_j$'s IDR is hardened. For scenario (i) arbitrarily select an entity $b_p$ in $a_j$'s IDR and include it in set $C$. For scenario (ii) include the hardened entities in the IDR of $a_j$ into set $C$. This is done for each entity $a_j \in A$. For each entity in set $C$ select the corresponding subset in set $\mathcal{S}$. The union of these set of subsets would result in the set $U$. Thus solving the set cover problem. Hence the theorem is proved.
\end{proof}

\subsubsection{Restricted Case I: Problem Instance with One Minterm of Size One}
This restriction imposed on the IDRs is the same as that of restricted case I of the ENH problem. Using the definition of \emph{Protection set} (Definition \ref{PSet}) and the result in Theorem \ref{th:thPS} we design an algorithm (Algorithm \ref{alg:alg1}) that solves the problem for this restricted case optimally in polynomial time (proved in Theorem \ref{CaseI}).

\begin{theorem}
	\label{th:thPS}
	Given a system $\mathcal{I}(E,\mathcal{F}(E))$ with IDRs of form restricted case I and $E' \subset E$ entities failing initially, for any entity $e_i$ and $e_j$ with $e_i \ne e_j$ either (a) $PS(e_i|E') \subseteq PS(e_j|E')$, (b) $PS(e_j|E') \subseteq PS(e_i|E')$, or (c) $PS(e_i|E') \cap PS(e_j|E') = \emptyset$.
\end{theorem}
\begin{proof}
	Consider a directed graph $G=(V,E_D)$. The vertex set $V$ consists of a vertex for each entity in $E$. For each IDR of form $y \leftarrow x$ there is a directed edge $(x,y) \in E_D$. In this proof the term vertex and entity is used interchangeably as an entity is essentially a vertex in $G$. It can be shown that $G$ is either (a) Directed Acyclic Graph (DAG) with maximum in-degree of at most $1$ or, (b) contain at most one cycle with no incoming edge to any vertex in the cycle and maximum in-degree of at most $1$, or (c) collection of graphs (a) and/or (b). Consider a vertex $x_i \in V$. Let $G' = (V',E'_D)$ be a subgraph of $G$ with $V'$ consisting of $x_i$ and all the vertices that has a directed path from $x_i$. Moreover, the edge set $E'_D$ consists of all edges $(x,y) \in E_D$ with $x,y \in V'$ except for any edge $(y,x_i)$ with $y_i \in V'$. Such a subgraph $G'$ would be a directed tree with (i) one or more entities in $V' \backslash \{x_i\}$ is in $A' \cup B'$. Let $X$ denote the set of such entities which satisfy this property, or (ii) no entities in $V' \backslash \{x_i\}$ is in $E'$. If the entity $x_i$ is hardened then for case (i) all the entities in $V'$ would be protected from failure except for entities in all subtrees with roots in $X$. The set of entities in such subtrees are contained in a set $Z$ (say). For this condition if $e_j \in V' \backslash Z$ then $PS(e_j|E') \subset PS(e_i|E')$. Else if $e_j \in Z$ then $PS(x_i|E') \cap PS(e_j|E') = \emptyset$. For case (ii) for any entity $x_j \in V'$ the condition $PS(e_j|E') \subseteq PS(E_i|E')$ always holds (the equality holds for graphs of type (b) as stated above). This property holds for all entities in the entity set $E$. Hence proved.
\end{proof}

\begin{algorithm}
	\small
	\KwData{A system $\mathcal{I}(E,\mathcal{F}(E))$, set $E'$ with $|E'|=K$ entities failing initially and the set $P$ of entities to be protected from failure.}
	\KwResult{A set of entities $H$ to be hardened.}
	\Begin{			
		For each entity $e_i \in (E)$ compute the \emph{Protection Sets} $PS(e_i| E')$ \;
		Initialize $H=\emptyset$ \;
		\While  {$P \ne \emptyset$}{
			Choose the Protection Set with highest $|PS(e_i| E') \cap P|$\;
			Update $H \leftarrow H \cup \{e_i\}$ \;
			Update $P \leftarrow P \backslash PS(e_i| E')$\;
			\For{all $d_j \in E$}{
				$PS(e_j| E') = PS(e_j| E') \backslash PS(e_i| E')$\;
			}
		}
		\Return{$H$} \;
	}		
	\caption{Algorithm for TEH problem with IDRs in form of Restricted Case I}
	\label{alg:alg1}
\end{algorithm}

\begin{theorem}{} \label{CaseI} Algorithm \ref{alg:alg1} solves the Targeted Entity Hardening problem with IDRs having single minterms of size $1$ optimally in polynomial time. 
\end{theorem}
\begin{proof}
The \emph{Protection Sets} of the entities can be found in a similar way as that of computing \emph{Kill Sets} defined in \cite{sen2014identification}. It can be shown that computing these sets for all entities in $E$ can be done in $\mathcal{O}(n^3)$ where $n=|E|$. The while loop in Algorithm \ref{alg:alg1} iterates for a maximum of $n$ times. Step 5 can be computed in $\mathcal{O}(n^2)$ time. The for loop in step 8 iterates for $n$ times. For any given $e_j$ and $e_i$, 	$PS(e_j| E') = PS(e_j| E') \backslash PS(e_i| E')$ can be computed in $\mathcal{O}(n^2)$ time with the worst case being the condition when $|PS(e_i| E')| = |PS(e_j| E')| = n$. As step 9 is nested in a for loop within the while loop this accounts for the most expensive step in the algorithm. The time complexity of this step is $\mathcal{O}(n^4)$. Thus Algorithm \ref{alg:alg1} runs polynomially in $n$ with time complexity being $\mathcal{O}(n^4)$.

In Algorithm \ref{alg:alg1} the while loop iterates till all the entities in $P$ are protected from failure. In step 5 the entity $e_i$ with protection set $PS(e_i| E')$ having most number of entities belonging to set $P$ is chosen to be hardened. Correspondingly the entity $e_i$ is added to the hardening set $H$. The set $P$ is updated by removing the entities in $PS(e_i| E')$. Similarly all the protection sets are updated by removing the entities in $PS(e_i| E')$. 

We use the result from Theorem \ref{th:thPS} to prove the optimality of Algorithm \ref{alg:alg1}. An entity $e_i$ is selected to be hardened at any iteration of the while loop has maximum number of entities in $PS(e_i| E') \cap P$. All entities $e_j$ with $PS(e_j|E') \subseteq PS(e_i|E')$ would have $PS(e_j| E') \cap P \subseteq PS(e_i| E') \cap P$. Moreover there exist no entity $e_k$ for which $PS(e_i|E') \subset PS(e_k|E')$ otherwise $e_k$ would have been hardened instead. Hence there exist no other entity that protect other entities in $P$ including $PS(e_i| E') \cap P$. So Algorithm \ref{alg:alg1} selects the minimum number of entities to harden that protects all entities in $P$.
\end{proof}

\subsubsection{Restricted Case II: IDRs having arbitrary number of minterm of size $1$}
For instance created in Theorem \ref{th:thmG} the IDRs were logical disjunctions of minterms with size $1$. We consider this restriction to design an approximation algorithm for the TEH problem and is shown in Theorem \ref{th:thmApp}.
\begin{theorem}
	\label{th:thmApp}
	The Targeted Entity Hardening Problem is $\mathcal{O}(log(|P|)$ approximate when IDRs are logical disjunctions of minterms with size $1$.
\end{theorem}
\begin{proof}
	We first compute the protection set $PS(e_i |E')$ for all entities $e_i \in E$. Each protection set is pruned by removing entities that are not in set $P$. Now the Targeted Entity Hardening Problem can be directly transformed into Minimum Set Cover problem by setting $U=P$ and $\mathcal{S}=\{PS(e_1|E'), PS(e_2 | E'),..., PS(x_|E| | E')\}$. Selecting the corresponding entities of the protection sets that solve the Minimum Set Cover problem would also solve the Targeted Entity Hardening problem. There exists an approximation ratio of order $\mathcal{O}(log(n))$ (where $n$ is the number of elements in set $U$) for the Set Cover problem. Therefore using the approximation algorithm that solves the Set Cover problem, the same ratio holds for the Targeted Entity Hardening problem with $n=|P|$. Hence proved.
\end{proof} 
\section{Optimal and Heuristic Solution to the Problems}
\label{Solutions}
Owing to the problems being NP-complete, we provide optimal solutions to them by formulating Integer Linear Program (ILP). For both the problems we also provide sub optimal heuristic that runs in polynomial time.
\subsection{Solutions to the Entity Hardening Problem}
\subsubsection{Optimal Solution using Integer Linear Programming}
We propose an Integer Linear Program (ILP) that solves the ENH problem optimally. For a system $\mathcal{I}(E, \mathcal{F}(E))$ let $G=\{g_1,g_2,...,g_n\}$ be variables denoting entities in set $E$. Given an integer $K$, $G$ is a array of $K$ $1$'s and $n-K$ $0$'s where $g_i=1$ if the entity $e_i \in E$ fails at $t=0$ and $g_i = 0$ if the the entity is operational at $t=0$. Thus the array $G$ gives the set of $K$ entities failing initially. Additionally for each entity $e_j \in E$ a set of variables $x_{jd}$ with $0 \le d \le n-1$ and $d \in I^+ \cup \{0\}$ are created. The value of $x_{jd} = 1$ denotes that the entity $x_j$ is in failed state at $t=d$ and $x_{jd} = 0$ denotes it is operational. As noted earlier for $|E| =n $ the cascade can proceed till $n-1$ so the range of $d$ is $[0,n-1]$. Using these definitions the objective of the ENH problem is as follows --- \\
\noindent
\begin{equation}\label{eqn:ilpobj2}
min \Big(\overset{n}{\underset{i=1}{\sum}}x_{i(n-1)}\Big)
\end{equation}

\noindent
The constraints guiding the problem are as follows:

\noindent
{\em Constraint Set 1}: $\sum\limits_{i=1}^{n} q_{x_i} \le k$ , with $q_{x_i} \in [0,1]$. If an entity $x_i$ is hardened then $q_{x_i} = 1$ and $0$ otherwise. \\ \\
\noindent
{\em Constraint Set 2:} $x_{i0} \ge g_{i} - q_{x_i}$. This constraint implies that only if an entity is not defended and $g_i=1$ then the entity will fail at the initial time step.\\
\noindent \\
{\em Constraint Set 3}: $x_{id} \geq x_{i(d-1)}, \forall d, 1 \leq d \leq n-1$, in order to ensure that for an entity which fails in a particular time step would remain in failed state at all subsequent time steps.\\
\noindent \\
{\em Constraint Set 4}: Modeling of constraints to capture the cascade propagation for IIM is similar to the constraints established in \cite{sen2014identification} with modifications to capture the hardening process. A brief overview of this constraint is provided here. Consider an IDR ${e_i} \leftarrow {e_j}{e_p}{e_l} + {e_m}{e_n} + {e_q}$. The following steps are enumerated to depict the cascade propagation with respect to this constraint:

\vspace{0.02in}
\noindent
{\em Step 1:} Replace all minterms of size greater than one with a variable. In the example provided we have the transformed minterm as ${e_i} \leftarrow c_1 + c_2 + e_q$ with $c_1 \leftarrow {e_j}{e_p}{e_l}$ and $c_2 \leftarrow {e_m}{e_n}$ ($c_1,c_2 \in \{0,1\}$) as the new IDRs. 

\vspace{0.02in}
\noindent
{\em Step 2:} For each variable $c$, a constraint is added to capture the cascade propagation. Let $N$ be the number of entities in the minterm on which $c$ is dependent. In the example for the variable $c_1$ with IDR $c_1 \leftarrow {e_j}{e_p}{e_l}$, constraints $c_{1d} \geq \frac{x_{j(d-1)} + x_{p(d-1)}+x_{l(d-1)}}{N} \forall d \in [1,n-1]$ are introduced ($N=3$ in this case). If IDR of an entity is already in form of a single minterm of arbitrary size, i.e.,$e_i \leftarrow {e_j}{e_p}{e_l}$ (say) then constraints $x_{id} \geq \frac{x_{j(d-1)}+x_{p(d-1)}+ x_{l(d-1)}}{N} - q_{x_i}$ and $x_{id} \le x_{j(d-1)} + x_{p(d-1)}+x_{l(d-1)}  \forall d \in [1,n-1]$ are introduced (with $N=3$). These constraints satisfies that if the entity $e_i$ is hardened initially then it is not dead at any time step. 

\vspace{0.02in}
\noindent
{\em Step 3:} Let $M$ be the number of minterms in the transformed IDR as described in Step 1. In the given example with IDR ${e_i} \leftarrow c_1 + c_2 + e_q$  constraints of form $x_{id} \geq c_{1(d-1)} + c_{2(d-1)} + x_{q(d-1)}-(M-1)-q_{x_i}$ and $x_{id} \leq \frac{c_{1(d-1)} + c_{2(d-1)} + x_{q(d-1)}}{M} \forall d \in [0,1]$ are introduced. These constraints ensures that even if all the minterms of $e_i$ has at least one entity in dead state then it will be alive if the entity is hardened initially. \\

\noindent
With objective (\ref{eqn:ilpobj2}) along with the constraints minimize the number of entities failed at the end of the cascading failure with a hardening budget of $k$ and $K$ entities failing initially. The ILP gives an optimal solution to the ENH problem, however its run time is non-polynomial.
 
\subsubsection{Heuristic Solution}
In this subsection we provide a greedy heuristic solution to the Entity Hardening problem. For a given system $\mathcal{I}(E,\mathcal{F}(E))$ with set of entities $E' (|E'| = K)$ failing initially and hardening budget $k$, a heuristic is developed based on the following two metrics --- (a) \emph{Protection Set} as defined in Section \ref{CompAna}, (b) \emph{Cumulative Fractional Minterm Hit Value (CFMHV)} (Definition \ref{CFM}). 

\begin{defn}
	\label{FM}
	The Fractional Minterm Hit Value of an entity $e_j \in E$ in a system $\mathcal{I}(E,\mathcal{F}(E))$ is denoted as $FMHV(e_j, X)$. It is calculated as $FMHV(e_j, X) = \sum_{i = 1}^{m} \frac{1}{|s_i|}$. In the formulation $m$ are the minterms in which $e_j$ appears over all IDRs except for the IDRs of entities in set $X$. The parameter $s_i$ denotes $i^{th}$ such minterm. If entity $e_j$ is hardened (or protected from failure) then the computed value provides an estimate of the future impact on protection of other non operational entities.  
\end{defn} 

\begin{defn}
	\label{CFM}
	The Cumulative Fractional Minterm Hit Value of an entity $e_j \in E$ is denoted as $CFMHV(e_j)$. It is computed as $CFMHV(e_j) = \sum_{\forall x_i \in PS(e_j|E')} FMHV(x_i, PS(x_i|E'))$. This gives a measure of the future impact on protecting non functional entities when the entity $e_j$ is hardened and entities $PS(e_j|E')$ are protected from failure. 
\end{defn} 

Using these definitions a heuristic is formulated in Algorithm \ref{algHeu}. For each iteration of the while loop in the algorithm, the entity having highest cardinality of the set $PS(x_i| A' \cup B') \cap P$ is hardened. This ensures that at each step the number of entities protected is maximized. In case of a tie, the entity having highest Cumulative Fractional Minterm Hit Value among the set of tied entities is selected. This causes the selection of an entity that has the potential to protect maximum number of entities in subsequent iterations. Thus, the heuristic greedily maximizes the number of entities protected when an entity is hardened at the current iteration with metric to measure its impact of protecting other non operational entities in future iterations. Algorithm \ref{algHeu} runs in polynomial time, more specifically the time complexity is $\mathcal{O}(|P| k (n+m)^2)$ (where $n=|E|$ and $m=$ Number of minterms in $\mathcal{F}(E)$). 

\begin{algorithm}
\small	
	\KwData{A system $\mathcal{I}(E,\mathcal{F}(E))$, set of entities $E'$ failing initially with $|E'|= K$ and hardening budget $k$. 		
		}
	\KwResult{Set of hardened entities $\mathcal{H}$.
		}
	\Begin{	
			Initialize $\mathcal{H} \leftarrow \emptyset$\ and $\mathcal{D} \leftarrow \emptyset$\;
			Update $\mathcal{F}(E)$ as follows --- (a) let $Q$ be the set of entities that does not fail on failing $K$ entities, (b) remove IDRs corresponding to entities in set $Q$, (c) update the minterm of remaining IDRs by removing entities in set $Q$\;
			Update $E \leftarrow E \setminus Q$ \;
			\While {($|\mathcal{H}|$ is not equal to $k$)} {
			For each entity $e_i \in E \backslash \mathcal{D}$ compute the Protection Sets $PS(e_i| E')$ \;
			For each entity $e_i \in E \backslash \mathcal{D}$ compute $CFMHV(e_i)$\;
			\If{There exists multiple entities having same value of highest cardinality of the set $PS(e_i| E')$}{
				Let $e_p$ be an entity having highest $CFMHV(e_p)$ among all $e_p$'s in the set of entities having highest cardinality of the set $PS(e_i| A' \cup B')$\;
				If there is a tie choose arbitrarily\;
				Update $\mathcal{H} \leftarrow \mathcal{H} \cup \{e_p\}$ \;
				Update $\mathcal{D} \leftarrow \mathcal{D} \cup PS(e_p| E')$\;
				Update $\mathcal{F}(E)$ by removing entities in $PS(e_p| E')$ both in the left and right side of the IDRs \;
			}
			\Else{
				Let $e_i$ be an entity having highest cardinality of the set $PS(e_i| E')$\;
				Update $\mathcal{H} \leftarrow \mathcal{H} \cup \{e_i\}$ \;
				Update $\mathcal{D} \leftarrow \mathcal{D} \cup PS(e_i| E')$\;
				Update $\mathcal{F}(E)$ by removing entities in $PS(e_i| E')$ both in the left and right side of the IDRs \;		
			}
			}	
	}
	\textbf{return} $\mathcal{H}$ \;		
\caption{Heuristic Solution to the ENH Problem}
\label{algHeu}
\end{algorithm}

\subsection{Solutions to the Targeted Entity Hardening Problem}
\subsubsection{Optimal solution using Integer Linear Program}
The ILP formulation of the TEH problem is similar to that of ENH problem. The only difference being there is no hardening budget in TEH problem and additionally there is a set $P \subset E$ of entities that should be protected from failure. We use the same notations as of the ILP that solves the ENH problem. Using this the objective of the TEH problem is formulated as follows: \\
\noindent
\begin{equation}\label{eqn:ilpobj1}
min \Big(\overset{n}{\underset{i=1}{\sum}}q_{x_i} \Big)
\end{equation}

\noindent
The constraint sets 2,3, and 4 of the ENH problem is employed in the TEH problem as well along with an additional constraint set as described below: \\

\noindent
{\em Additional Constraint Set:} For all entities $e_i, \in P$, $x_{i(n-1)} = 0$. This ensures that all the entities in set $P$ are protected from failure at the final time step. 

With these constraints, the objective in (\ref{eqn:ilpobj1}) minimizes the number of hardened entities that results in protection of all entities in set $P$. 

\subsubsection{Heuristic Solution}
In this subsection we provide a greedy heuristic solution to the TEH problem. For a given system $\mathcal{I}(E,\mathcal{F}(E))$ with set of entities as $E' (|E'| = K)$ failing initially and set of entities to protet being $P$, a heuristic is developed based on the following two metrics --- (a) \emph{Protection Set} as defined in Section \ref{CompAna}, (b) \emph{Prioritized Cumulative Fractional Minterm Hit Value (PCFMHV)} (Definition \ref{CFMT}). 

\begin{defn}
	\label{FMT}
	The Prioritized Fractional Minterm Hit Value of an entity $e_j \in E$ in an interdependent network $\mathcal{I}(E,\mathcal{F}(E))$ is denoted as $FMHV(e_j, X)$. It is calculated as $PFMHV(e_j, P) = \sum_{i = 1}^{m} \frac{1}{|s_i|}$. In the formulation $m$ are the minterms in which $e_j$ appears over IDRs in non operational entities in set $P$. The parameter $s_i$ denotes $i^{th}$ such minterm. If the $e_j$ is hardened (or protected from failure) the value computed provides an estimate future impact on protection of other non operational entities in set $P$.  
\end{defn} 

\begin{defn}
	\label{CFMT}
	The Prioritized Cumulative Fractional Minterm Hit Value of an entity $e_j \in E$ is denoted as $PCFMHV(e_j)$. It is computed as $PCFMHV(e_j) = \sum_{\forall x_i \in PS(e_j|E')} PFMHV(x_i, PS(x_i|E'))$. This gives a measure of future impact on protecting non functional entities in $P$ when the entity $e_j$ is hardened and entities $PS(e_j|E')$ are protected from failure. 
\end{defn} 

\begin{algorithm}[ht!]
	\small
	\KwData{A system $\mathcal{I}(E,\mathcal{F}(E))$, set of $K$ vulnerable entities and the set $P$ of entities to be protected from failure. 
	}
	\KwResult{A set of entities $H$ to be hardened.
	}
	\Begin{			
		Initialize $\mathcal{D}=\emptyset$ and $H=\emptyset$ \;
		Update $\mathcal{F}(E)$ as follows --- (a) let $Q$ be the set of entities that does not fail on failing $K$ entities, (b) remove IDRs corresponding to entities in set $Q$, (c) update the minterm of remaining IDRs by removing entities in set $Q$\;
		\While{$P \ne \emptyset$}{
			For each entity $e_i \in E \backslash \mathcal{D}$ compute the Protection Sets $PS(e_i| E')$ \;
			For each entity $e_i \in E \backslash \mathcal{D}$ compute $PCFMHV(e_i)$\;
			\If{There exists multiple entities having same value of highest cardinality of the set $PS(e_i| E') \cap P$}{
				Let $e_p$ be an entity having highest $CFMHV(e_p)$ among all $e_p$'s in the set of entities having highest cardinality of the set $PS(e_i| A' \cup B')$\;
				If there is a tie choose arbitrarily\;
				Update $H \leftarrow H \cup \{e_p\}$ \;
				Update $\mathcal{D} \leftarrow \mathcal{D} \cup PS(e_p| E')$\;
				Update $P \leftarrow P \backslash PS(e_p| E')$\;
				Update $\mathcal{F}(E)$ by removing entities in $PS(e_p| E')$ both in the left and right side of the IDRs \;
			}
			\Else{
				Let $e_i$ be an entity having highest cardinality of the set $PS(e_i| E') \cap P$\;
				Update $H \leftarrow H \cup \{e_p\}$ \;
				Update $\mathcal{D} \leftarrow \mathcal{D} \cup PS(e_i| E')$\;
				Update $P \leftarrow P \backslash PS(e_i| E')$\;
				Update $\mathcal{F}(E)$ by removing entities in $PS(e_i| E')$ both in the left and right side of the IDRs \;		
			}
		}
		\Return{$H$} \;
	}		
	\caption{Heuristic solution to the TEH problem}
	\label{alg:alg2}
\end{algorithm} 

Using these definitions, the heuristic is formulated in Algorithm \ref{alg:alg2}. For each iteration of the while loop in the algorithm, the entity having highest cardinality of the set $PS(x_i| A' \cup B') \cap P$ is hardened. This ensures that at each step the number of entities protected in set $P$ is maximized. In case of a tie, the entity having highest Prioritized Cumulative Fractional Minterm Hit Value among the set of tied entities is selected. This causes the selection of the entity that has the potential to protect maximum number of entities in updated set $P$ in subsequent iterations. Thus, the heuristic greedily minimizes the set of entities hardened which would cause protection of all entities in $P$. The heuristic overestimates the cardinality of $H$ from the optimal solution. Algorithm \ref{alg:alg2} runs in polynomial time, more specifically the time complexity is $\mathcal{O}(|P| n (n+m)^2)$ (where $n=|E|$ and $m=$ Number of minterms in $\mathcal{F}(E)$). 

It is to be noted if Algorithm \ref{algHeu} and \ref{alg:alg2} returns $H$ (or $\mathcal{H}$ as in ENH problem) with $|H| \ge K$ then we harden entities belonging to the set of $K$ initially failing entities. This is because hardening these $K$ initially failing entities would protect all entities in the interdependent network from failure.

\section{Generating IDRs and Experimental Results}
\label{ExpRes}
\subsection{Generating Dependency Equations for Power Network}
In this subsection, we describe a strategy to generate dependency equations of an intra-dependent power network. We restrict to load bus, generator bus, neutral bus and transmission line as the entities in the power network. For a given power network, AC power equations are solved to determine the direction of flow in the transmission lines. We use the power flow solver available in MatPower software for different bus systems \cite{zimmerman2011matpower}. For a given set of buses and transmission lines, the MatPower software uses load demand of the bus, impedance of the transmission lines etc. to solve the power flow and outputs the voltage of each bus in the system. We restrict to real power flow analysis. For a given solution, the real part of generation is taken as the power generated by a generator bus. Similarly, the real part of the load demand is taken as demand value of a load bus. For two buses $e_1$ and $e_2$ connected by a transmission line $e_{12}$ the power flowing through the transmission line is calculated as $P_{12} = Real(V_1 * (\frac{V_1 - V-2}{I_{12}})^*)$, where $V_1$ is the voltage at bus $e_1$, $V_2$ is the voltage at bus $e_2$, $I_{12}$ is the impedance of the transmission line $e_{12}$ and $(\frac{V_1 - V-2}{I_{12}})^*$ denotes the complex conjugate of $(\frac{V_1 - V_2}{I_{12}})$. $P_{12}$ is the real component of the power flowing in the transmission line $e_{12}$. Power flows from bus $e_1$ to $e_2$ if $P_{12}$ is positive and from bus $e_2$ to $e_1$ otherwise.  

The generation of the dependency equation is explained through a nine bus system shown in Figure \ref{fig:dependentPower}. The figure represents a system $\mathcal{I}(E,\mathcal{F}(E))$ with set $E$ consisting of generator buses from $G_1$ through $G_3$, load buses $L_1$ through $L_4$, neutral buses $\{N_1,N_2\}$ and transmission lines $T_1$ through $T_9$. The values in the red blocks denote the amount of power a generator is generating, the green block being the load requirements and blue neutral (value of $0$). The value in the grey blocks correspond to power flow in the transmission lines. The transmission lines don't have any IDR. The IDRs for a bus $b_1$ is constructed by the following --- (a) let $b_2$, $b_3$ be the buses and $b_{12}$ (between $b_1$ and $b_2$) and $b_{13}$ between ($b_1$ and $b_3$) be the transmission lines for which power flows from these buses to $b_1$, (b) the dependency equation for the bus $b_1$ is constructed as conjunction of minterms of size $2$ (consisting of the bus from which power is flowing and the transmission line) with each conjunction corresponding to bus that has power flowing to it. For this example the dependency equation $b_1 \leftarrow b_{12} b_2 + b_{13} b_{3}$ is created. Using this definition the dependency equations for the buses in Figure \ref{fig:dependentPower} are created and is shown in Table \ref{tb:depenPower}. 

The following points are to be noted regarding the generation rule --- (a) The transmission lines can only fail initially due to a man made attack or natural disaster. Hence it entails the underlying assumption that the transmission lines would have enough capacity to transmit any power that is required to flow in it, (b) The generator bus is also only susceptible to initial failure and is assumed to have a generation capacity that is enough to supply the power demanded by a instance of power flow, (c) Neutral and Load buses are prone to both initial and induced failure. For example consider the failure of transmission lines $T_9$ and $T_1$ at $t=0$. Owing to this the load bus $L_1$ and neutral bus $N_2$ fails at $t=2$. At $t=3$ load bus $L_2$ fails due to the failure of buses $L_1, N_2$. It is to be noted that load bus $L_3$ does not fails as it still receives power from $N_1$ as transmission line $T_4$ is expected to have a capacity that can support a power flow equal to the demand of $L_3$.

\begin{table}[ht]
\begin{minipage}[b]{0.35\linewidth}
\centering
\begin{tabular}{ | l | r | r | r |} \hline
   {\bf Dependency Equations} \\ \hline
			$L_1 \leftarrow T_1 G_1$  \\ \hline
			$L_2 \leftarrow T_2 L_1 + T_7 N_2$  \\ \hline
			$L_3 \leftarrow T_3 L_1 + T_4 N_1$   \\ \hline
			$L_4 \leftarrow T_6 N_1 + T_8 N_2$ \\ \hline
			$N_1 \leftarrow T_5 G_3$ \\ \hline
			$N_2 \leftarrow T_9 G_2$ \\ \hline
		\end{tabular}
		\caption{IDRs of the buses in Figure \ref{fig:dependentPower}}
		\protect\label{tb:depenPower}
\end{minipage}\hfill
\begin{minipage}[b]{0.7\linewidth}
\centering
\includegraphics[width=0.8\textwidth]{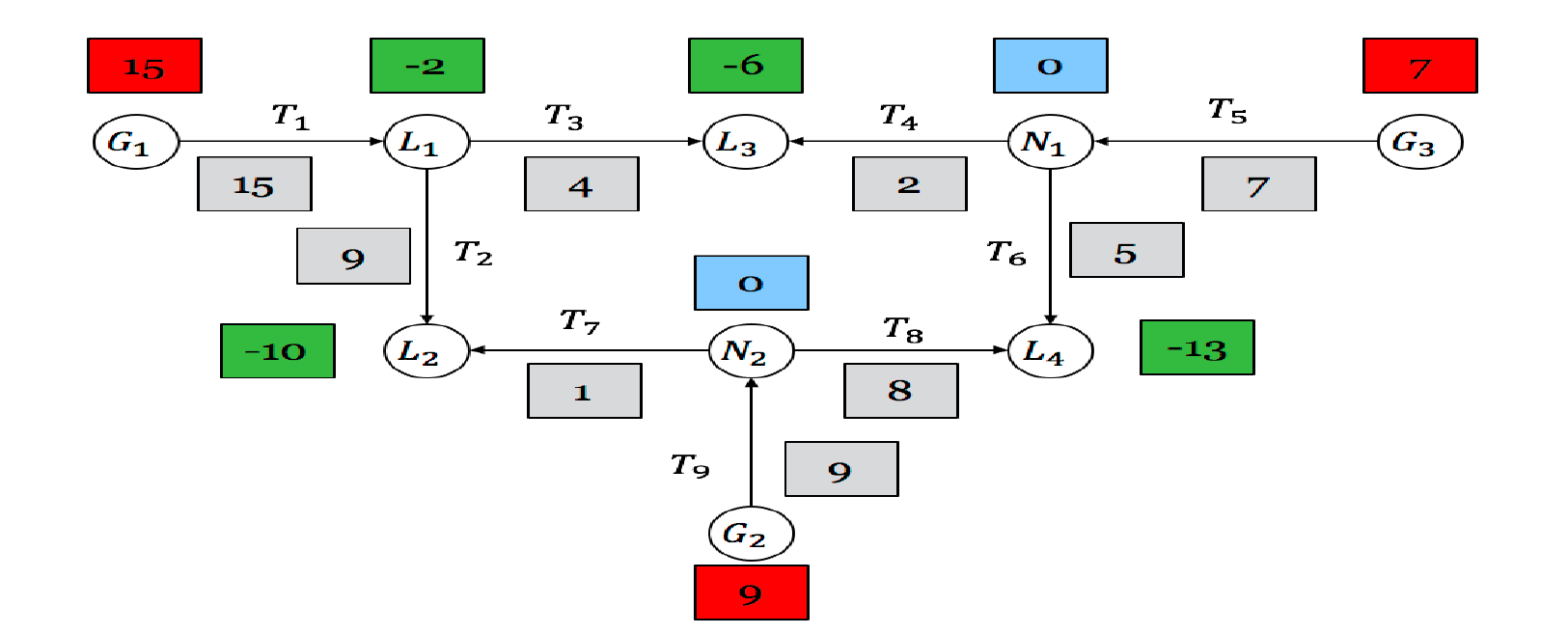}
    \captionof{figure}{Example of Power Network Dependency}
    \label{fig:dependentPower}
\end{minipage}
\end{table}

Owing to the underlying assumptions in the the creation of dependency equations, there is a limitation to its applicability to real world problems. However, with respect to power network, creating dependency equations like the one discussed is a preliminary step. Further research is required to be done to have a more accurate abstract representation of the dependency equations that can have widespread applicability to real world problems. The purpose of this subsection is --- (1) presenting a preliminary way the dependency equations can be generated for power network, (2) larger data sets that can be used to measure the performance of the optimal solution to the heuristic. 

\subsection{Generating Dependency Equations for Interdependent Power-Communication Network}
In this subsection, we describe rules to generate dependency relations for interdependent power and communication network infrastructure as used in \cite{sen2014identification}. Real world data of Maricopa County, Arizona , USA was taken. This county is one of the most densest populated region of Arizona with approximately $60 \%$ residents. Specifically, we wanted to measure the amount of resource required to protect entities in particular regions of the county when these regions have a set of entities failing initially. The data for power network was obtained from Platts (http://www.platts.com/) that contains $70$ generator buses (including solar homes that generate minuscule unit of power) and $470$ transmission lines. The communication network data was obtained from GeoTel (http://www.geo-tel.com/) consisting of  $2,690$ cell towers, $7,100$ fiber-lit buildings and  $42,723$ fiber links. Figures \ref{fig:powerNet} and \ref{fig:comNet} displays the snapshot of power network and communication network for a particular region of Maricopa county. In Figure \ref{fig:powerNet} the orange dots represent the generator buses and continuous yellow lines represent the transmission lines. In Figure \ref{fig:comNet} fiber-lit buildings are represented by pink dots, cell towers by orange dots and fiber links by continuous green lines. 

The \emph{load} of the power network are assumed to be cell towers and fiber-lit buildings. There exist other entities that draws electrical power. Since it is not relevant for the comparative analysis of the heuristic and the ILP such entities are ignored. The interdependent power-communication system is represented mathematically as $\mathcal{I}(E,\mathcal{F}(E))$ with $E = A \cup B$. $A$ and $B$ consist of the entities in the power network and communication network respectively. With respect to this data the power network consist of three type of entities --- generating stations, load (which are cell towers and fiber-lit buildings) and transmission lines (denoted by $a_1,a_2,a_3$ respectively). The communication network comprises of the following type of entities --- cell towers, fiber-lit buildings and fiber links (denoted by $b_1,b_2,b_3$ respectively). It is to be noted that the fiber-lit buildings and cell towers are considered as both power network entities as well as communication network entities. From the raw data the dependency equations are constructed using the following rules.

%Parapharse this bounded region. Note: In robustness the same block will be there. So you need to do two separate paraphrashing
%-------------------------------------------------------------------------------------------------------------------------------------------
\vspace{0.05in}
\noindent
{\bf Rules:} We take into consideration that an entity in the power network is dependent on a set of entities in the communication network for either being operational and vice-versa. To keep things uncomplicated, we consider the dependency equations with at most two minterms. For the same reason we consider the size of each minterm is at most two. 

\vspace{0.05in}
\noindent
{\em Generators ($a_{1, i}, 1 \leq i \leq p$, where $p$ is the total number of generators):} We assume that every generator ($a_{1.i}$) is, i) dependent on the closest Cell Tower ($b_{1,j}$), or, ii) closest Fiber-lit building ($b_{2,k}$) and the corresponding Fiber link ($b_{3,l}$) connecting $b_{2,k}$ and $a_{1,i}$. Hence, we have $a_{1,i} \leftarrow b_{1,j}+b_{2,k} \times b_{3,l}$.

\vspace{0.05in}
\noindent
{\em Load ($a_{2,i}, 1 \leq i \leq q$, where $q$ is the total number of loads):} The power network loads do not depend on any entities in communication network

\vspace{0.05in}
\noindent
{\em Transmission Lines ($a_{3,i}, 1 \leq i \leq r$, where $r$ is the total number of transmission lines):} The transmission lines in the power network do not depend on any entities in communication network.

\vspace{0.05in}
\noindent
{\em Cell Towers ($b_{1,i}, 1 \leq i \leq s$, where $s$ is the total number of cell towers):} The Cell Towers depend on two components, i) the closest pair of generators, and, ii) corresponding transmission line, connecting the generator to the cell tower. Thus we have $b_{1,i} \leftarrow a_{1,j} \times a_{3,k}+a_{1,j\rq{}} \times a_{3,k\rq{}}$.

\vspace{0.05in}
\noindent
{\em Fiber-lit Buildings ($b_{2,i}, 1 \leq i \leq t$, where $t$ is the total number of fiber-lit buildings): } The Fiber-lit Buildings depend on two components, i) the closest pair of generators, and, ii) corresponding transmission line, connecting the generator to the fiber-lit buildings. Thus we have $b_{2,i} \leftarrow a_{1,j} \times a_{3,k}+a_{1,j\rq{}} \times a_{3,k\rq{}}$.

\vspace{0.05in}
\noindent
{\em Fiber Links ($b_{3,i}, 1 \leq i \leq u$, where $u$ is the total number of fiber links):} The Fiber Links aren't dependent on any power network entity. These links require power only for the amplifiers connected to them. The amplifiers are required if the length of the fiber link is above a certain threshold. We consider only those fiber links which are 'quite long', need power. The fiber links depend on the closest pair of generators and the transmission lines connecting the generators to the fiber link under consideration. Thus we have $b_{3,i} \leftarrow a_{1,j} \times a_{3,k}+a_{1,j\rq{}} \times a_{3,k\rq{}}$. We do not consider that these fiber links need any power as we cannot determine the length of the fiber links or the exact threshold value due to the lack of data.
%-------------------------------------------------------------------------------------------------------------------------------------------

\begin{figure}[ht]
\centering
\begin{subfigure}{.45\textwidth}
  \centering
  \includegraphics[width=7cm,height=7cm,keepaspectratio]{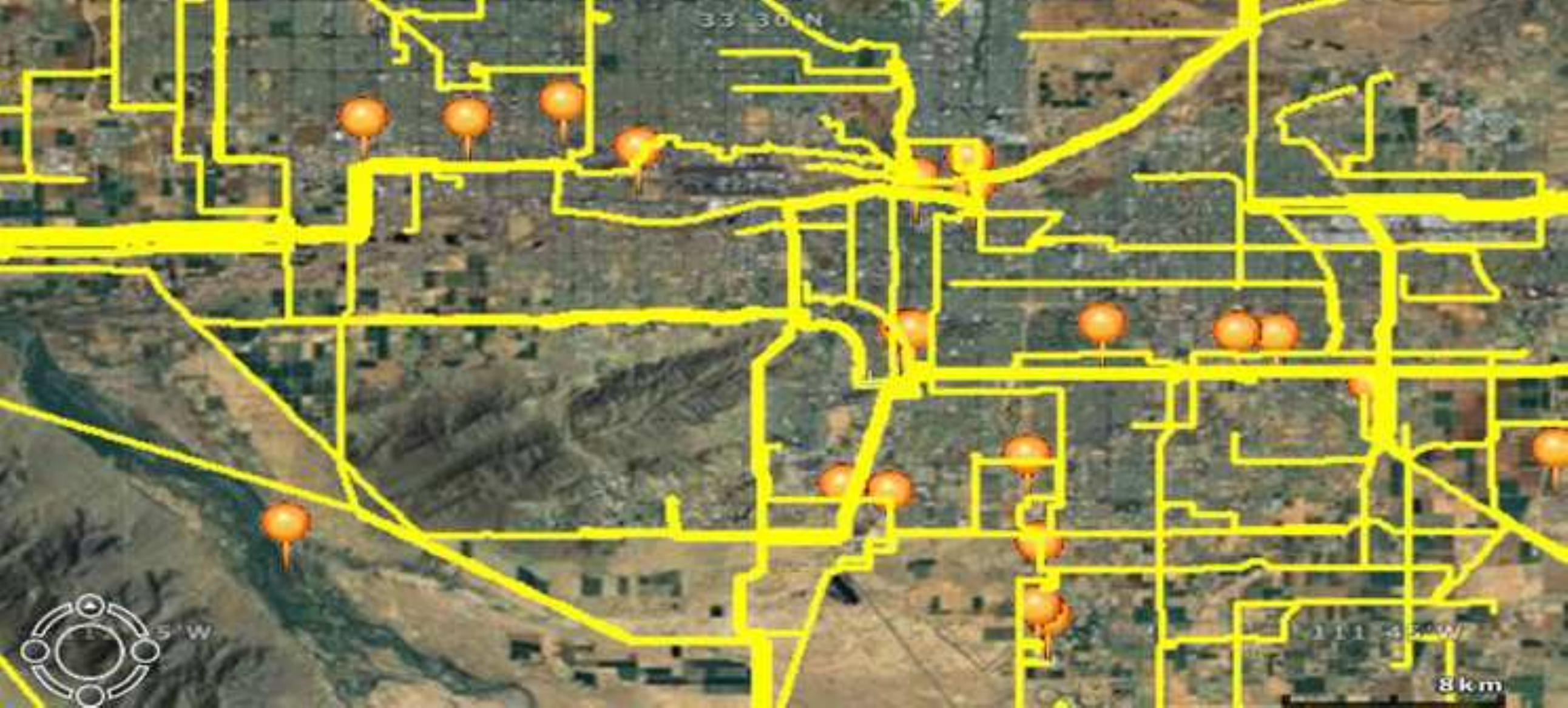}
  \caption{Snapshot of Power Network}
  \label{fig:powerNet}
\end{subfigure}
\begin{subfigure}{.45\textwidth}
  \centering
  \includegraphics[width=7cm,height=7cm,keepaspectratio]{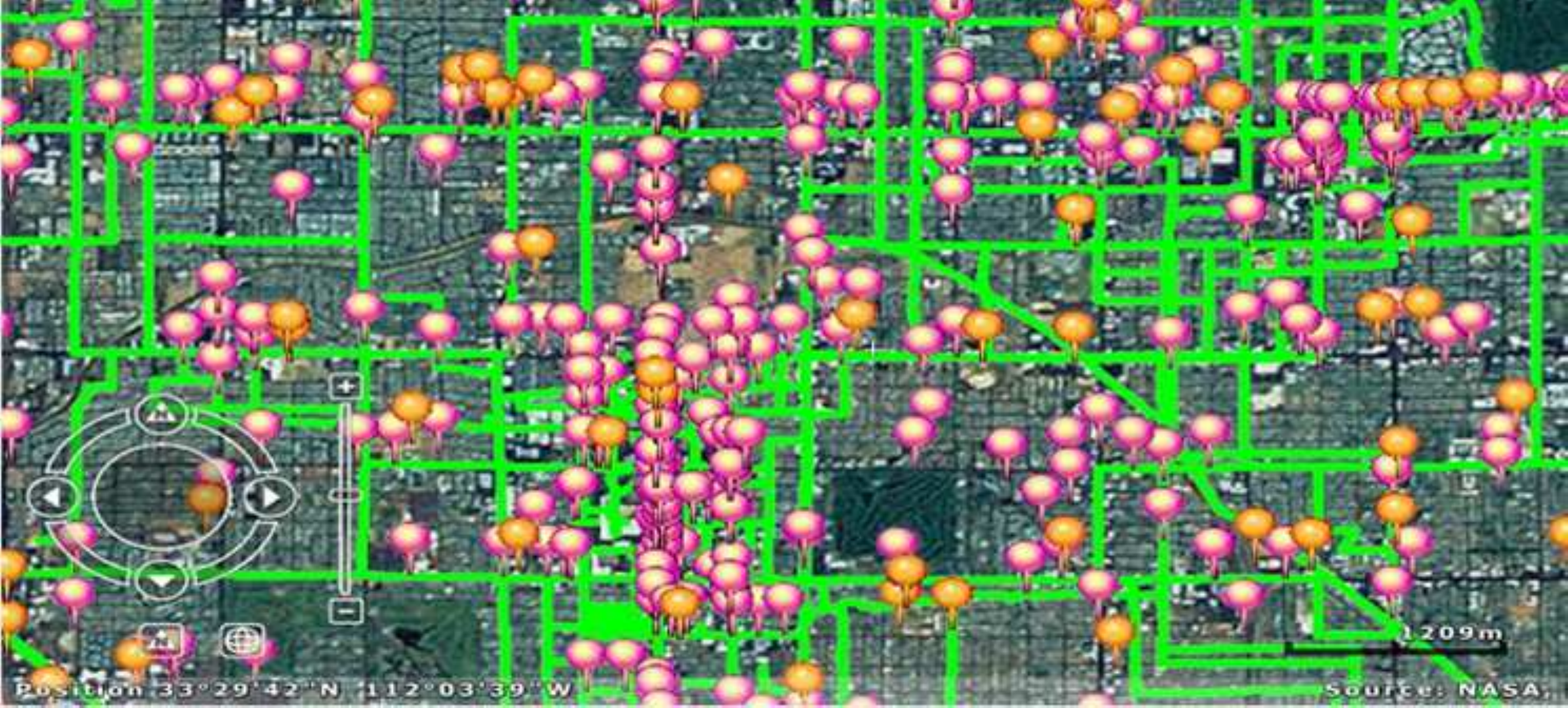}
  \caption{Snapshot of Communication Network}
  \label{fig:comNet}
\end{subfigure}
\caption{Snapshots of the real world data corresponding to power and communication network}
\label{fig:data}
\end{figure}

\subsection{Comparative Study of the ILP and Heuristic for the Problems} 
A comparative study of the ILP and heuristic solution for both the problems is done in this subsection. A machine with intel i5 processor and 8 GB of RAM was used to execute the solutions. The coding was done in \emph{java} and a student licensed \emph{IBM CPLEX} external library file is used to execute the ILP. $8$ different bus systems available from MatPower with number buses $24,30,39,57,89,118,145,300$ were used to generate the dependency equations for power network (using the rules described in Section 7.1). The time to generate the dependency equations were less than $2 ms$. Within the Maricopa county $4$ disjoint regions were considered labeled as Region 1 through 4. Dependency equations for the interdependent power-communication network were generated for these regions using the rules described in Section 7.2. The java codes along with data files of the generated dependency equations are open sourced and is available in the following url \emph{https://github.com/jbanerje1989/HardeningProblem}.

The number of entities in each of the $12$ data sets are enumerated in Table \ref{tbl:Sets}. To determine the initially failing entities we used the ILP solution of $\cal{K}$ most vulnerable entities in \cite{sen2014identification}. The $\cal{K}$ most vulnerable entities problem finds a set of $\cal{K}$ entities in a system $\mathcal{I}(E,\mathcal{F}(E))$ which when failed at $t=0$ causes the maximum number of entities to fail. For a given data set representing a system $\mathcal{I}(E,\mathcal{F}(E))$, for both the problems the initially failing entities was taken as a set $E'$ ($|E'| = K$) such that --- (a) The set $E'$ constitutes the $K$ most vulnerable entities in the system, (b) Failing the entities in set $E'$ would cause failure of at least $|E| / 2$ entities in total. The cardinality of the set $E'$ along with the total number of entities failed are enumerated in Table \ref{tbl:Sets}.

In comparing the ILP and heuristic solution of the ENH problem we considered $5$ distinct hardening budgets for each data set. With $K$ being the number of initially failing entities in a data set the hardening budgets were chosen between $[1, K-1]$ (with value of $K$ obtained from Table \ref{tbl:Sets}). It is also ensured that the hardening budgets chosen had a high variance. Figures \ref{fig:case24} - \ref{fig:D4} shows the total number of entities protected from failure for each data set using the ILP and heuristic solution. The run-time performance of the solutions are provided in Table \ref{tbl:RuntimeH} (in the table 'Heu' refers to the heuristic solution and $Hi$ refers to the hardening budget corresponding to the $i^{th}$ budget from left used in the bar graph plots). From Figures \ref{fig:case24} - \ref{fig:D4} it can be seen that the heuristic performs almost similar to that of the ILP solution in terms of quality. The maximum percent difference of the total number of entities protected in the ILP when compared to the heuristic solution occurs for a hardening budget of $39$ in the $145$ bus system (Figure \ref{fig:case145}) with the percent difference being $3.1 \%$. In terms of run-time, heuristic outperforms the ILP with the heuristic computing solutions nearly $200$ times faster in larger systems (as seen for the $300$ bus system in Table \ref{tbl:RuntimeH}). Hence it can be reasonably argued that the heuristic produces fast near optimal solutions for the ENH problem.

A similar kind of experimental analysis is performed for the TEH problem. $5$ distinct protection sets $P$ were considered for each data set. Let $F$ denote the set entities failed in total when $K$ entities fail initially. The cardinality of set $F$ and the value of $K$ was taken from Table \ref{tbl:Sets} for each data set. The cardinality of the protection set for a given data set was chosen between $[1,|F| - 1]$ ensuring that the chosen values have high variance. For a given cardinality $C$ the protection set $P$ was constructed by choosing $C$ entities from the set $F$ corresponding to a particular data set. Figures \ref{fig:case24T} - \ref{fig:D4T} shows the comparison of the Heuristic solution with the ILP in terms of total number of entities hardened for a given cardinality of protection budget. The run-time comparison of the solutions are provided in Table \ref{tbl:RuntimeT}. A maximum percent difference of $25 \%$ (ILP compared with Heuristic) in the number of entities hardened can be seen in Region 2 for a $|P|$ value of $13$ (Figure \ref{fig:D2T}). However, for most of the cases the heuristic produces near optimal or optimal solution. The heuristic also compute the solutions nearly $200$ times faster than the ILP for larger systems as seen in Table \ref{tbl:RuntimeT}. Hence it can be claimed that the heuristic solution to the TEH problem produces near optimal solution at a much faster time compared to the ILP solution.

\begin{table}[ht]
    \centering
		\begin{tabular}{|c|c|c|c|}  \hline
			\multicolumn{1}{|c}{\bf{DataSet}} & \multicolumn{1}{|c|}{\bf{Num. Of Entities}} & \multicolumn{1}{|c|}{\bf{$K$}} & \multicolumn{1}{|c|}{\bf{Num. of Entities Killed}} \\ \cline{1-4}
			\bf{24 bus} & $58$ & $8$ & $29$ \\ \hline
			\bf{30 bus} & $71$ & $13$ & $36$ \\ \hline
			\bf{39 bus} & $84$ & $17$ & $42$ \\ \hline
			\bf{57 bus} & $135$ & $26$ & $68$ \\ \hline
			\bf{89 bus} & $295$ & $78$ & $147$ \\ \hline
			\bf{118 bus} & $297$ & $89$ & $149$ \\ \hline
			\bf{145 bus} & $567$ & $191$ & $284$ \\ \hline
			\bf{300 bus} & $709$ & $145$ & $355$ \\ \hline
			\bf{Region 1} & $48$ & $6$ & $26$ \\ \hline
			\bf{Region 2} & $46$ & $8$ & $23$ \\ \hline
			\bf{Region 3} & $48$ & $6$ & $24$ \\ \hline
			\bf{Region 4} & $53$ & $8$ & $27$ \\ \hline
		\end{tabular}
		\captionsetup{justification=centering}
		\caption{Number of entities, $K$ value chosen and number of entities failed when the $K$ vulnerable entities are failed initially for different data sets}
		\protect\label{tbl:Sets}
\end{table}

\noindent\begin{minipage}{\textwidth}
\begin{minipage}[b]{0.3\linewidth}
\centering
\includegraphics[width=5cm,height=5cm,keepaspectratio]{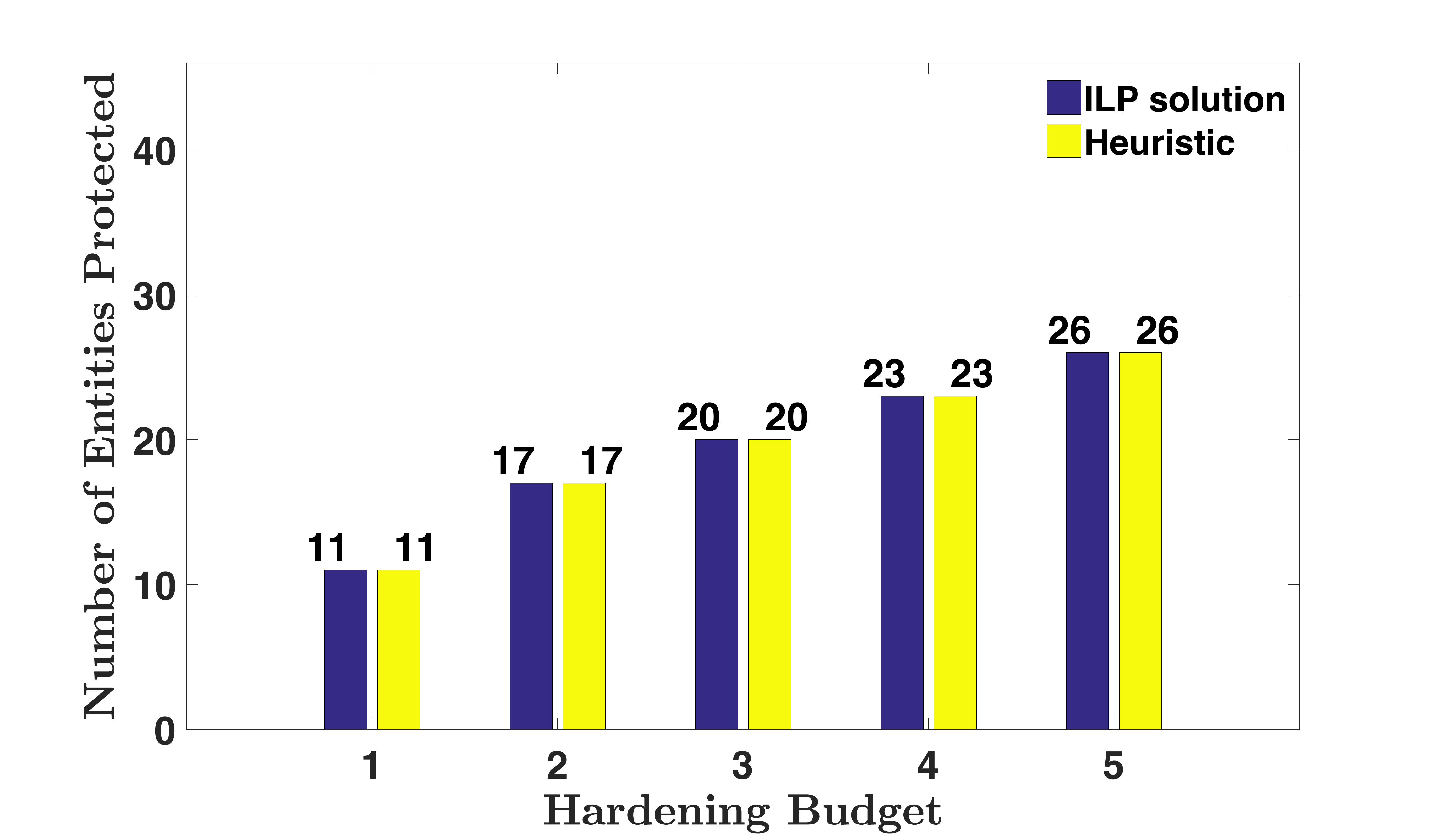}
  \captionsetup{justification=centering}
  \captionof{figure}{Comparison of ILP solution with Heuristic for 24 bus system (ENH)}
  \label{fig:case24}
\end{minipage} \hfil{}
\begin{minipage}[b]{0.3\linewidth}
\centering
\includegraphics[width=5cm,height=5cm,keepaspectratio]{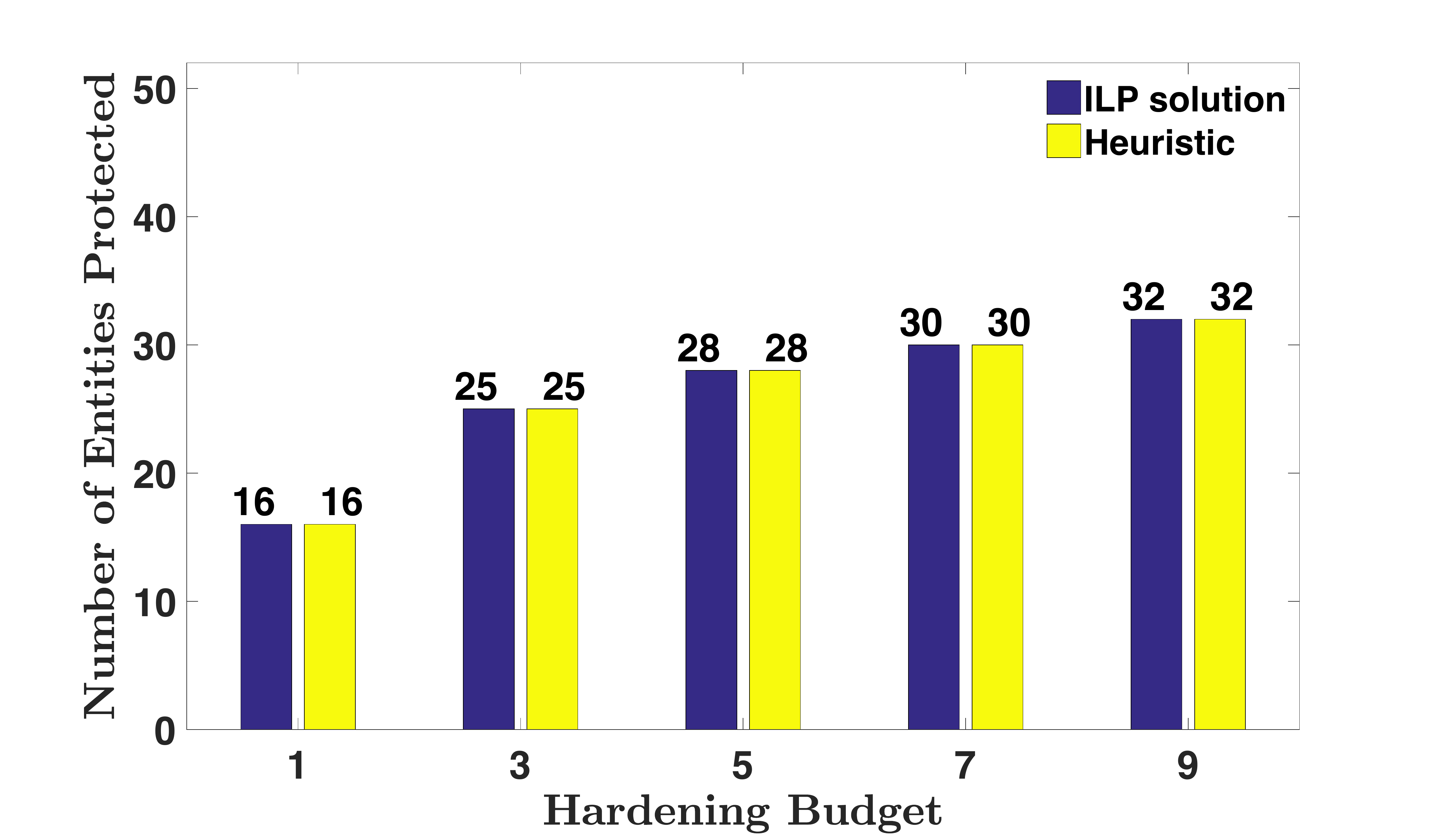}
  \captionsetup{justification=centering}
  \captionof{figure}{Comparison of ILP solution with Heuristic for 30 bus system (ENH)}
  \label{fig:case30}
\end{minipage}
\hfil{}
\begin{minipage}[b]{0.3\linewidth}
\centering
\includegraphics[width=5cm,height=5cm,keepaspectratio]{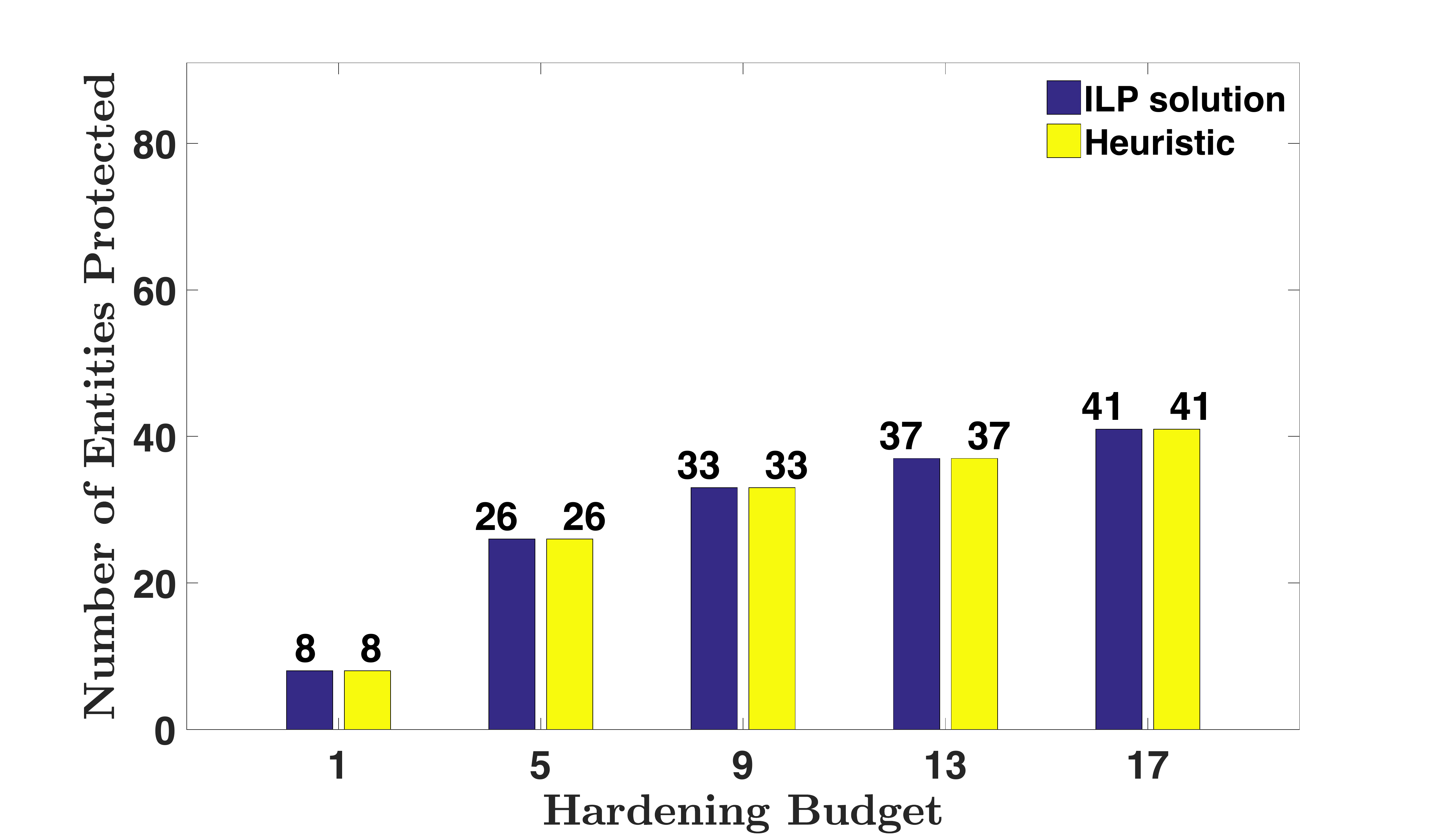}
  \captionsetup{justification=centering}
  \captionof{figure}{Comparison of ILP solution with Heuristic for 39 bus system (ENH)}
  \label{fig:case39}
\end{minipage}
\end{minipage}

\noindent\begin{minipage}{\textwidth}
\begin{minipage}[b]{0.3\linewidth}
\centering
\includegraphics[width=5cm,height=5cm,keepaspectratio]{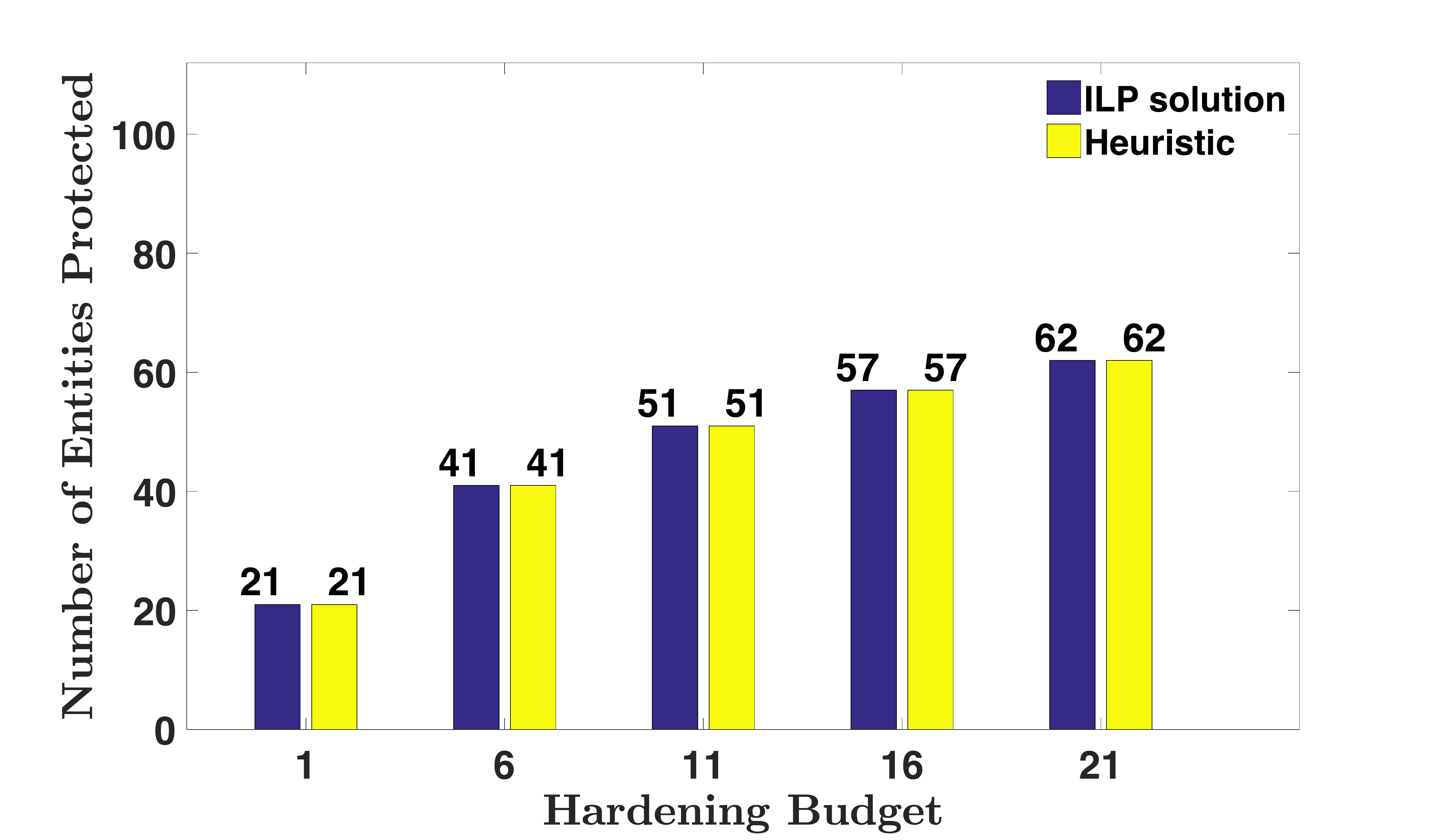}
  \captionsetup{justification=centering}
  \captionof{figure}{Comparison of ILP solution with Heuristic for 57 bus system (ENH)}
  \label{fig:case57}
\end{minipage} \hfil{}
\begin{minipage}[b]{0.3\linewidth}
\centering
\includegraphics[width=5cm,height=5cm,keepaspectratio]{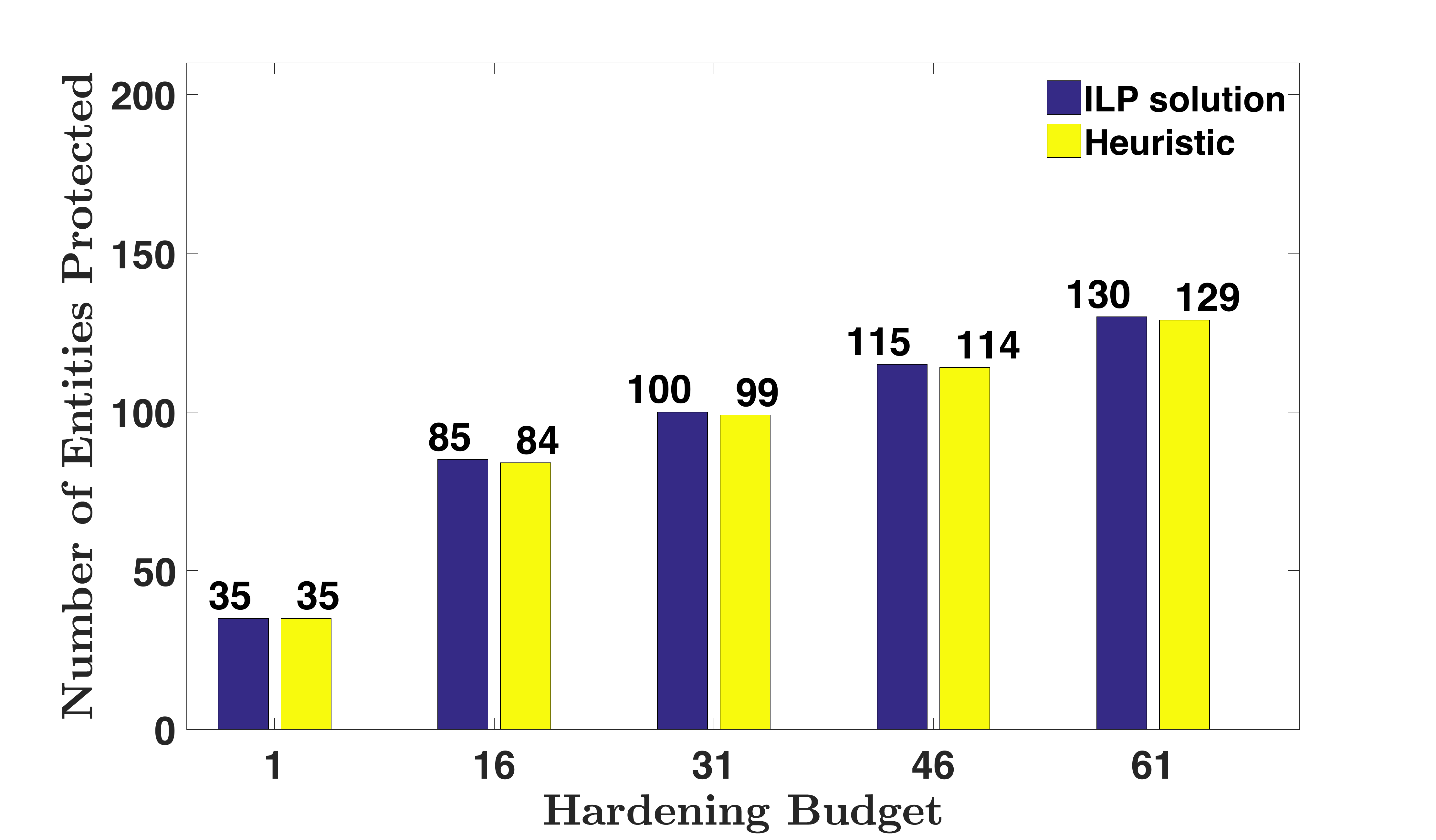}
  \captionsetup{justification=centering}
  \captionof{figure}{Comparison of ILP solution with Heuristic for 89 bus system (ENH)}
  \label{fig:case89}
\end{minipage}
\hfil{}
\begin{minipage}[b]{0.3\linewidth}
\centering
\includegraphics[width=5cm,height=5cm,keepaspectratio]{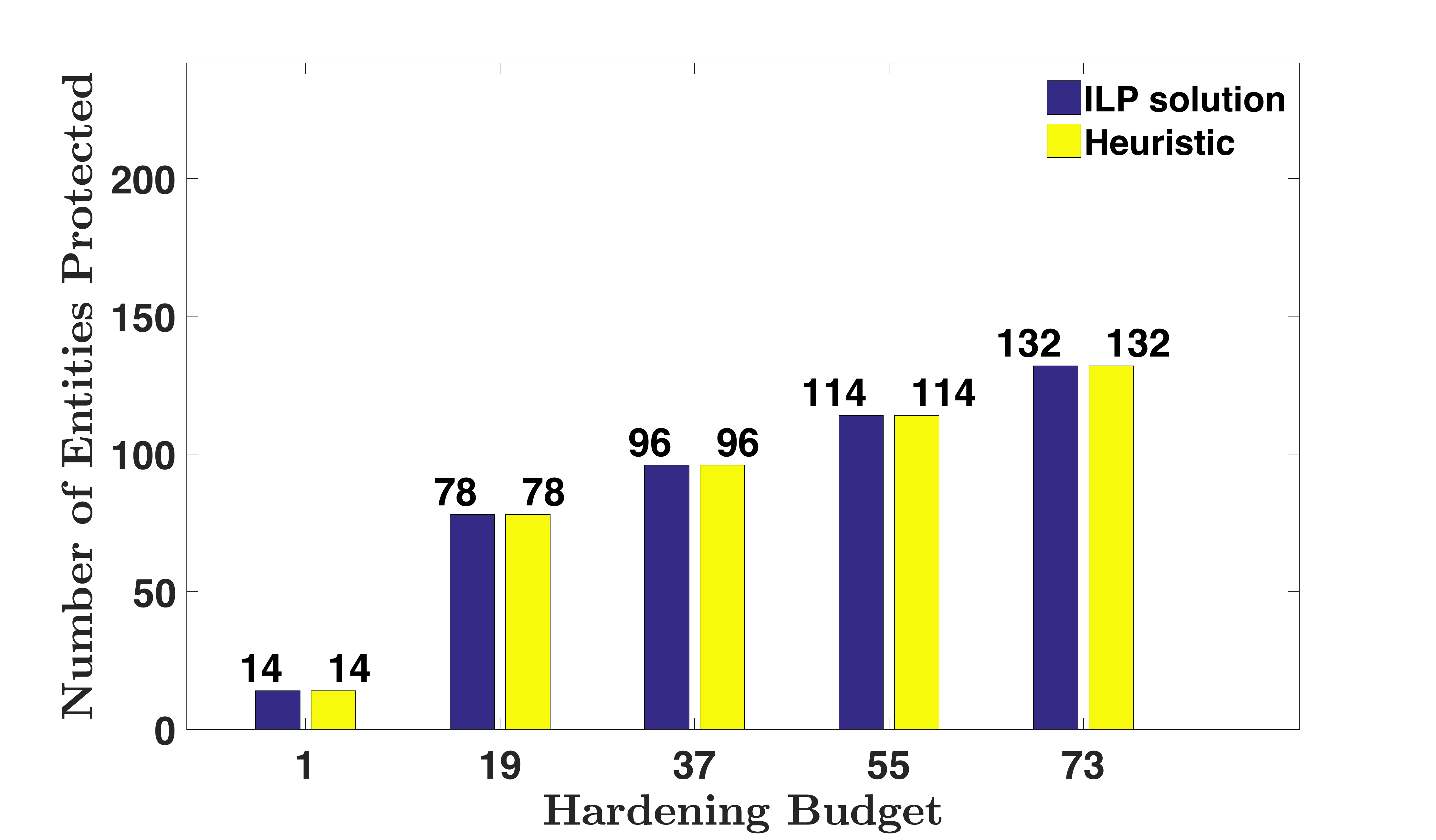}
  \captionsetup{justification=centering}
  \captionof{figure}{Comparison of ILP solution with Heuristic for 118 bus system (ENH)}
  \label{fig:case118}
\end{minipage}
\end{minipage}

\noindent\begin{minipage}{\textwidth}
\begin{minipage}[b]{0.3\linewidth}
\centering
\includegraphics[width=5cm,height=5cm,keepaspectratio]{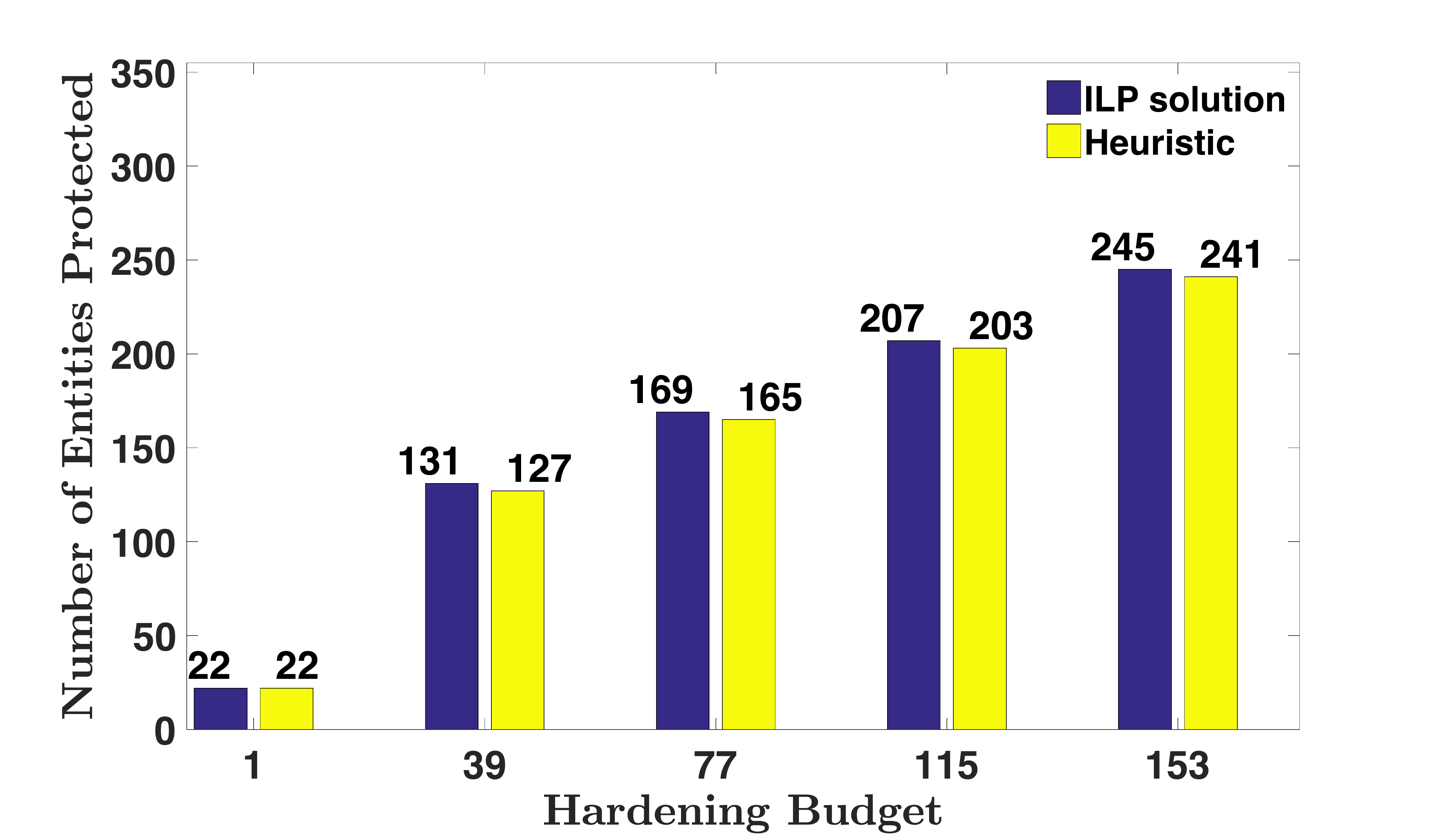}
  \captionsetup{justification=centering}
  \captionof{figure}{Comparison of ILP solution with Heuristic for 145 bus system (ENH)}
  \label{fig:case145}
\end{minipage} \hfil{}
\begin{minipage}[b]{0.3\linewidth}
\centering
\includegraphics[width=5cm,height=5cm,keepaspectratio]{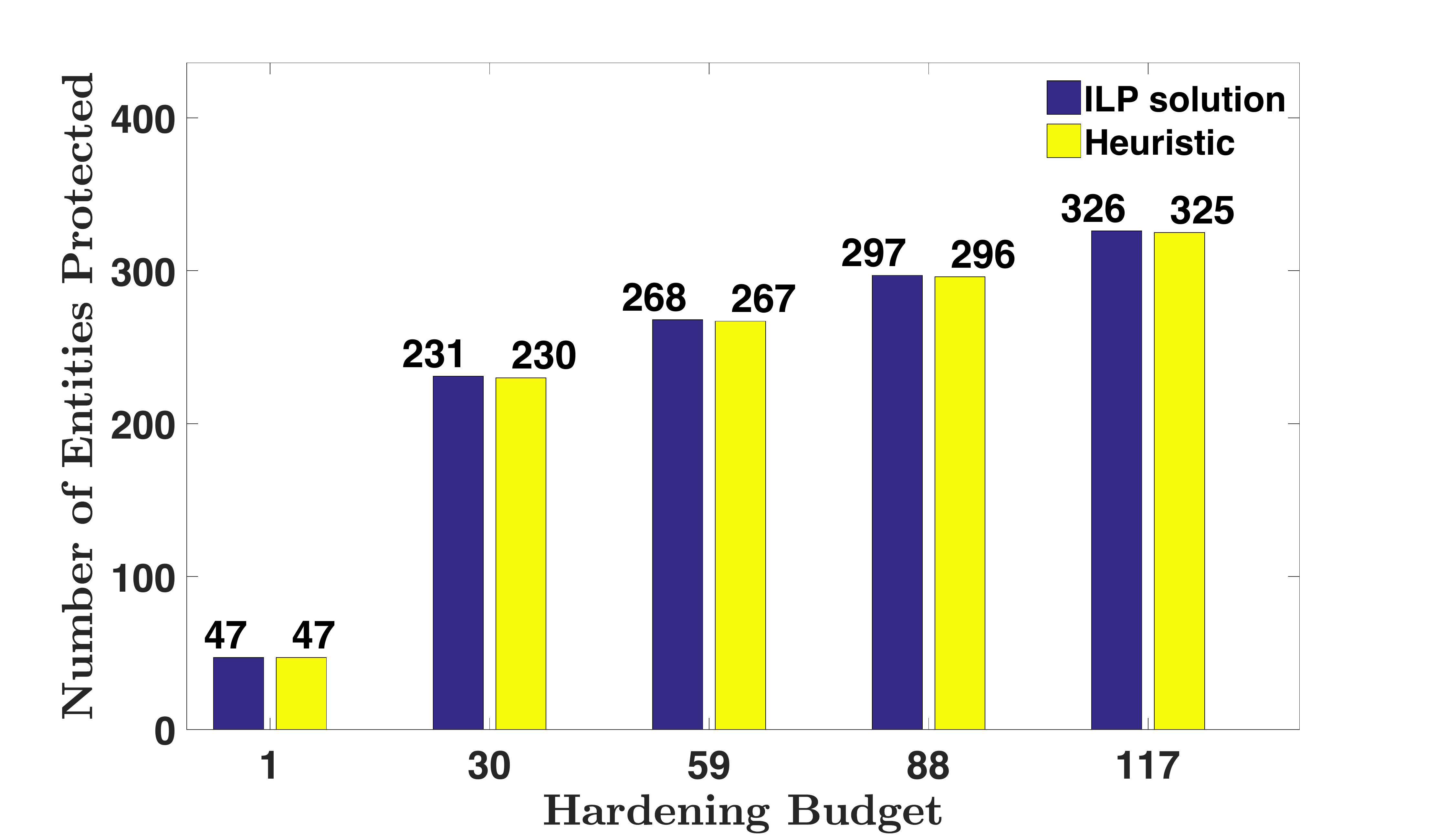}
  \captionsetup{justification=centering}
  \captionof{figure}{Comparison of ILP solution with Heuristic for 300 bus system (ENH)}
  \label{fig:case300}
\end{minipage}
\hfil{}
\begin{minipage}[b]{0.3\linewidth}
\centering
\includegraphics[width=5cm,height=5cm,keepaspectratio]{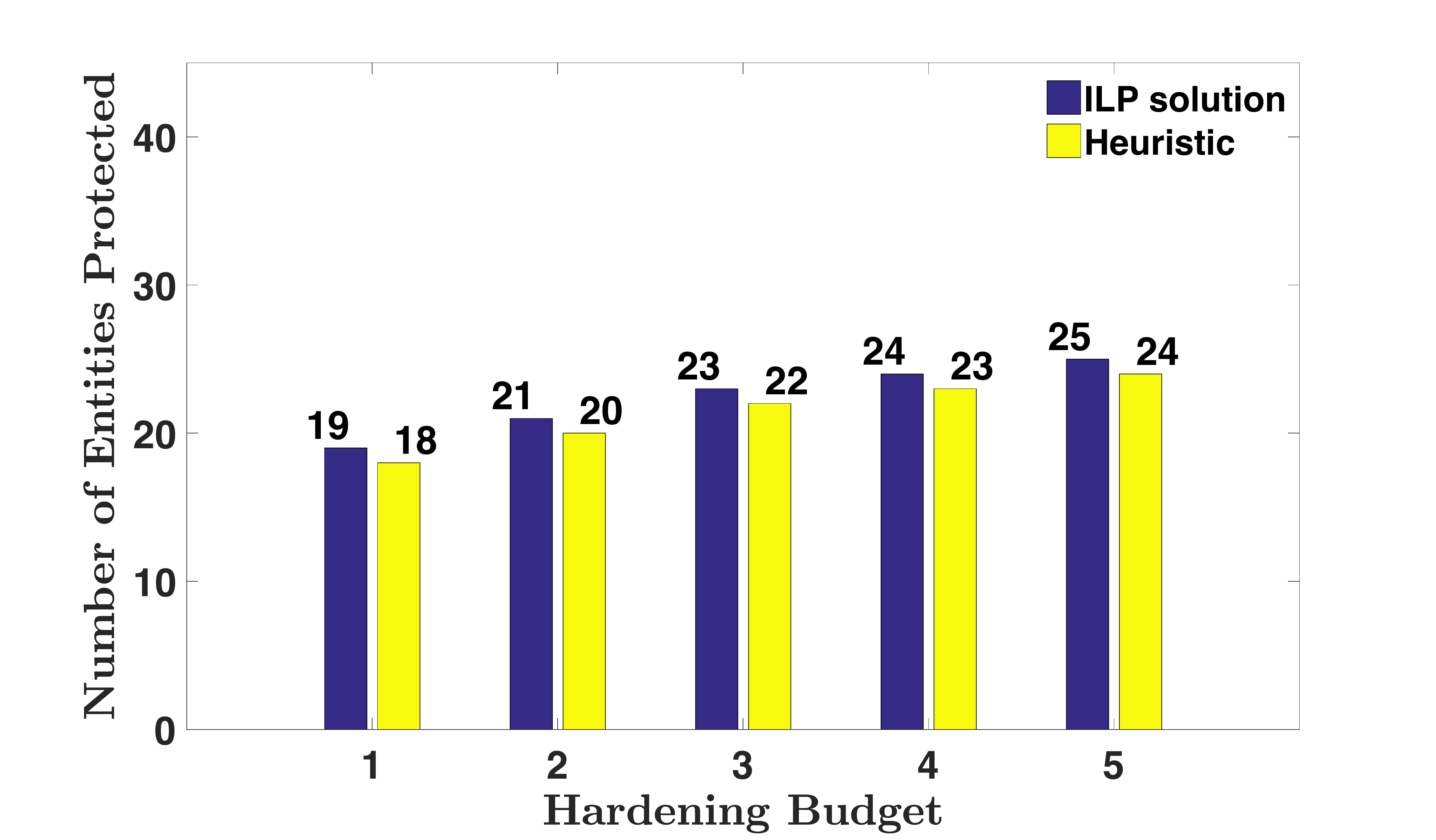}
  \captionsetup{justification=centering}
  \captionof{figure}{Comparison of ILP solution with Heuristic for Region 1 (ENH)}
  \label{fig:D1}
\end{minipage}
\end{minipage}

\noindent\begin{minipage}{\textwidth}
\begin{minipage}[b]{0.3\linewidth}
\centering
\includegraphics[width=5cm,height=5cm,keepaspectratio]{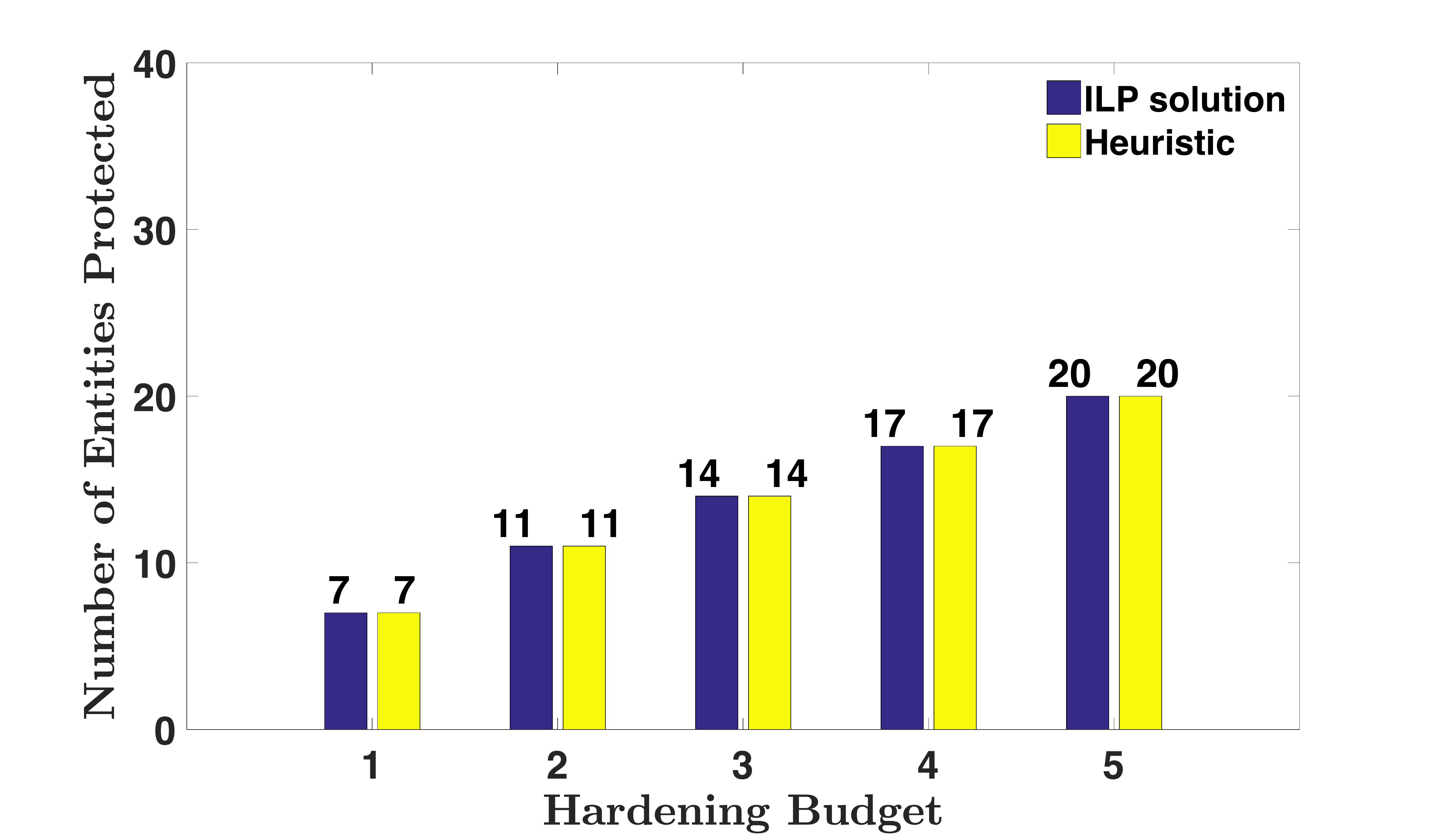}
  \captionsetup{justification=centering}
  \captionof{figure}{Comparison of ILP solution with Heuristic for Region 2 (ENH)}
  \label{fig:D2}
\end{minipage} \hfil{}
\begin{minipage}[b]{0.3\linewidth}
\centering
\includegraphics[width=5cm,height=5cm,keepaspectratio]{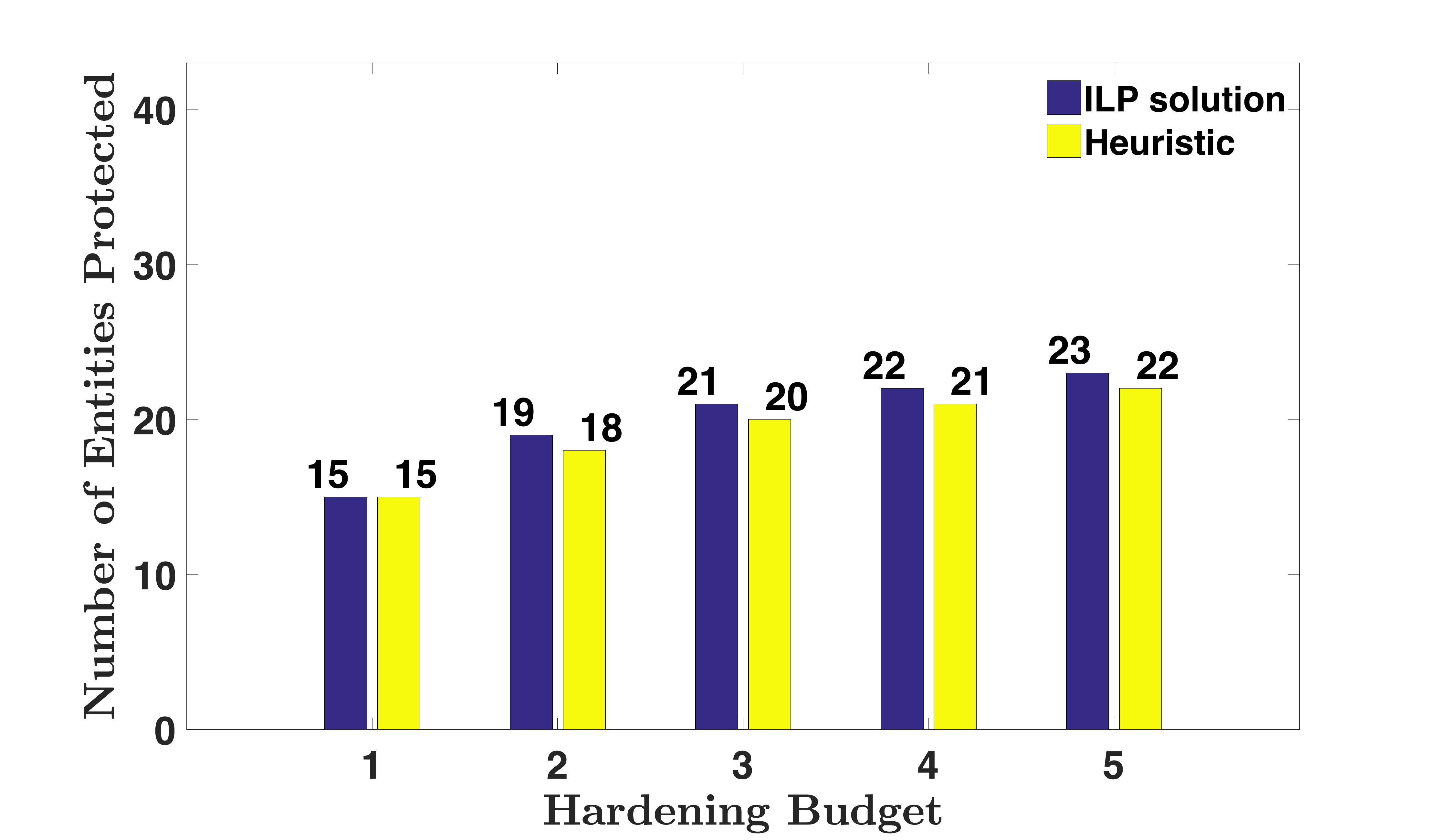}
  \captionsetup{justification=centering}
  \captionof{figure}{Comparison of ILP solution with Heuristic for Region 3 (ENH)}
  \label{fig:D3}
\end{minipage}
\hfil{}
\begin{minipage}[b]{0.3\linewidth}
\centering
\includegraphics[width=5cm,height=5cm,keepaspectratio]{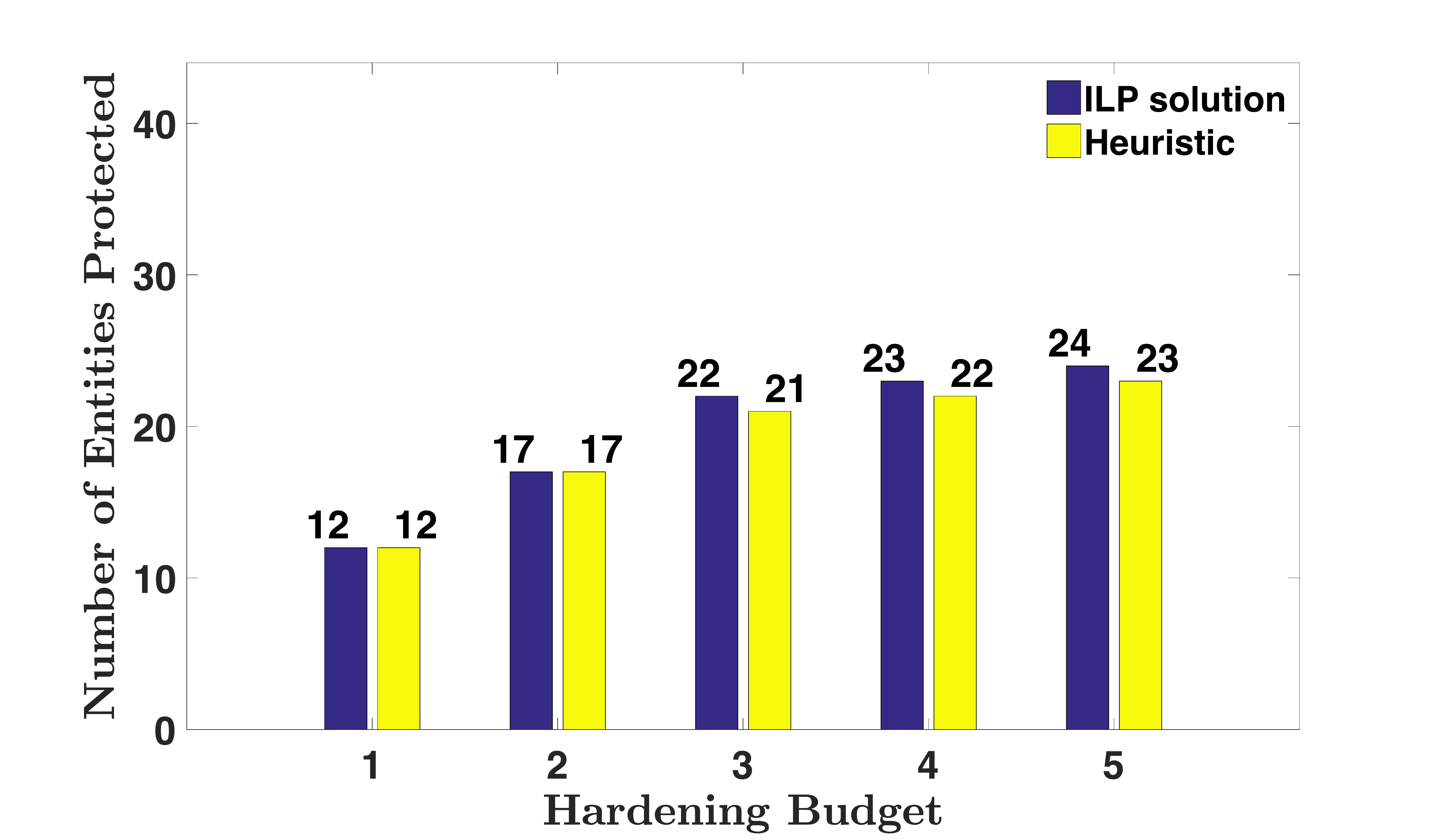}
  \captionsetup{justification=centering}
  \captionof{figure}{Comparison of ILP solution with Heuristic for Region 4 (ENH)}
  \label{fig:D4}
\end{minipage}
\end{minipage}

\begin{table}[ht]
    \centering
		\begin{tabular}{|c|c|c||c|c||c|c||c|c||c|c|}  \hline
			\multicolumn{1}{|c|}{} & \multicolumn{10}{c|}{\bf{Running time (in sec)}}\\ \cline{2-11}
			\multicolumn{1}{|c}{\bf{DataSet}} & \multicolumn{2}{|c||}{\bf{H1}} & \multicolumn{2}{|c||}{\bf{H2}} & \multicolumn{2}{|c||}{\bf{H3}} & \multicolumn{2}{|c||}{\bf{H4}} & \multicolumn{2}{|c|}{\bf{H5}}\\ \cline{2-11}
			\cline{2-11} & \bf{ILP} & \bf{Heu} & \bf{ILP} & \bf{Heu} & \bf{ILP} & \bf{Heu} & \bf{ILP} & \bf{Heu} & \bf{ILP} & \bf{Heu} \\ \hline
			\bf{24 bus} & $0.45$ & $0.01$ & $0.25$ & $0.01$ & $0.72$ & $0.01$ & $0.23$ & $0.01$ & $0.21$ & $0.01$\\ \hline
			\bf{30 bus} & $2.44$ & $0.01$ & $0.38$ & $0.01$ & $0.38$ & $0.01$ & $0.35$ & $0.01$ & $0.34$ & $0.01$\\ \hline
		    \bf{39 bus} & $0.80$ & $0.01$ & $0.50$ & $0.01$ & $0.49$ & $0.01$ & $0.48$ & $0.01$ & $0.47$ & $0.01$\\ \hline
		    \bf{57 bus} & $2.67$ & $0.03$ & $1.73$ & $0.01$ & $2.21$ & $0.01$ & $2.27$ & $0.01$ & $1.68$ & $0.01$\\ \hline
			\bf{89 bus} & $23.2$ & $0.05$ & $14.6$ & $0.03$ & $14.6$ & $0.03$ & $14.5$ & $0.03$ & $14.7$ & $0.75$\\ \hline	
			\bf{118 bus} & $20.9$ & $0.04$ & $16.2$ & $0.06$ & $17.2$ & $0.09$ & $17.1$ & $0.02$ & $17.1$ & $0.02$\\ \hline
			\bf{145 bus} & $85.2$ & $0.05$ & $71.0$ & $0.10$ & $71.3$ & $0.18$ & $68.4$ & $0.06$ & $78.3$ & $0.07$\\ \hline
			\bf{300 bus} & $282$ & $0.15$ & $222$ & $1.56$ & $217$ & $0.85$ & $253$ & $0.39$ & $264$ & $0.40$\\ \hline
			\bf{Region 1} & $0.53$ & $0.01$ & $0.36$ & $0.01$ & $0.34$ & $0.01$ & $0.36$ & $0.01$ & $0.36$ & $0.01$\\ \hline
			\bf{Region 2} & $13.8$ & $0.01$ & $12.9$ & $0.01$ & $12.8$ & $0.01$ & $13.1$ & $0.01$ & $13.2$ & $0.01$\\ \hline
			\bf{Region 3} & $1.92$ & $0.01$ & $1.36$ & $0.01$ & $1.29$ & $0.01$ & $1.31$ & $0.01$ & $1.44$ & $0.01$\\ \hline
			\bf{Region 4} & $1.48$ & $0.01$ & $1.43$ & $0.01$ & $1.10$ & $0.01$ & $1.06$ & $0.01$ & $1.05$ & $0.01$\\ \hline
		\end{tabular}
		\captionsetup{justification=centering}
		\caption{Run time comparison of Integer Linear Program and Heuristic for different Data Sets (ENH)}
		\protect\label{tbl:RuntimeH}
\end{table}

\noindent\begin{minipage}{\textwidth}
\begin{minipage}[b]{0.3\linewidth}
\centering
\includegraphics[width=5cm,height=5cm,keepaspectratio]{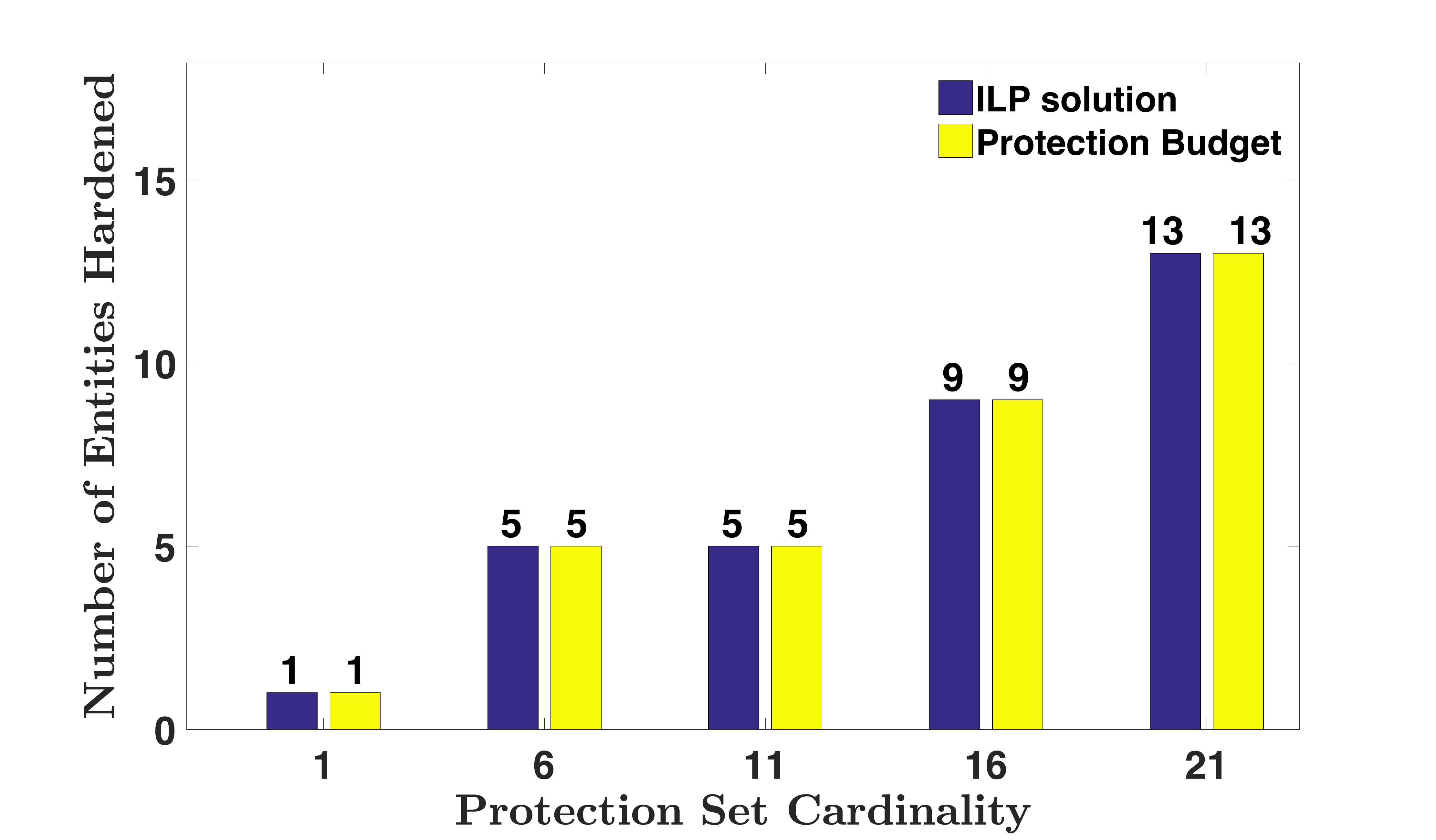}
  \captionsetup{justification=centering}
  \captionof{figure}{Comparison of ILP solution with Heuristic for 24 bus system (TEH)}
  \label{fig:case24T}
\end{minipage} \hfil{}
\begin{minipage}[b]{0.3\linewidth}
\centering
\includegraphics[width=5cm,height=5cm,keepaspectratio]{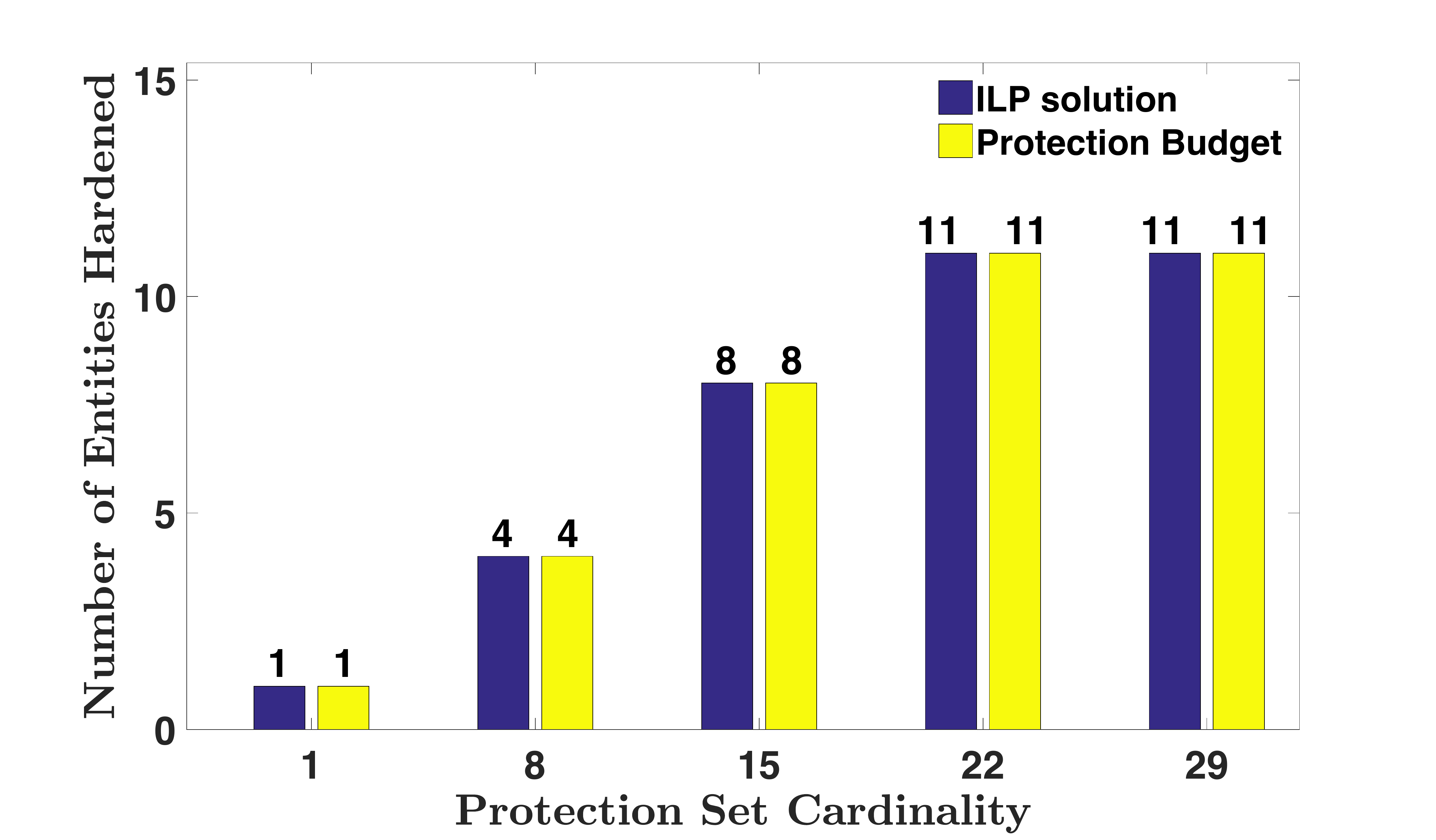}
  \captionsetup{justification=centering}
  \captionof{figure}{Comparison of ILP solution with Heuristic for 30 bus system (TEH)}
  \label{fig:case30T}
\end{minipage}
\hfil{}
\begin{minipage}[b]{0.3\linewidth}
\centering
\includegraphics[width=5cm,height=5cm,keepaspectratio]{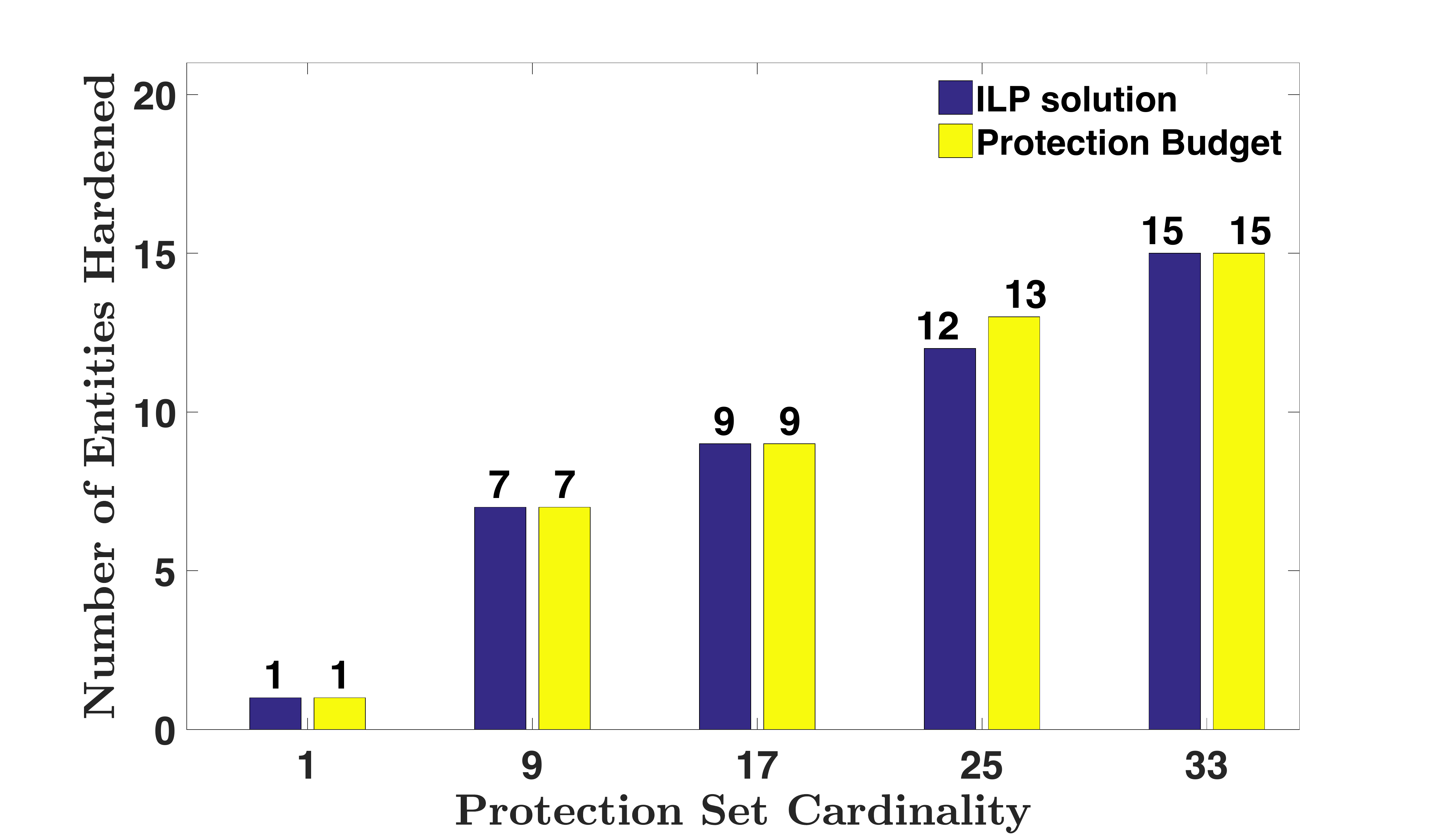}
  \captionsetup{justification=centering}
  \captionof{figure}{Comparison of ILP solution with Heuristic for 39 bus system (TEH)}
  \label{fig:case39T}
\end{minipage}
\end{minipage}

\noindent\begin{minipage}{\textwidth}
\begin{minipage}[b]{0.3\linewidth}
\centering
\includegraphics[width=5cm,height=5cm,keepaspectratio]{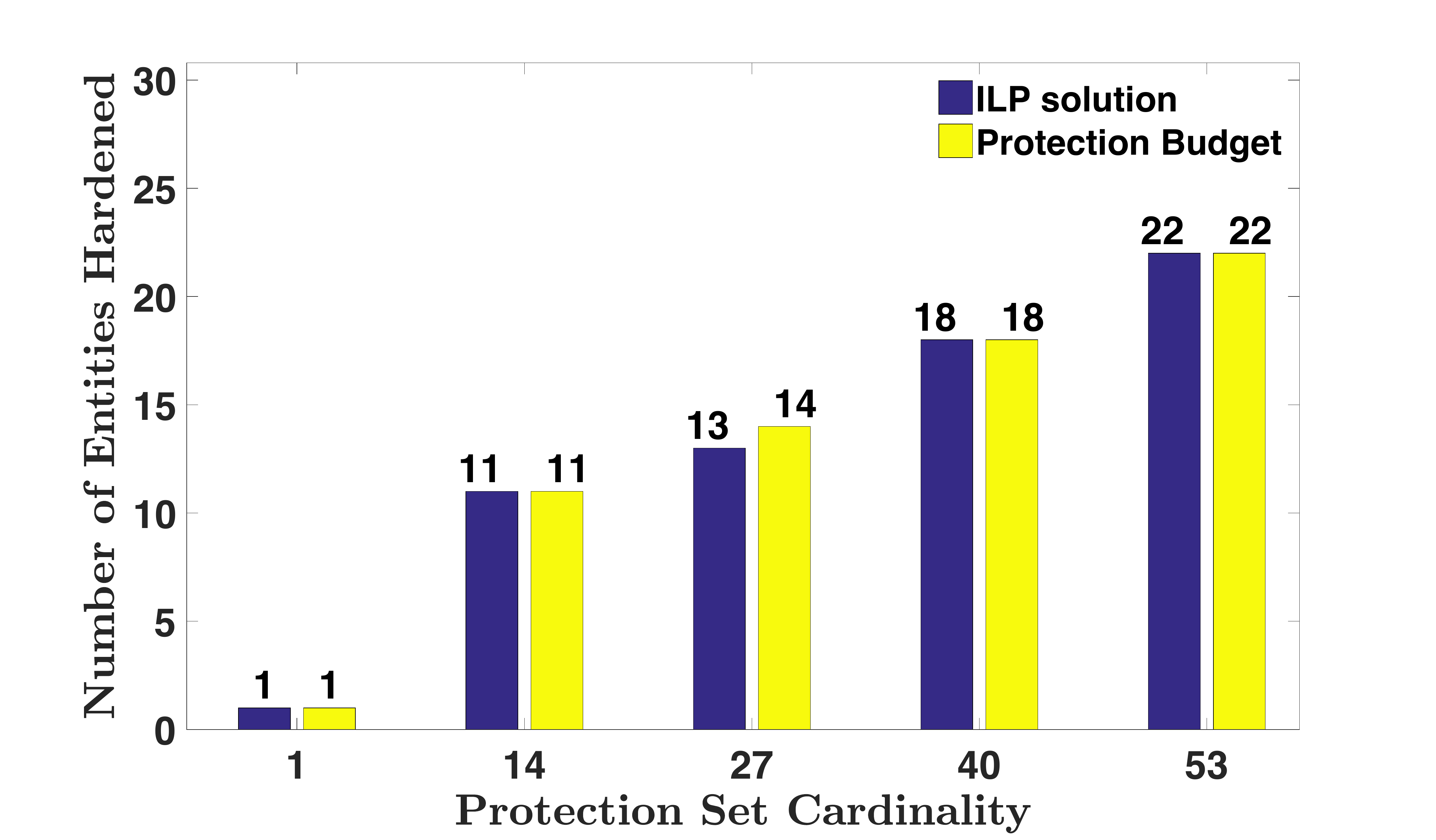}
  \captionsetup{justification=centering}
  \captionof{figure}{Comparison of ILP solution with Heuristic for 57 bus system (TEH)}
  \label{fig:case57T}
\end{minipage} \hfil{}
\begin{minipage}[b]{0.3\linewidth}
\centering
\includegraphics[width=5cm,height=5cm,keepaspectratio]{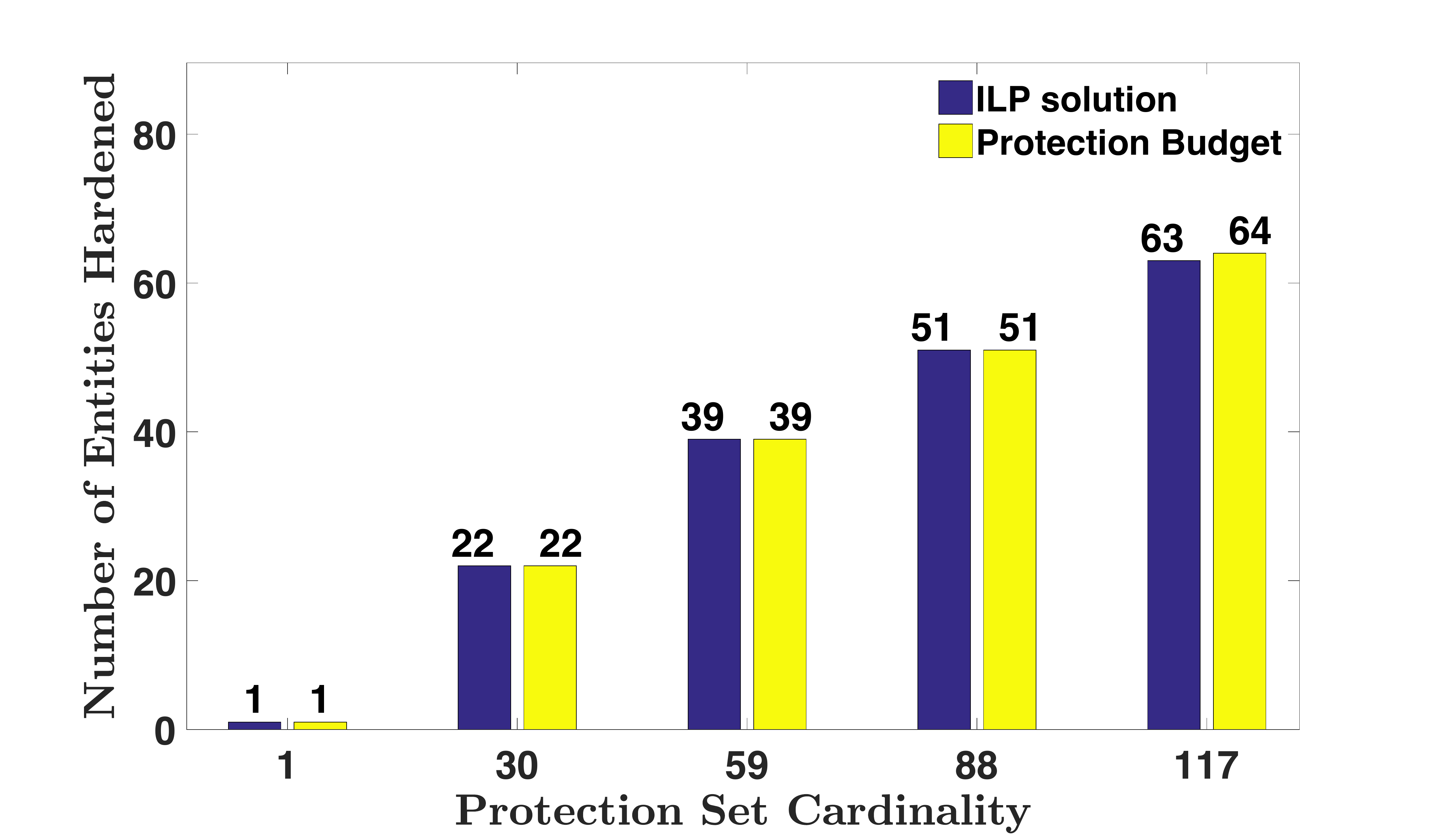}
  \captionsetup{justification=centering}
  \captionof{figure}{Comparison of ILP solution with Heuristic for 89 bus system (TEH)}
  \label{fig:case89T}
\end{minipage}
\hfil{}
\begin{minipage}[b]{0.3\linewidth}
\centering
\includegraphics[width=5cm,height=5cm,keepaspectratio]{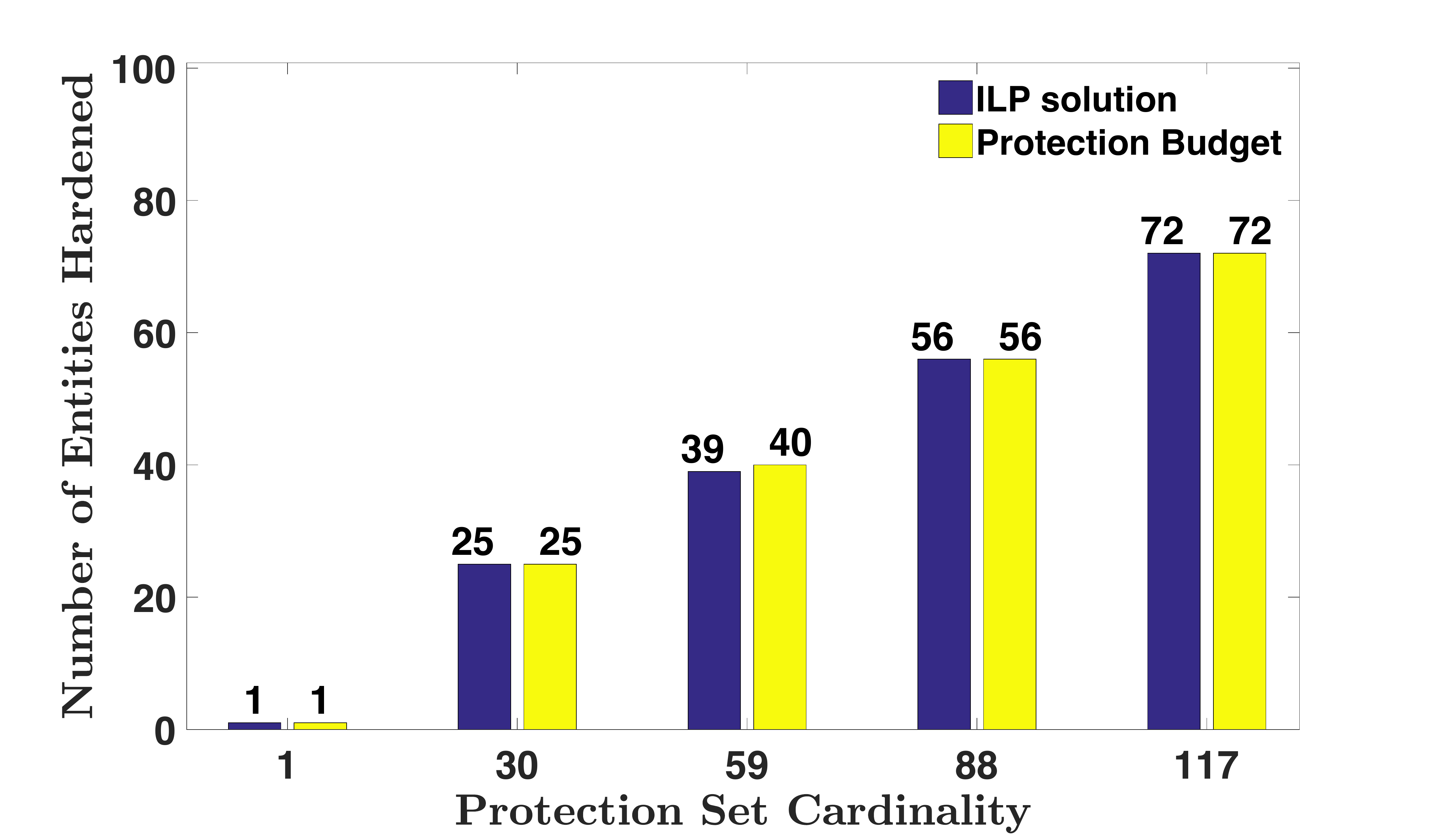}
  \captionsetup{justification=centering}
  \captionof{figure}{Comparison of ILP solution with Heuristic for 118 bus system (TEH)}
  \label{fig:case118T}
\end{minipage}
\end{minipage}

\noindent\begin{minipage}{\textwidth}
\begin{minipage}[b]{0.3\linewidth}
\centering
\includegraphics[width=5cm,height=5cm,keepaspectratio]{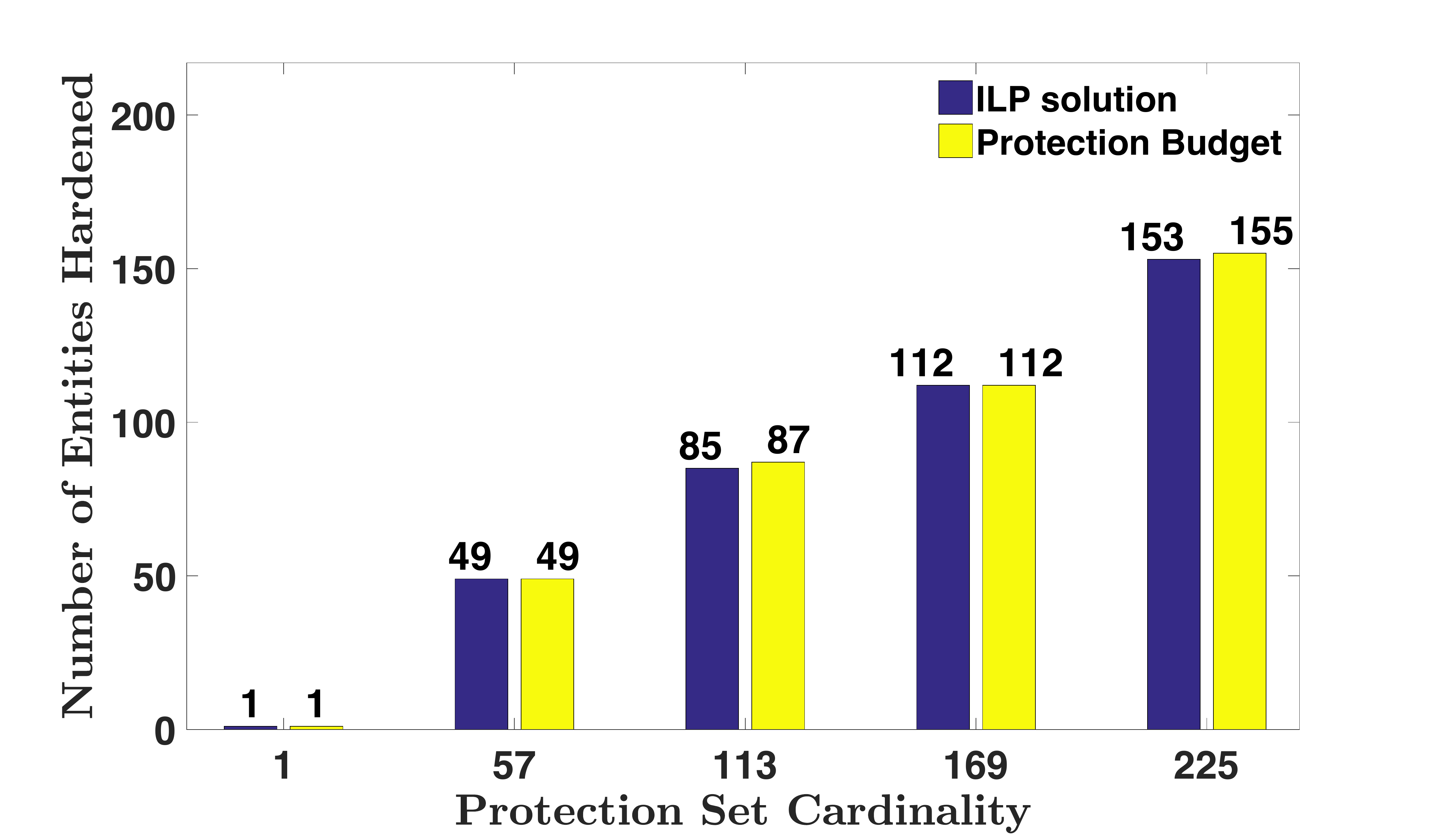}
  \captionsetup{justification=centering}
  \captionof{figure}{Comparison of ILP solution with Heuristic for 145 bus system (TEH)}
  \label{fig:case145T}
\end{minipage} \hfil{}
\begin{minipage}[b]{0.3\linewidth}
\centering
\includegraphics[width=5cm,height=5cm,keepaspectratio]{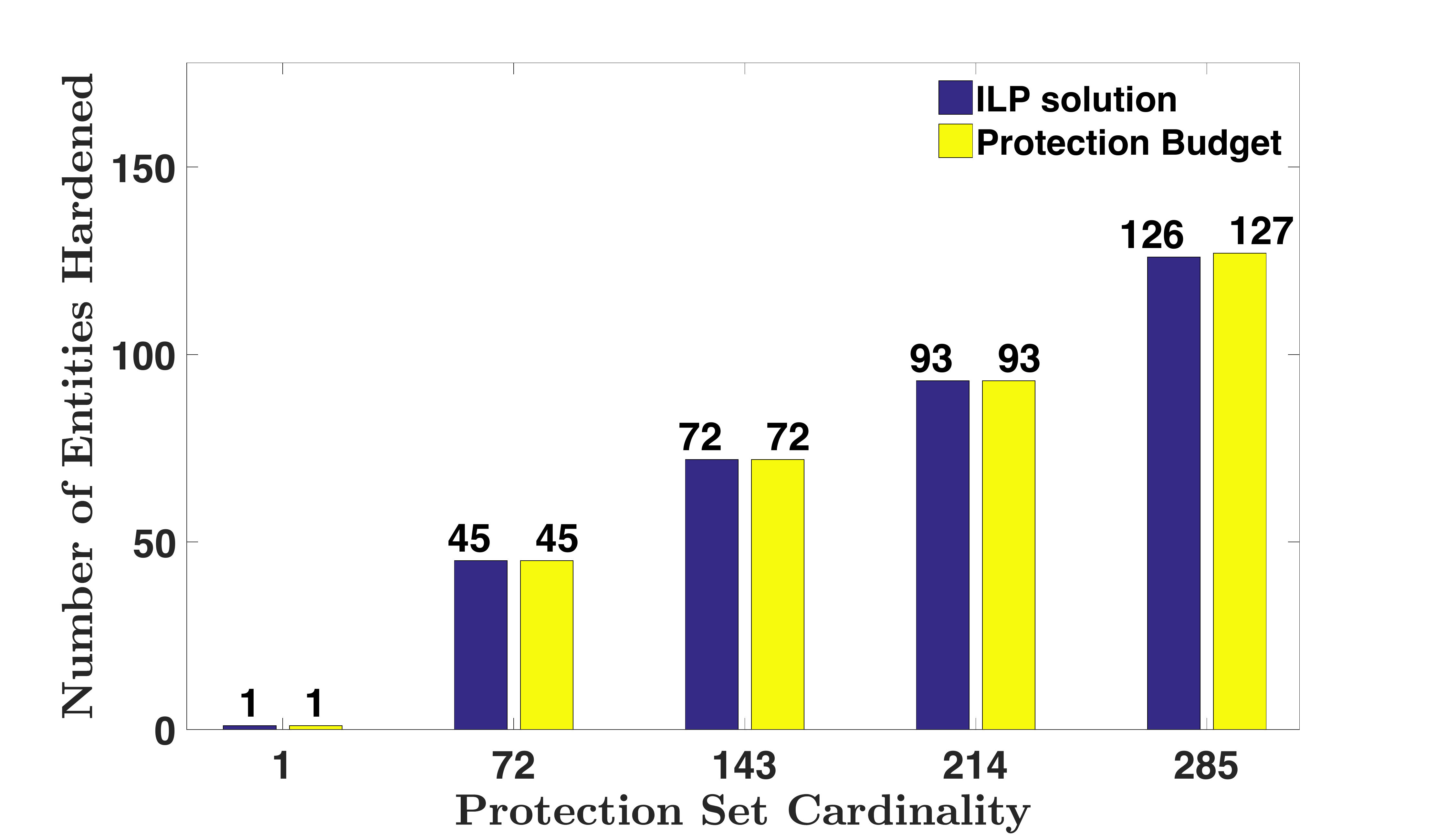}
  \captionsetup{justification=centering}
  \captionof{figure}{Comparison of ILP solution with Heuristic for 300 bus system (TEH)}
  \label{fig:case300T}
\end{minipage}
\hfil{}
\begin{minipage}[b]{0.3\linewidth}
\centering
\includegraphics[width=5cm,height=5cm,keepaspectratio]{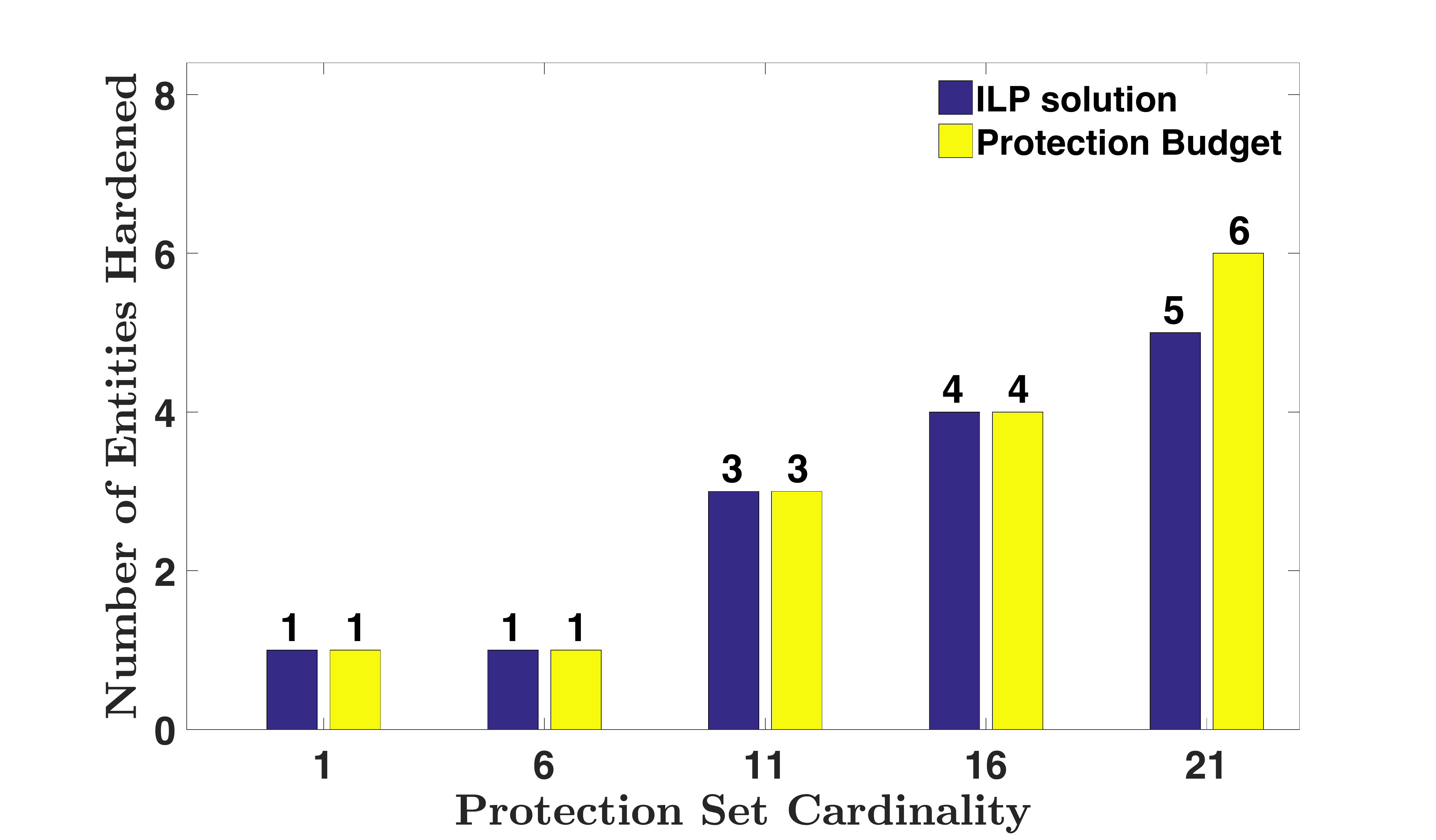}
  \captionsetup{justification=centering}
  \captionof{figure}{Comparison of ILP solution with Heuristic for Region 1 (TEH)}
  \label{fig:D1T}
\end{minipage}
\end{minipage}

\noindent\begin{minipage}{\textwidth}
\begin{minipage}[b]{0.3\linewidth}
\centering
\includegraphics[width=5cm,height=5cm,keepaspectratio]{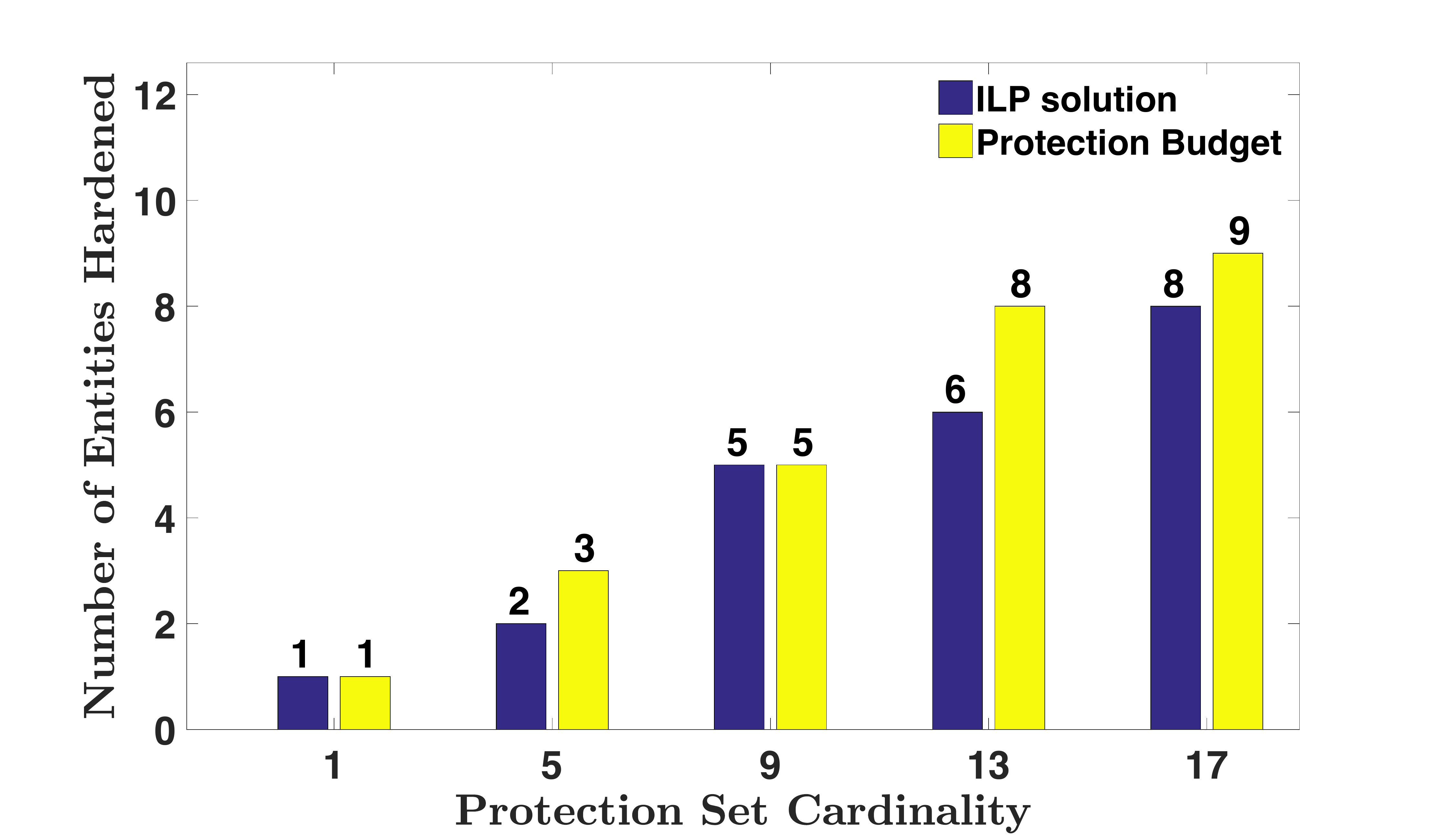}
  \captionsetup{justification=centering}
  \captionof{figure}{Comparison of ILP solution with Heuristic for Region 2 (TEH)}
  \label{fig:D2T}
\end{minipage} \hfil{}
\begin{minipage}[b]{0.3\linewidth}
\centering
\includegraphics[width=5cm,height=5cm,keepaspectratio]{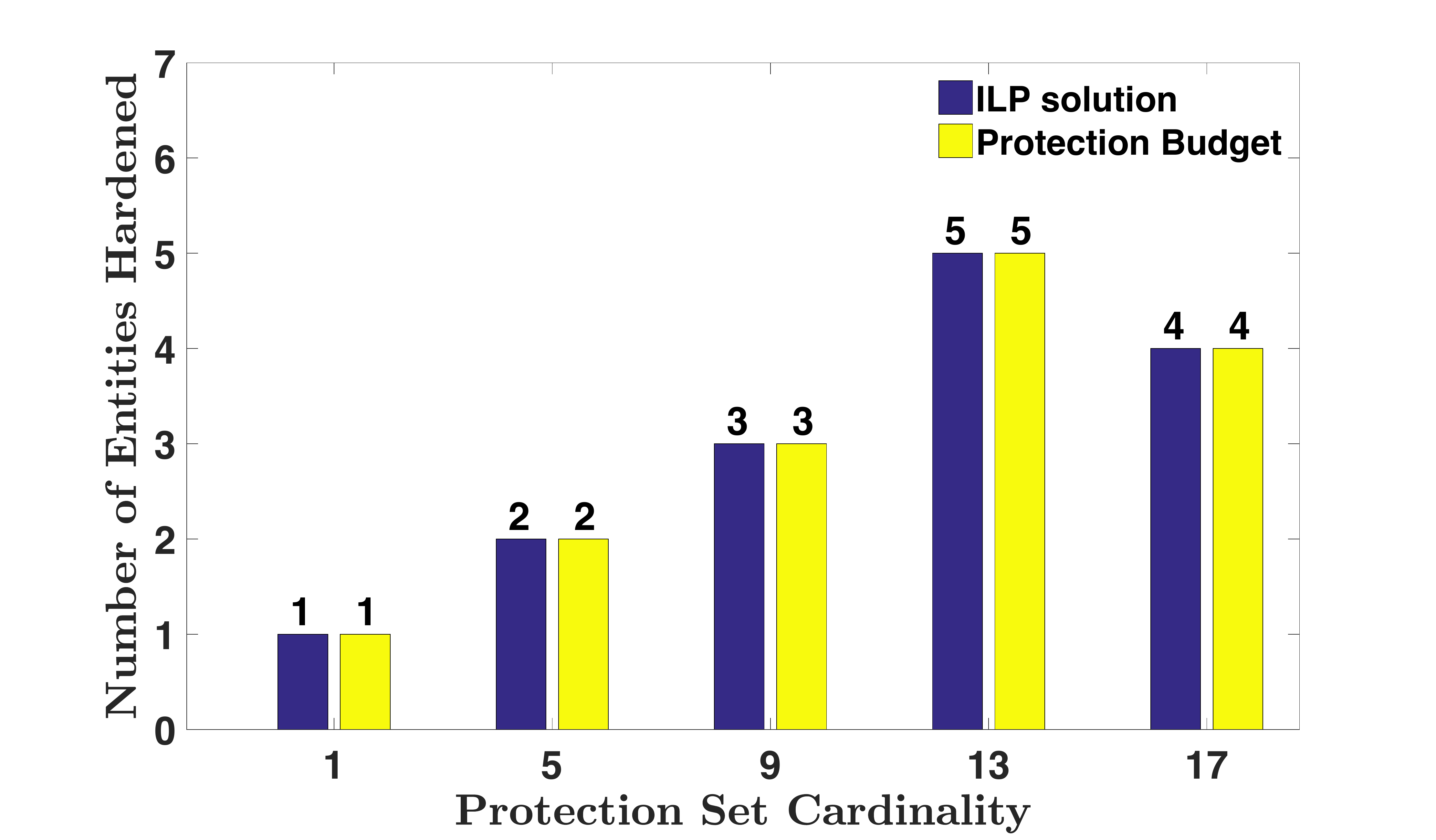}
  \captionsetup{justification=centering}
  \captionof{figure}{Comparison of ILP solution with Heuristic for Region 3 (TEH)}
  \label{fig:D3T}
\end{minipage}
\hfil{}
\begin{minipage}[b]{0.3\linewidth}
\centering
\includegraphics[width=5cm,height=5cm,keepaspectratio]{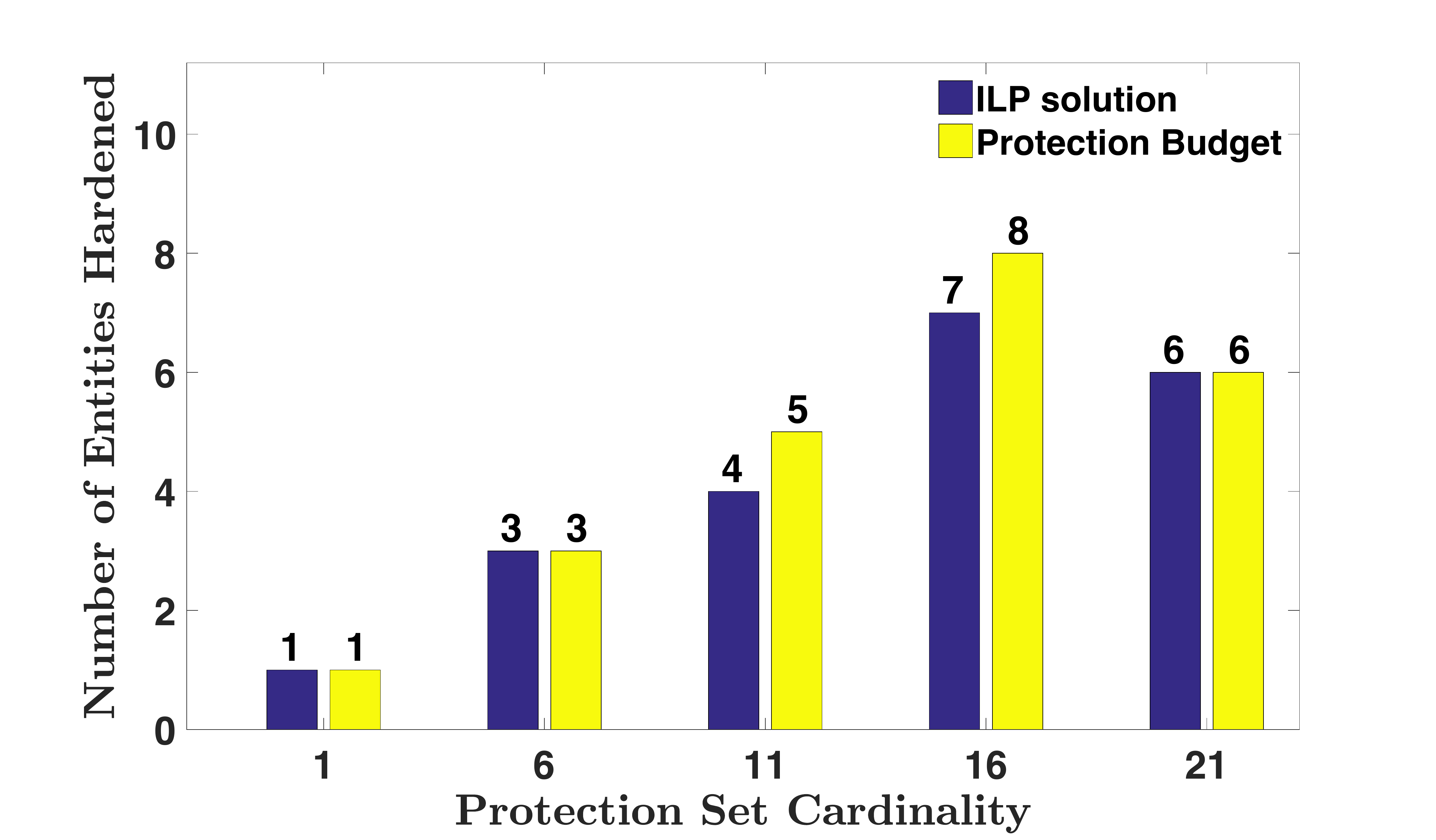}
  \captionsetup{justification=centering}
  \captionof{figure}{Comparison of ILP solution with Heuristic for Region 4 (TEH)}
  \label{fig:D4T}
\end{minipage}
\end{minipage}

\begin{table}[ht]
    \centering
		\begin{tabular}{|c|c|c||c|c||c|c||c|c||c|c|}  \hline
			\multicolumn{1}{|c|}{} & \multicolumn{10}{c|}{\bf{Running time (in sec)}}\\ \cline{2-11}
			\multicolumn{1}{|c}{\bf{DataSet}} & \multicolumn{2}{|c||}{\bf{P1}} & \multicolumn{2}{|c||}{\bf{P2}} & \multicolumn{2}{|c||}{\bf{P3}} & \multicolumn{2}{|c||}{\bf{P4}} & \multicolumn{2}{|c|}{\bf{P5}}\\ \cline{2-11}
			\cline{2-11} & \bf{ILP} & \bf{Heu} & \bf{ILP} & \bf{Heu} & \bf{ILP} & \bf{Heu} & \bf{ILP} & \bf{Heu} & \bf{ILP} & \bf{Heu} \\ \hline
			\bf{24 bus} & $0.42$ & $0.01$ & $0.22$ & $0.01$ & $0.20$ & $0.01$ & $0.19$ & $0.01$ & $0.19$ & $0.01$\\ \hline
			\bf{30 bus} & $0.57$ & $0.01$ & $0.37$ & $0.01$ & $0.34$ & $0.01$ & $0.34$ & $0.01$ & $0.31$ & $0.01$\\ \hline
		    \bf{39 bus} & $0.82$ & $0.01$ & $0.49$ & $0.01$ & $0.49$ & $0.01$ & $0.46$ & $0.01$ & $0.47$ & $0.01$\\ \hline
		    \bf{57 bus} & $2.23$ & $0.02$ & $1.69$ & $0.02$ & $2.00$ & $0.02$ & $2.07$ & $0.02$ & $1.91$ & $0.01$\\ \hline
			\bf{89 bus} & $17.3$ & $0.05$ & $14.4$ & $0.08$ & $14.4$ & $0.16$ & $14.0$ & $0.11$ & $14.1$ & $0.07$\\ \hline	
			\bf{118 bus} & $17.6$ & $0.13$ & $17.3$ & $0.15$ & $16.9$ & $0.10$ & $16.3$ & $0.08$ & $17.0$ & $0.07$\\ \hline
			\bf{145 bus} & $79.0$ & $0.06$ & $76.6$ & $0.26$ & $77.2$ & $0.27$ & $75.7$ & $0.24$ & $74.9$ & $0.25$\\ \hline
			\bf{300 bus} & $302$ & $0.18$ & $241$ & $1.01$ & $234$ & $1.41$ & $229$ & $1.03$ & $230$ & $1.20$\\ \hline
			\bf{Region 1} & $0.55$ & $0.01$ & $0.40$ & $0.01$ & $0.43$ & $0.01$ & $0.35$ & $0.01$ & $0.34$ & $0.01$\\ \hline
			\bf{Region 2} & $14.5$ & $0.01$ & $13.5$ & $0.01$ & $13.4$ & $0.01$ & $13.4$ & $0.01$ & $13.3$ & $0.01$\\ \hline
			\bf{Region 3} & $1.58$ & $0.01$ & $1.40$ & $0.01$ & $1.29$ & $0.01$ & $1.29$ & $0.01$ & $1.29$ & $0.01$\\ \hline
			\bf{Region 4} & $1.36$ & $0.01$ & $1.12$ & $0.01$ & $1.21$ & $0.01$ & $1.09$ & $0.01$ & $1.03$ & $0.01$\\ \hline
		\end{tabular}
		\captionsetup{justification=centering}
		\caption{Run time comparison of Integer Linear Program and Heuristic for different Data Sets (TEH)}
		\protect\label{tbl:RuntimeT}
\end{table}

\section{Conclusion}
\label{Conclusion}
In this paper, we have studied the Entity Hardening problem and the Targeted Entity Hardening problem in Critical Infrastructure network(s). We have used the IIM model to capture the interdependencies  and dependencies that exist between power-communication network and power network in isolation respectively. Using such a model, we have formulated the Entity Hardening and the Targeted Entity Hardening problems. Both problems are proved to be NP-Complete. For both the problems, the optimal solution, obtained from the ILP, is compared with the developed heuristic solution using dependency equations generated from power-communication network data of the Maricopa County, Arizona  and power network data of different bus systems obtained from MatPower. As noted in the Experimental Analysis for both the problems the performance of the heuristics are comparable to that of ILP and solutions are produced in much lesser time.

\end{document}